%% file: gmc.tex
\documentclass{lmcs} 
\pdfoutput=1
\usepackage[utf8]{inputenc}

\usepackage{lastpage}
\lmcsdoi{18}{4}{7}
\lmcsheading{}{\pageref{LastPage}}{}{}%
{Feb.~08,~2022}{Nov.~22,~2022}{}


\input{macros.tex}
\usepackage{booktabs}

\usepackage{fancyvrb}

\begin{document}

\title{Geometric Model Checking of Continuous Space}
\titlecomment{The research leading to this publication was partially supported
by the MUR Projects PRIN 2017FTXR7S, “IT- MaTTerS”, PRIN 2020TL3X8X
“T-LADIES”, and PNRR PRI ECS00000017 “THE - Tuscany Health Ecosystem”.
The authors are listed in alphabetical order, as they equally contributed to this work.}

\author[N.~Bezhanishvili]{Nick Bezhanishvili}[a]
\address{Institute for Logic, Language and Computation, University of Amsterdam, The Netherlands}	%
\email{n.bezhanishvili@uva.nl}  %

\author[V.~Ciancia]{Vincenzo Ciancia\lmcsorcid{0000-0003-1314-0574}}[b]	%
\address{Istituto di Scienza e Tecnologie dell'Informazione ``A. Faedo'', Consiglio Nazionale delle Ricerche, Pisa, Italy}	%
\email{\{vincenzo.ciancia,diego.latella,mieke.massink\}@cnr.it}  %

\author[D.~Gabelaia]{David Gabelaia\lmcsorcid{0000-0002-8317-7949}}[c]	%
\address{TSU Razmadze Mathematical Institute, Tbilisi, Georgia}	%
\email{gabelaia@gmail.com} %

\author[G.~Grilletti]{Gianluca Grilletti\lmcsorcid{0000-0002-1631-3648}}[d]	%
\address{Munich Center for Mathematical Philosophy, Ludwig-Maximilians-Universit\"{a}t M\"{u}nchen, Germany}	%
\email{grilletti.gianluca@gmail.com} %

\author[D.~Latella]{Diego Latella\lmcsorcid{0000-0002-3257-9059}}[b]	%

\author[M.~Massink]{Mieke Massink\lmcsorcid{0000-0001-5089-002X}}[b]	%

\begin{abstract}
	Topological Spatial Model Checking is a recent paradigm where model checking techniques are developed for the topological interpretation of Modal Logic. The Spatial Logic of Closure Spaces, \SLCS, extends Modal Logic with reachability connectives that, in turn, can be used for expressing interesting spatial properties, such as ``being near to'' or ``being surrounded by''.
	\SLCS constitutes the kernel of a solid logical framework for reasoning about \emph{discrete} space, such as graphs and digital images, interpreted as quasi discrete closure spaces. 
	Following a recently developed \emph{geometric} semantics of Modal Logic, we propose an interpretation of \SLCS in \emph{continuous space}, admitting a geometric spatial model checking procedure, by resorting to models based on polyhedra. Such representations of space are increasingly relevant in many domains of application, due to recent developments of 3D scanning and visualisation techniques that exploit mesh processing. We introduce \PolyLogicA, a geometric spatial model checker for \SLCS formulas on polyhedra and demonstrate feasibility of our approach on two 3D polyhedral models of realistic size. %
Finally, we introduce a geometric definition of bisimilarity, proving that it characterises logical equivalence.
\end{abstract}

\maketitle

\input{introduction.tex}

\input{background-simplicial-complexes.tex}

\input{slcs-on-polyhedra.tex}

\input{model-checking.tex}

\input{tool.tex}

\input{simplicial-bisimilarity.tex}

\input{conclusions.tex}

\bibliographystyle{alphaurl}
\bibliography{gmc.bib}

\appendix

\input{appendix.tex}

\end{document}

%% file: macros.tex
\newif\iflong
\longfalse
\newcommand{\slreach}{\gamma}
\usepackage[caption=false]{subfig}

\usepackage{listings}

\usepackage{amsmath}
\usepackage{amssymb}
\usepackage{stmaryrd}
\SetSymbolFont{stmry}{bold}{U}{stmry}{m}{n}
\usepackage{graphicx}
\usepackage[export]{adjustbox}
\usepackage{xcolor}
\usepackage{xspace}
\usepackage{tikz}
\usetikzlibrary{patterns}
\usetikzlibrary{decorations.markings}
\usepackage{tikz-layers}
\usepackage{apptools}

\newcommand{\SLCS}{SLCS\xspace}

\newcommand{\VoxLogicA}{VoxLogicA\xspace}
\newcommand{\topochecker}{Topochecker\xspace} 
\newcommand{\PolyLogicA}{PolyLogicA\xspace}
\newcommand{\PolyVisualiser}{PolyVisualiser\xspace}

\newcommand{\R}{\mathbb{R}}
\newcommand{\powset}{\mathcal{P}}
\newcommand{\card}[1]{\# #1}

\newcommand{\Cl}{\mathcal{C}}
\newcommand{\interior}{\mathcal{I}}
\newcommand{\AP}{\mathrm{AP}}

\newcommand{\tuple}[1]{\left\langle #1 \right\rangle}
\newcommand{\toPL}{\stackrel{PL}{\to}}

\newcommand{\face}{\preceq}
\newcommand{\facer}{\succeq}

\newcommand{\model}[1]{\mathcal{#1}}

\newcommand{\sem}[1]{\left\llbracket #1 \right\rrbracket^{\model{X}}}
\newcommand{\semKM}[1]{\left\llbracket #1 \right\rrbracket^{\KM}}

\newcommand{\Ksc}{\mathrm{K}}
\newcommand{\Lsc}{\mathrm{L}}
\newcommand{\poly}[1]{\left| #1 \right|}
\newcommand{\relint}[1]{\widetilde{#1}}

\newcommand{\KM}{\mathcal{M}}
\newcommand{\toPM}{\stackrel{\pm}{\to}}

\newcommand{\OUT}{\texttt{out}}
\newcommand{\IN}{\texttt{in}}

\newcommand{\vnote}[1]{{\color{blue}VC: #1}}
\newcommand{\lequiv}{\equiv}
\newcommand{\subpath}[2]{{#1}^{#2}}

\usepackage{latexsym}

\newcommand{\m}{{\mathcal X}}
\newcommand{\GL}{\gamma}
\newcommand{\sat}{\mbox{Sat}}

%% file: introduction.tex
\section{Introduction and Related Work}
\label{sec:intro}

Spatial reasoning and spatial properties are of interest in a vast number of domains, ranging from collective adaptive systems, concerned with the emergence of spatial patterns, and the mobility and distribution of devices in cyber-physical systems to medical imaging and interactive visualisation.
Recently, novel variants of model checking have been developed, moving the focus from checking \emph{temporal} properties to \emph{spatial} properties, see for example~\cite{Gr+09,CLLM14,CLLM16,CLLM16bertinoro,Ne+18,SPATEL,MBLSF21}, and, in fact, also to the combination of reasoning on time and space in spatio-temporal model checking~\cite{CGLLM14,CLMP15,CGLLM15,Ci+18,CLMPV16,Gri16,TsigkanosKG17}.

This so-called \emph{topological} approach to spatial logic and spatial model checking has its origin in the ideas by McKinsey and Tarski~\cite{McKT44}, who recognised the possibility of reasoning about space using topology as a mathematical framework for the interpretation of modal logic (see~\cite{vBB07} for a thorough introduction). The work by Ciancia et al. (see e.g.~\cite{CLLM14,CLLM16}) builds on these theoretical developments using \emph{Closure Spaces}, a generalisation of topological spaces encompassing also 
general \emph{discrete} spatial structures such as graphs~\cite{Gal99,Gal14}, as underlying model for the \emph{Spatial Logic for Closure Spaces} \SLCS.\@ The original version of this spatial logic included  two spatial operators, the \emph{near} operator and the \emph{surrounded} operator. The points in space satisfying `near $\phi$' are all those points close to any point satisfying $\phi$. In other words, the \emph{near} operator is interpreted as a closure operator on space. The points satisfying `$\phi$ surrounded by $\psi$', instead, are all those points satisfying $\phi$ from which no path can be found that passes by a point not satisfying $\phi$ without first passing by a point satisfying $\psi$. In other words, these are those points, satisfying $\phi$, that are surrounded by points satisfying $\psi$.

Two different spatial model-checkers for finite (quasi-discrete) closure spaces were developed based on this foundational work: \topochecker and \VoxLogicA.\footnote{\topochecker: a topological model checker, see
http://topochecker.isti.cnr.it, https://github.com/
vincenzoml/topochecker}$^,$\footnote{\VoxLogicA:\@ See \url{https://github.com/vincenzoml/VoxLogicA}.}
These tools have been used successfully in several applications. For example, for the analysis of the spatial and spatio-temporal aspects of excessive delays in returning bikes in the London bike sharing system~\cite{CLMPV16} and the spatio-temporal aspects of the emergence of Turing patterns~\cite{CLLM16bertinoro,Ne+18} and the robustness of their formation in the presence of noise.  Also in the field of medical imaging these spatial model checking techniques have shown to be very promising as a novel approach to the segmentation of malignant brain lesions~\cite{BCLM19,Ba+20} in 3D MRI scans as well as for the segmentation of normal brain tissue, such as white and grey matter~\cite{BelmonteCLM19}.  A similar approach has also been used for the segmentation of nevi\footnote{Benign and malignant lesions of the skin.} in 2D dermoscopic images~\cite{BBCLM21}. The application of spatial model checking techniques on publicly available datasets, consisting of hundreds of such medical images, showed that an accuracy and computational efficiency can be obtained that is in line with the state-of-the-art in such fields.

So far, spatial model checking approaches focused on \emph{discrete} spatial structures, i.e., discrete sets of  points related by an adjacency relation, namely graphs. 
However, in several domains, e.g., medical imaging and visualisation in virtual reality, reasoning on \emph{continuous} space is required. For example, computerised 3D visualisation of medical images can help physicians to make better diagnoses or treatment plans. Images used for visualisation often consist of \emph{continuous} spatial structures that are divided into suitable areas of different size using mesh techniques such as \emph{triangular surface meshes} or \emph{tetrahedral volume meshes} (see for example~\cite{LevinePRZ2012}).

In the present work we therefore focus on the development of the foundations to reason about and model-check properties of \emph{continuous} space. Interestingly, we will also show how our earlier results on discrete space are, in part, also relevant to model-checking continuous space. 
Unlike the topological semantics, where formulas are interpreted in the powerset algebra of a topological space, in the polyhedral semantics formulas are assigned polyhedral subsets of an $m$-dimensional Euclidean space. Polyhedral subsets can be thought of as finite unions of simplexes (i.e., $n$-dimensional triangles). Using piecewise linear geometry (triangulations, nerves), the work in~\cite{BMMP18}
gives a full characterization of the intuitionistic and modal logics of the class of all compact polyhedra and~\cite{A-DBGM20,ADay19} provide an infinite family of polyhedrally complete modal and intermediate logics.

The version of \SLCS{} we use in this paper is obtained by extending the polyhedral modal logic with a  \emph{spatial reachability} modality $\slreach$, which is a variant of the $\rho$ operator originally proposed in~\cite{BCLM19} that, in turn, can be seen as a spatial version of the Existential Until operator of temporal logics (for more examples of spatial reachability
operators, see also~\cite{NBBL22}).
Roughly speaking, $\slreach(\phi,\psi)$ (pronounced as $\psi$ is reachable through $\phi$) means that a point satisfying $\psi$ is reachable by a  path satisfying $\phi$ along the way. The reachability modality is quite expressive and other operators, relevant for the intended applications (such as ``surrounded'', or ``grow'', discussed in more detail throughout the paper), can be defined based on it. We show that the reachability modality can be defined for polyhedral models. 

One of the main conceptual results of the present paper is that, starting from a polyhedral model $\model{X}$, we can build a relational (Kripke) model $\KM(\model{X})$ satisfying the same formulas of the logic.
In particular, $\KM(\model{X})$ is isomorphic to the face poset of an underlying triangulation of $\model{X}$.
Triangulation is a standard technique of piecewise linear geometry in which each polyhedron is decomposed in simplexes. 
That triangulations play an important role in the logical analysis of polyhedra has already been observed in~\cite{BMMP18,ADay19,A-DBGM20}. However, here we show this also for the language enriched with the reachability modality $\slreach$.
In particular, $\KM(\model{X})$ captures all the properties expressible in this extended language.

The finite state, Kripke-style semantics that we define preserves all the information that can be discerned by \SLCS formulas. This is the key for introducing a novel \emph{geometric model checking} technique to analyse continuous space. We prove that the continuous model of the extended language can be turned into a finite relational model for the same language without losing any of the logical information. A model checking algorithm, along the lines of~\cite{BCLM19}, has been implemented in the free and open source geometric model checker \PolyLogicA, which brings  the core features of \VoxLogicA (global model checking, concurrent multi-core execution, ``memoization'' at the syntactic level) to the continuous space domain. 

Finally, we define a notion of \emph{bisimilarity} between two polyhedral models, and we prove that bisimilarity preserves and reflects logical equivalence. 

 Summarizing, the original contributions in this paper are:

\begin{itemize}
\item The development of a (continuous space) polyhedral semantics for the logic \SLCS building upon recent developments in polyhedral semantics for modal and intuitionistic logic~\cite{BMMP18,A-DBGM20,ADay19};
\item The introduction of a \emph{finite} state relational Kripke model $\KM(\model{X})$ 
for any given (continuous space) polyhedral model $\model{X}$, such that $\KM(\model{X})$ provides a full logical invariant for $\model{X}$ with respect to \SLCS;
\item A novel geometric spatial model checking algorithm for the verification of \SLCS properties on polyhedral models;
\item First proof-of-concept model checking of two realistic-size polyhedral models using \PolyLogicA, the first-in-its-kind prototype implementation of the geometric model checking algorithm;
\item A novel notion of \emph{simplicial bisimilarity} characterising logical equivalence for \SLCS\ formulas.
\end{itemize}
 
\paragraph{Further related work.}
The theoretical framework for spatial model checking of continuous space in the present paper   is based on spatial models involving polyhedra (see also the work in~\cite{CM21}, generalising the semantics of \SLCS to a categorical setting based on hyperdoctrines). Polyhedra also play an important role in development of model checking algorithms for the verification of behavioural properties of real-time and hybrid systems (see for example~\cite{HenzingerH94a,Alur11,BogomolovFGH17,Henzinger2000,AlurGHLM19} and references therein). In that context polyhedra, and their related notions such as template polyhedra~\cite{SankaranarayananDI08,BogomolovFGH17} and zonotopes~\cite{GirardG08}, are obtained from sets of linear inequalities involving real-time constraints on system behaviour and are a natural representation of sets of states of such systems. In the present paper we focus on \emph{spatial} properties of continuous space rather than behavioural properties.

Topology and simplicial complexes also play an important role in the domain of geographic information systems (GIS). In that domain of application simplicial complexes are used as an efficient data structure to store large geospatial data sets~\cite{Breunig2020} in 2D or 3D. They also form the core of several important tools in this domain such as the GeoToolKit~\cite{BalovnevBBCMPSSST04}. This could therefore be in principle another interesting domain of application for spatial model checking techniques, enriching the spatial query languages that are currently used in this database oriented domain.

Model Checking of simplicial complexes is also the topic of a recent, independently developed result by Loreti and Quadrini~\cite{LoQ21}. Therein, a logic  
is defined, sharing similar syntax, but \emph{not its semantics}, with the logic \SLCS\ used in the present paper{.} 
More precisely, the domain upon which formulas are interpreted, are not (sets of) points in polyhedra, but rather (sets of) simplexes. 
This is reflected by the choice of the adjacency relation between objects\footnote{Actually, the paper proposes three different relations, see Definition~1--3 therein.} giving rise to \emph{nearness} and \emph{reachability}. In contrast, the domain of interpretation of our spatial logic \SLCS\ are points of topological spaces, with polyhedral atomic valuations, and reachability is the classical topological notion, defined via paths.

Looking at tools, perhaps closest to our work is the python library \texttt{pymeshlab}~\cite{pymeshlab}, which is able to programmatically modify 3D meshes based on pre-built operators (mostly traditional 3D imaging filters). Since also the selection of connected components is an operator of the library, we envisage that it could be possible, given an SLCS formula, to encode it as a python program containing \texttt{pymeshlab} functions, effectively using the library as a backend for \VoxLogicA. Indeed, using such a library within the context of a traditional programming language does not lay at the same level of abstraction of using a simple declarative language such as \SLCS, with automatic parallelisation, and automatic memoization (caching) of intermediate results, so a direct comparison of the library and our tool would be misleading.

\paragraph{Outline.} Section~\ref{sec:background} introduces the basic geometrical notions and notation. Section~\ref{sec:slcs} recalls \SLCS\ and provides its semantics on polyhedral models.   Section~\ref{sec:polyhedral-model-checking} and Section~\ref{sec:polylogica} present the foundations for geometric model checking and the related model checker \PolyLogicA, respectively, including two proof-of-concept examples on realistic size polyhedra. In Section~\ref{sec:bisim} the concept of simplicial bisimilarity is introduced and it is shown that it characterises logic equivalence for \SLCS\ formulas. Section~\ref{sec:conclusions} concludes the paper with an outlook for future work. Proofs of the most relevant theoretical results are reported in Appendix~\ref{sec:appendix}.

%% file: background-simplicial-complexes.tex
\section{Background}
\label{sec:background}

In this section, we establish the basic geometric notions that we use in this work. See~\cite[Chapter 2]{Maunder1980} for more details on these matters. In particular, we recall simplexes and simplicial complexes. They are frequently used in, for example,  computer graphics and simulation. Polyhedra are the set-theoretic union of simplicial complexes and form the spatial models for the polyhedra model checking approach proposed in subsequent sections.

\begin{defi}[Simplex]\label{def:simplex}
	\iflong A \emph{(geometric)}
	\else
	A
	\fi
	\emph{$d$-simplex} $\sigma$ is the convex hull of a finite set $V = \{v_0,v_1, \dots, v_d\}$ $\subseteq$  $\R^m$ of $d+1$ affinely independent points,\footnote{
		$v_0,\dots,v_d$ are \emph{affinely independent} if $v_1 - v_0,\dots,v_d-v_0$ are linearly independent.
		In particular, this condition implies that $d \leq m$.
	} that is the set
                \[
		\sigma = \left\{  \lambda_0 v_0 + \cdots + \lambda_d v_d  \;\middle|\;  \forall i. \lambda_i \in [0,1] \;\text{and}\; \sum_{i=0}^d \lambda_i = 1  \right\}.
              \]
	The number $d$ is called the \emph{dimension} of $\sigma$ and $v_0,\dots,v_d$ are called its \emph{vertices}.
\end{defi}

\iflong
\noindent
A simplex uniquely determines the set of its vertices~\cite[Proposition 2.3.3]{Maunder1980}, so the concept of dimension and set of vertices are well-defined.
Two simplexes coincide iff they share the same set of vertices.
We say that the set of vertices $V$ \emph{spans} the simplex $\sigma$.
\fi
\iflong
Simplexes are extremely simple geometric objects, dependent essentially only on their dimension~\cite[Proposition 2.3.4]{Maunder1980}. For instance, in three dimensions, simplexes can be just points, segments, triangles, or tetrahedra.
There is one ``limit case'' to keep in mind:
if $V = \emptyset$, then also $\sigma = \emptyset$.
We call this the \emph{empty simplex} and define its dimension to be $-1$.
\fi
\noindent In Definition~\ref{def:simplex}, any subset of $\{v_0,\dots, v_d\}$ is also a set of affinely independent points, and thus it spans a simplex $\tau$:
we call $\tau$ a \emph{face} of $\sigma$ (in symbols $\tau \face \sigma$), and we call it a \emph{proper face} if $\tau \ne \emptyset$ and $\tau \ne \sigma$.

Simplexes are, informally speaking, the simplest linear convex bounded shapes. A two-dimensional simplex is a triangle; a three-dimensional simplex is a tetrahedron, and so on. Note that the two-dimensional faces of a tetrahedron are triangles, whose sides are line segments (one-dimensional simplexes), whose faces (endpoints) are points (zero-dimensional simplexes).

Next, we identify the ``internal part'' of a simplex.
\begin{defi}[Relative interior]
	In terms of the notation of Definition~\ref{def:simplex}, let the \emph{relative interior} of $\sigma$ be the set
		$\relint{\sigma} := \left\{ \sum_{i=0}^d \lambda_i v_i  \;\middle|\;  \forall i. \lambda_i \in (0,1] \;\text{and}\; \sum_{i=0}^d \lambda_i = 1  \right\}$.
\end{defi}
\noindent
Note that if $\sigma$ is non-empty then also $\relint{\sigma}$ is non-empty.
For instance, $b_{\sigma} := \sum_{i=0}^d \frac{1}{d+1} v_i$ (i.e., the \emph{barycentre} of $\sigma$) is an element of $\relint{\sigma}$.
In particular, the relative interior of a point $p$ is $p$ itself, and the relative interior of the empty simplex is the empty simplex itself.
	
We emphasize another interesting property, which follows from the definition of face and that of relative interior: Each simplex $\sigma$ is partitioned by the relative interiors of its faces, that is, $\sigma = \bigcup\{ \relint{\tau} \;|\; \tau \face \sigma \}$. For example, a triangle can be partitioned into its interior (an open triangle), three open segments (sides without endpoints) and the three vertices.

In this paper we adopt the Kuratowski axiomatic characterisation of topological spaces, based on the closure operator.
We denote the latter by $\Cl$.
Given topological space $(X,\Cl)$ we let $\interior$ denote the interior operator, i.e., the dual of $\Cl$ defined as  $\interior(A)= X\setminus \Cl(X\setminus A)$ for all $A \subseteq X$.
Finally, for $X' \subseteq X$, the topological sub-space of $(X,\Cl)$ generated by $X'$ is the pair $(X',\Cl')$, where $\Cl'(A) = \Cl(A) \cap X'$ for all $A \subseteq X'$.
For ease of notation, we will indicate the topological space $(X,\Cl)$ simply as $X$ when no confusion arises.

Simplexes are bounded, convex, compact subspaces of $\R^m$~\cite[Proposition 2.3.3]{Maunder1980}.
Being subsets of an Euclidean space $\R^{m}$, a simplex $\sigma$ inherits the topological structure of the sub-space from $\R^{m}$.
Let us indicate by $\Cl_m$ the closure operator of $\R^{m}$ and by $\Cl_{\sigma}$ the induced closure operator on $\sigma$.
Notice that, since $\sigma$ is a closed subset of $\R^{m}$, the closure of any subset $A \subseteq \sigma$ is the same computed in both topologies, that is, $\Cl_m(A) = \Cl_\sigma(A)$.
An emblematic example is the closure of $\relint{\sigma}$:
$\Cl_\sigma(\relint{\sigma}) = \Cl_m(\relint{\sigma}) = \sigma$.
However, the interior of a set may depend on the topology considered:
for example, $\interior_m(\sigma)$ is either $\relint{\sigma}$ if the dimension of $\sigma$ is $m$, or the empty set otherwise.
On the other hand, $\interior_{\sigma}(\sigma) = \sigma$, independently from the dimension.

More complex spaces are obtained by ``gluing'' simplexes together.

\begin{defi}[Simplicial complex]\label{def:simplicialComplex}
	A
	\iflong \emph{(geometric)}
	\fi 
	\emph{simplicial complex} $\Ksc$ is a finite set of simplexes of $\R^m$ such that:
	\begin{enumerate}
		\item\label{prop:one} If $\sigma \in \Ksc$ and $\tau$ is a face of $\sigma$, then $\tau \in \Ksc$;
		\item\label{prop:two} If $\sigma, \tau \in \Ksc$, then $\sigma \cap \tau$ is a face of $\sigma$ and $\tau$ (possibly the empty simplex).
	\end{enumerate}
\end{defi}

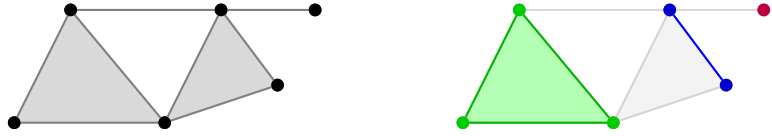
\begin{figure}[t]
	\centering
	\begin{tikzpicture}
		\tikzstyle{point}=[circle,thick,draw=black,fill=black,inner sep=0pt,minimum width=4pt,minimum height=4pt]
		\tikzstyle{edge}=[gray,thick]
		\tikzstyle{tria}=[fill=gray!30]

		\node[point] (a) at (   0,  0) {};
		\node[point] (b) at (   2,  0) {};
		\node[point] (c) at ( .75,1.5) {};
		\node[point] (d) at (2.75,1.5) {};
		\node[point] (e) at ( 3.5, .5) {};
		\node[point] (f) at (   4,1.5) {};

		\fill [tria] (a.center) -- (b.center) -- (c.center) -- cycle;
		\fill [tria] (b.center) -- (d.center) -- (e.center) -- cycle;

		\draw [edge] (a) -- (b);
		\draw [edge] (a) -- (c);
		\draw [edge] (b) -- (c);
		\draw [edge] (b) -- (d);
		\draw [edge] (c) -- (d);
		\draw [edge] (b) -- (e);
		\draw [edge] (d) -- (e);
		\draw [edge] (d) -- (f);

		\node[point] (a) at (   0,  0) {};
		\node[point] (b) at (   2,  0) {};
		\node[point] (c) at ( .75,1.5) {};
		\node[point] (d) at (2.75,1.5) {};
		\node[point] (e) at ( 3.5, .5) {};
		\node[point] (f) at (   4,1.5) {};
	\end{tikzpicture}
	\hspace{4em}
	\begin{tikzpicture}
		\tikzstyle{point}=[circle,thick,draw=black,fill=black,inner sep=0pt,minimum width=4pt,minimum height=4pt]
		\tikzstyle{edge}=[gray,thick]
		\tikzstyle{tria}=[fill=gray!30]

		\begin{scope}[opacity=.3]
			\node[point] (a) at (   0,  0) {};
			\node[point] (b) at (   2,  0) {};
			\node[point] (c) at ( .75,1.5) {};
			\node[point] (d) at (2.75,1.5) {};
			\node[point] (e) at ( 3.5, .5) {};
			\node[point] (f) at (   4,1.5) {};

			\fill [tria] (a.center) -- (b.center) -- (c.center) -- cycle;
			\fill [tria] (b.center) -- (d.center) -- (e.center) -- cycle;

			\draw[edge] (a) -- (b);
			\draw[edge] (a) -- (c);
			\draw[edge] (b) -- (c);
			\draw[edge] (b) -- (d);
			\draw[edge] (c) -- (d);
			\draw[edge] (b) -- (e);
			\draw[edge] (d) -- (e);
			\draw[edge] (d) -- (f);

			\node[point] (a) at (   0,  0) {};
			\node[point] (b) at (   2,  0) {};
			\node[point] (c) at ( .75,1.5) {};
			\node[point] (d) at (2.75,1.5) {};
			\node[point] (e) at ( 3.5, .5) {};
			\node[point] (f) at (   4,1.5) {};
		\end{scope}

		\fill [tria,green!30] (a.center) -- (b.center) -- (c.center) -- cycle;
		\draw[edge,green!70!black] (a) -- (b);
		\draw[edge,green!70!black] (a) -- (c);
		\draw[edge,green!70!black] (b) -- (c);
		\node[point,green!80!black] (a) at (   0,  0) {};
		\node[point,green!80!black] (b) at (   2,  0) {};
		\node[point,green!80!black] (c) at ( .75,1.5) {};

		\draw[edge,blue] (d) -- (e);
		\node[point,blue!80!black] (d) at (2.75,1.5) {};
		\node[point,blue!80!black] (e) at ( 3.5, .5) {};

		\node[point,purple] (f) at (   4,1.5) {};

	\end{tikzpicture}
	\caption{\label{fig:simplicialComplex}On the left, a graphical representation of a simplicial complex is shown in grey and black.
	On the right, we highlight some of the faces of the simplicial complex, of dimension 2, 1 and 0 (the green triangle, blue segment and red point, respectively).
	}	
	
\end{figure}

\begin{figure}[t]
	\centering
	\begin{tikzpicture}
		\tikzstyle{point}=[circle,thick,draw=black,fill=black,inner sep=0pt,minimum width=4pt,minimum height=4pt]
		\tikzstyle{edge}=[gray,thick]
		\tikzstyle{tria}=[fill=gray!30]

		\node[point] (a) at (    0,0) {};
		\node[point] (b) at (    0,2) {};
		\node[point] (c) at ( 2.5,1) {};
		
		\node[point] (d) at ( 1.3,  2) {};
		\node[point] (e) at ( 2.5,  2) {};
		\node[point] (f) at ( 1.9,  0) {};

		\fill [tria] (a.center) -- (b.center) -- (c.center) -- cycle;
		\fill [tria,fill=white!70!gray,opacity=.7] (d.center) -- (e.center) -- (f.center) -- cycle;

		\draw [edge] (a) -- (b);
		\draw [edge] (b) -- (c);
		\draw [edge] (c) -- (a);
		\draw [edge] (d) -- (e);
		\draw [edge] (e) -- (f);
		\draw [edge] (f) -- (d);

		\node[point] (a) at (   0,  0) {};
		\node[point] (b) at (   0,  2) {};
		\node[point] (c) at ( 2.5,  1) {};
		\node[point] (d) at ( 1.3,  2) {};
		\node[point] (e) at ( 2.5,  2) {};
		\node[point] (f) at ( 1.9,  0) {};
	\end{tikzpicture}
	\caption{\label{fig:notSimplicialComplex}A graphical representation of a collection of simplexes (in 2D) \emph{not} forming a simplicial complex.
	In particular, the collection of simplexes depicted consists of 6 points, 6 edges and 2 triangles.
	Notice that Condition $(2)$ of Definition~\ref{def:simplicialComplex} is not respected:
	the intersection between the two triangles is not a simplex of the collection itself. %
	}	
	
\end{figure}
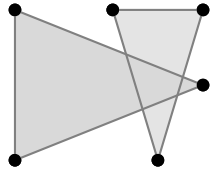

\noindent
Figure~\ref{fig:simplicialComplex} illustrates an example of a simplicial complex $\Ksc$ and highlights some of the simplexes comprising it.
Figure~\ref{fig:notSimplicialComplex} illustrates a collection of simplexes that \emph{does not} form a simplicial complex.
The \emph{dimension} of $\Ksc$ is the maximum of the dimensions of its simplexes, i.e., $2$ in Figure~\ref{fig:simplicialComplex}.
The face relation on simplexes is inherited by simplicial complexes in the expected way:
the face relation $\face$ on simplicial complex $\Ksc$ is the union of the face relations on the simplexes composing $\Ksc$.

Given simplicial complex $\Ksc = \{\sigma_1,\ldots,\sigma_n\}$, the \emph{polyhedron} of $\Ksc$, 
denoted by  $\poly{\Ksc}$, is the set-theoretic union of the simplexes in $\Ksc$. 
Formally, $\poly{\Ksc} := \bigcup_{i=1}^n \sigma_i$. Figure~\ref{subfig:K1} shows a simplicial complex, which is a set of 11 elements -- 4 points, 5 open segments and 2 open triangles. The set-theoretic union of these 11 simplexes is the square shown in Figure~\ref{subfig:SQ} -- the corresponding polyhedron. Note that a different simplicial complex, shown in Figure~\ref{subfig:K2}, has the same (underlying) polyhedron.

As in the case of simplexes, the polyhedron $\poly{\Ksc}$ is a subset of the ambient space $\R^{m}$ and so inherits the topological structure of sub-space of $\R^{m}$.
Let us indicate with $\Cl_{\Ksc}$ the corresponding closure operator --- notice that this operator only depends on the set $\poly{\Ksc}$.
A point of $\poly{\Ksc}$ may belong to several of the simplexes in $\Ksc$.
However,	 there is a natural way to associate to each point of $\poly{\Ksc}$ the ``smallest'' simplex it belongs to.
\begin{lem}\label{lemma:partition}
	Each point of $\poly{\Ksc}$ belongs to the relative interior of exactly one non-empty simplex in $\Ksc$.
	That is, $\relint{\Ksc} := \{ \relint{\sigma} \,|\, \sigma \in \Ksc\setminus \{\emptyset\} \}$ is a partition of $\poly{\Ksc}$.
\end{lem}

\noindent
We call $\relint{\Ksc}$ a \emph{simplicial partition} of $|\Ksc|$, and we call its elements the \emph{cells} of the partition.\footnote{
	We use the terminology \emph{cells} in this way for the purposes of this paper; there is no relation between such cells and the so-called \emph{cell complexes} of algebraic topology.
}
Note that distinct simplicial complexes induce distinct partitions, even when they are associated to the same polyhedron, as shown in Figure~\ref{fig:twoKOnePoly}.
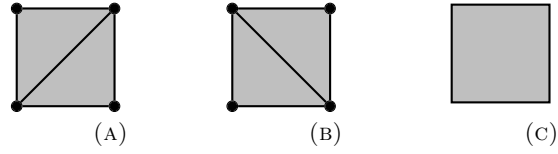
\begin{figure}
       \centering
	\subfloat[\hspace{-0.5cm}\parbox{1in}{}]{\label{subfig:K1}
		\begin{tikzpicture}[scale=1.3]%
			\tikzstyle{point}=[circle,draw=black,fill=white,inner sep=0pt,minimum width=4pt,minimum height=4pt]
			\node (p2)[point] at (1,0) {};
			\node (p3)[point] at (1,1) {};
			\node (p4)[point] at (2,0) {};
			\node (p5)[point] at (2,1) {};
			
			\draw [thick] (p2) -- (p3);
			\draw [thick] (p3) -- (p5);
			\draw [thick] (p2) -- (p5);
			\draw [thick] (p4) -- (p5);
			\draw [thick] (p2) -- (p4);
			
			\begin{scope}[on background layer]
			\fill [fill=lightgray] (p2.center) -- (p5.center) -- (p4.center);
			\fill [fill=lightgray] (p2.center) -- (p3.center) -- (p5.center);
			\end{scope}
			
			\filldraw [black] (p2) circle (1.25pt);
			\filldraw [black] (p5) circle (1.25pt);
                         \filldraw [black] (p4) circle (1.25pt);
			\filldraw [black] (p3) circle (1.25pt);
		\end{tikzpicture}
}\quad\quad\quad
	\subfloat[\hspace{-0.5cm}\parbox{1in}{}]{\label{subfig:K2}
		\begin{tikzpicture}[scale=1.3]%
			\tikzstyle{point}=[circle,draw=black,fill=white,inner sep=0pt,minimum width=4pt,minimum height=4pt]
			\node (p2)[point] at (1,0) {};
			\node (p3)[point] at (1,1) {};
			\node (p4)[point] at (2,0) {};
			\node (p5)[point] at (2,1) {};
			
			\draw [thick] (p2) -- (p3);
			\draw [thick] (p3) -- (p5);
			\draw [thick] (p3) -- (p4);
			\draw [thick] (p4) -- (p5);
			\draw [thick] (p2) -- (p4);
			
			\begin{scope}[on background layer]
			\fill [fill=lightgray] (p2.center) -- (p3.center) -- (p4.center);
			\fill [fill=lightgray] (p5.center) -- (p3.center) -- (p4.center);
			\end{scope}

			\filldraw [black] (p2) circle (1.25pt);
			\filldraw [black] (p5) circle (1.25pt);
                         \filldraw [black] (p4) circle (1.25pt);
			\filldraw [black] (p3) circle (1.25pt);
		\end{tikzpicture}
}\quad\quad\quad
	\subfloat[\hspace{-0.5cm}\parbox{1in}{}]{\label{subfig:SQ}
		\begin{tikzpicture}[scale=1.3]%
			
			\node (p2) at (1,0) {};
			\node (p3) at (1,1) {};
			\node (p4) at (2,0) {};
			\node (p5) at (2,1) {};

			\draw [thick] (p2.center) -- (p3.center) -- (p5.center) -- (p4.center) -- (p2.center) -- cycle;
			
			\begin{scope}[on background layer]
			\fill [fill=lightgray] (p2.center) -- (p3.center) -- (p5.center) -- (p4.center) -- (p2.center);
			\end{scope}

		\end{tikzpicture}
}
\caption{Two different simplicial complexes, (\ref{subfig:K1}) and (\ref{subfig:K2}) leading to the same polyhedron (\ref{subfig:SQ})}
\label{fig:twoKOnePoly}
\end{figure}
From now on, to ease readability, we fix a simplicial complex $\Ksc$,  with the associated $\poly{\Ksc}$ and $\relint{\Ksc}$.

\noindent Finally, we recall the topological notion of \emph{path}. 
\begin{defi}
A topological \emph{path} in a topological space $P$ is a total, continuous function $\pi : [0,1] \to P$, where $[0,1]$ is equipped with the subspace topology of $\R$.
\end{defi} 
With a mild abuse of notation, for $S$ a subset of $[0,1]$ and $\pi$ a path, we write $\pi(S)$ to denote $\{ \pi(x) \mid x \in S\}$.
\iflong
\noindent
As a corollary of Lemma~\ref{lemma:commonTriangulation} we obtain the following result.

\begin{corollary}\label{cor:simplicialPartitionRefinement}
	Given simplicial partitions $\relint{\Lsc}_1, \dots, \relint{\Lsc}_k$ of the same polyhedron $P$, there exists a simplicial partition $\relint{\Ksc}$ refining each $\relint{\Lsc}_i$.
\end{corollary}

\gnote{Proof?}

\gnote{Add note on triangulation (as in: homeomorphisms from $\poly{\Ksc}$ to a topological space) and how this works extends to that; add note on spaces which are locally simplexes.}
\fi

%% file: slcs-on-polyhedra.tex
\section{Interpreting \SLCS on Polyhedra}
\label{sec:slcs}

In this section we introduce the main theory driving our model checking approach to polyhedra.  In the classical topological tradition, valuations of atomic propositions can be arbitrary subsets of the space. In this work, instead, we restrict our attention  to a specific class of spatial regions, namely unions of cells of a fixed simplicial partition. This simple change makes it possible to define the spatial logic \SLCS on continuous space, while retaining decidability of the model checking problem. 
First of all, we introduce the syntax of the variant of \SLCS that we use in this paper, that is based on the binary modality $\slreach$ instead of $\rho$ of~\cite{BCLM19}; the relationship between $\rho$ and $\slreach$ will be shown in Proposition~\ref{prop:rho-and-gamma}.

\begin{defi}[Syntax]\label{def:slcs-syntax}
	The syntax of the logic \SLCS is:
	\begin{equation*}
		\phi  \;::=\;  \top  \;|\; p \;|\;  \neg \phi  \;|\;  \phi \land \phi  \;|\;  \Box \phi  \;|\;  \slreach ( \phi, \phi )
	\end{equation*}
	where $p$ is an atomic proposition, taken from a fixed finite set $\AP$.
\end{defi}

\noindent
\iflong 

For ease of presentation, we will work with a finite set $\AP = \{q_1, \dots, q_l\}$ of atomic propositions; however, this requirement is not essential and it is easily dispensable.

\else
Thus, we enhance the basic modal logic with a \emph{spatial reachability} operator $\slreach$.
As in the standard topological semantics for modal logic, we interpret formulas as sets of points.
Boolean operators are given their standard set-theoretical interpretation; disjunction $\lor$ is derived via the De Morgan laws. The $\Box$ modality corresponds to \emph{topological} interior $\interior$. 

The formula $\slreach(\phi,\psi)$ (``reach $\psi$ through $\phi$'') is satisfied by a point if there is a path rooted in that point, leading to a point satisfying $\psi$ and whose intermediate points all satisfy $\phi$. We provide a number of examples of the operators after presenting the models and formal semantics of the logic. 
\iflong
There are at least two good reasons to focus on this particular class of spatial objects.
First, this allows for a \emph{finitary} treatment of the topological semantics for modal logic
 \vnote{Expand on this}, that admits \emph{sequential refinements} of the spatial models considered (essentially as a consequence of Lemma~\ref{lemma:commonTriangulation}).
Second, the model checking problem is decidable, although the considered space is continuous. 
This guarantees a \emph{finitary} treatment of the semantics (namely, there are finitely many classes modulo logical equivalence) and decidability of the model checking problems.
Moreover, this work extends for the first time spatial model checking to a class of models defined in continuous space and that has been widely used in modern Computer Science (e.g., being the theoretical basis for the ubiquitous \emph{3D meshes}).
\else
\fi
In the following, we let  $\powset(P)$ denote the powerset of $P$.

\begin{defi}[Model]\label{def:polyhedralModel}
	A \emph{Polyhedral Model} is a triplet $\model{X} = \tuple{ P, \Ksc, V }$, where
		$P \subseteq \R^d$ is a polyhedron, $\Ksc$ is a simplicial complex such that $P = |\Ksc|$, and
		$V: \AP \to \powset(P)$ is a valuation such that $V(p)$ is a union of cells of $\relint{\Ksc}$.
\end{defi}

\noindent
Polyhedral models are essentially topological models with some extra restrictions on the valuation:
$P$ plays the role of the topological space and $V$ is used to interpret atomic propositions as specific subsets of this space, namely those that are the union of a finite number of simplicial cells.
\iflong 
The reader might wonder why we do not restrict the range of $V$ to polyhedra instead of union of cells.
The reason is that we want to interpret modal operators with their standard topological interpretation, and the spatial regions expressible this way are more complex than polyhedra (e.g., the complement of a polyhedron is not in general a polyhedron).
\fi
From now on, fix a polyhedral model $\model{X} = \tuple{ P, \Ksc, V }$. 

\iflong 
We are now ready to introduce the semantics of our logic.
\fi
\begin{defi}[Semantics]\label{def:slcs-semantics}
	Given $x \in P$, satisfaction $\model{X}, x \vDash \phi$ over formulas $\phi$ is given by the following inductive clauses, where we let $\sem{\phi}$ denote the set $\{ x\in P \,|\, \model{X}, x \vDash \phi  \}$:
	\begin{equation*}
	\begin{array}{lcl}
		\model{X}, x \vDash \top
			&\multicolumn{2}{l}{\text{always holds}}  \\
		\model{X}, x \vDash p
			&\iff
			&x \in V(p) \;\,\text{for } p \in \AP  \\
		\model{X}, x \vDash \neg \phi
			&\iff
			&\model{X}, x \nvDash \phi  \\
		\model{X}, x \vDash \phi \land \psi
			&\iff
			&\model{X}, x \vDash \phi \;\text{and}\; \model{X}, x \vDash \psi  \\
		\model{X}, x \vDash \Box \phi
			&\iff
			&x \in \interior_{P}(\sem{\phi}) \\
		\iflong
		\model{X}, x \vDash \slreach( \phi, \psi )
			&\iff
			&\exists \pi : [0,1] \to P.
				\begin{cases}
					\pi(0) = x  \\
					\forall r \in (0,1).\; \model{X}, \pi(r) \vDash \phi  \\
					\model{X}, \pi(1) \vDash \psi
				\end{cases}
		\else
		\model{X}, x \vDash \slreach( \phi, \psi )
			& \iff
			& \text{there exists a path $\pi$ such that}\\
			&&		\pi(0) = x,\,
			                 \pi(1) \in \sem{\psi}  \text{and }
					\pi((0,1)) \subseteq \sem{\phi} 					
		\fi
	\end{array}
	\end{equation*}
\end{defi}
\noindent
\iflong

\gnote{Add definition of topological interior or rephrase definition of $\Box$.}
\fi
\noindent
The definition of the satisfaction relation for  the standard operators of modal logic is  the usual
one for the classical topological interpretation. In particular, note the interpretation of $\Box \phi$ as the \emph{topological interior} with respect to the topology of $P = \poly{\Ksc}$, intuitively expressing that point $x$ is in the ``internal'' part of the set of points satisfying $\phi$.
Notice that the closure operator $\Cl_{P}$ can be obtained as the dual of topological interior $\Diamond \phi = \neg (\Box \neg \phi)$.
Figure~\ref{fig:example-model-A-and-formulas} illustrates these operators and their combination applied on a simple polyhedral model.

Regarding spatial reachability, a point $x$ satisfies  $\slreach( \phi, \psi )$ in model
$\model{X}$ if there is a path $\pi$ rooted in $x$ leading to a point $y$ satisfying $\psi$; in addition, all the points that lay in $\pi$, except $x$ and $y$, are required to satisfy $\phi$. Indeed, several different variants of reachability could be defined using this operator.
\iflong
For instance, 
$\slreach' ( \phi, \psi )$ such that
$\model{X}, x \vDash \slreach'( \phi, \psi) \iff $ \emph{there is path $\pi$ such that $\pi(1) \in \sem{\psi}$ and $\pi([0,1)) \subseteq \sem{\phi}$} is equivalent to $\phi \land \slreach( \phi, \psi )$; similarly
$\slreach'' ( \phi, \psi )$ such that
$\model{X}, x \vDash \slreach''( \phi, \psi) \iff $ there is path $\pi$ such that $\pi(1) \in \sem{\psi}$ and $\pi((0,1]) \subseteq \sem{\phi}$ is equivalent to $ \slreach( \phi, \psi \land \phi)$. 
\fi

As a prominent example, the reachability modality $\rho \, \psi\,[\phi]$ introduced in~\cite{BCLM19}, that we also employ to introduce some derived operators, can be defined as $\rho\, \psi [\phi] := \psi \lor \GL(\phi,\psi)$. Actually, the two operators are inter-definable, by letting $\slreach(\phi,\psi) :=  \rho\,(\phi \land \rho\,\psi[\phi])[\phi]$ (see Proposition~\ref{prop:rho-and-gamma}).
In this work we opt to use $\GL$ since, in the context of polyhedral models, its definition is more concise.
\iflong
 The two modalities are different; for instance, a point satisfying $\psi \land \neg\phi$ satisfies $\rho\, \psi\,[\phi]$ but does not satisfy $\slreach(\phi,\psi)$. However, Proposition~\ref{prop:rho-and-gamma} below shows that the operators $\rho$ and $\slreach$ are inter-definable. 
In this paper, instead, we follow the tradition in topology by basing $\slreach$ on right-closed paths;
this also simplifies proofs by saving an index.

\fi

Another relevant spatial modality is the \emph{surrounded} operator ${\mathcal S}$ (e.g.,~\cite{CLLM16,LinkerPS20,TsigkanosKG17,Ne+18} use it as a primitive of the language).
A point $x$ satisfies $\phi \, {\mathcal S} \, \psi$ if it lays in an area whose points satisfy $\phi$, and that is limited (i.e., surrounded) by points that satisfy $\psi$.
In other words, it is not possible to exit this area without passing by a point satisfying $\psi$.
Following~\cite{BCLM19}, we can define the operator ${\mathcal S}$ on polyhedral models in terms of $\rho$ through the following expression: $\phi \, \land \, \neg \rho \, (\neg (\phi \vee \psi))[\neg \psi] $.

Some examples involving the reachability operator are shown in Figure~\ref{fig:example-model-and-formulas}. We refer to the caption of that figure for more detailed explanation. Notably, we illustrate the derived operator $grow(a,b)$, that also played an important role in the brain tumour segmentation procedure presented in~\cite{BCLM19}. The operator $grow$ is reminiscent of the technique of \emph{region growing} in Medical Imaging, and it is used to characterise those areas of space satisfying $b$ that are in contact with areas of space satisfying $a$, or, in other words, the operator lets $a$ ``grow'' inside $b$ (and no further). The formal definition is $grow(\phi_1,\phi_2) := \phi_1 \lor (\phi_2 \land \rho\, \phi_1[\phi_2])$.

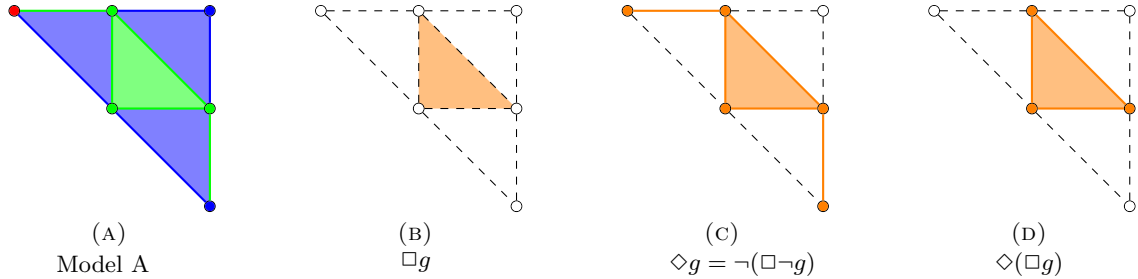
\begin{figure}[t]
       \centering
	\subfloat[\hspace{-0.5cm}\parbox{1in}{Model A}]{\label{subfig:modelA}
	\begin{tikzpicture}[scale=1.3]%
			\tikzstyle{point}=[circle,draw=black,fill=white,inner sep=0pt,minimum width=4pt,minimum height=4pt]
			  \node (p2)[point] at (0,2) {};
			\node (p4)[point] at (1,1) {};
			\node (p5)[point] at (1,2) {};
			\node (p6)[point] at (2,0) {};
			\node (p7)[point] at (2,1) {};
			\node (p8)[point] at (2,2) {};

			\draw [green,thick](p4) -- (p5);
			\draw [green,thick](p4) -- (p7);
			\draw [green,thick](p5) -- (p7);
			\draw [blue,thick](p2) -- (p4);
			\draw [green,thick](p2) -- (p5);
			\draw [blue,thick](p4) -- (p6);
			\draw [green,thick](p6) -- (p7);
			\draw [blue,thick](p7) -- (p8);
			\draw [blue,thick](p5) -- (p8);
			
			\begin{scope}[on background layer]
			\fill [fill=green!50](p4.center) -- (p5.center) -- (p7.center);
			\fill [fill=blue!50](p2.center) -- (p4.center) -- (p5.center);
			\fill [fill=blue!50](p4.center) -- (p6.center) -- (p7.center);
			\fill [fill=blue!50](p5.center) -- (p7.center) -- (p8.center);	
			\end{scope}
			
			\filldraw [red] (p2) circle (1.25pt);
			\filldraw [green] (p4) circle (1.25pt);
			\filldraw [green] (p5) circle (1.25pt);
			\filldraw [blue] (p6) circle (1.25pt);
			\filldraw [green] (p7) circle (1.25pt);
			\filldraw [blue] (p8) circle (1.25pt);

		\end{tikzpicture}
		}\hfill
	\subfloat[\hspace{-0.0cm}\parbox{1in}{$\Box g$}]{\label{subfig:Af1}
                 \begin{tikzpicture}[scale=1.3]%
			\tikzstyle{point}=[circle,draw=black,fill=white,inner sep=0pt,minimum width=4pt,minimum height=4pt]
			  \node (p2)[point] at (0,2) {};
			\node (p4)[point] at (1,1) {};
			\node (p5)[point] at (1,2) {};
			\node (p6)[point] at (2,0) {};
			\node (p7)[point] at (2,1) {};
			\node (p8)[point] at (2,2) {};

			\draw [dashed](p4) -- (p5);
			\draw [dashed](p4) -- (p7);
			\draw [dashed](p5) -- (p7);
			\draw [dashed](p2) -- (p4);
			\draw [dashed](p2) -- (p5);
			\draw [dashed](p4) -- (p6);
			\draw [dashed](p6) -- (p7);
			\draw [dashed](p7) -- (p8);
			\draw [dashed](p5) -- (p8);
			
			\begin{scope}[on background layer]
			\fill [fill=orange!50](p4.center) -- (p5.center) -- (p7.center);
			\end{scope}
			
			\filldraw [white] (p2) circle (1.25pt);
			\filldraw [white] (p4) circle (1.25pt);
			\filldraw [white] (p5) circle (1.25pt);
			\filldraw [white] (p6) circle (1.25pt);
			\filldraw [white] (p7) circle (1.25pt);
			\filldraw [white] (p8) circle (1.25pt);

		\end{tikzpicture}
		}\hfill
	\subfloat[\hspace{-0.5cm}\parbox{1in}{$\Diamond g = \neg(\Box \neg g)$}]{\label{subfig:Af2}
                 \begin{tikzpicture}[scale=1.3]%
			\tikzstyle{point}=[circle,draw=black,fill=white,inner sep=0pt,minimum width=4pt,minimum height=4pt]
			  \node (p2)[point] at (0,2) {};
			\node (p4)[point] at (1,1) {};
			\node (p5)[point] at (1,2) {};
			\node (p6)[point] at (2,0) {};
			\node (p7)[point] at (2,1) {};
			\node (p8)[point] at (2,2) {};

			\draw [orange,thick](p4) -- (p5);
			\draw [orange,thick](p4) -- (p7);
			\draw [orange,thick](p5) -- (p7);
			\draw [dashed](p2) -- (p4);
			\draw [orange,thick](p2) -- (p5);
			\draw [dashed](p4) -- (p6);
			\draw [orange,thick](p6) -- (p7);
			\draw [dashed](p7) -- (p8);
			\draw [dashed](p5) -- (p8);
			
			\begin{scope}[on background layer]
			\fill [fill=orange!50](p4.center) -- (p5.center) -- (p7.center);
			\end{scope}
			
			\filldraw [orange] (p2) circle (1.25pt);
			\filldraw [orange] (p4) circle (1.25pt);
			\filldraw [orange] (p5) circle (1.25pt);
			\filldraw [orange] (p6) circle (1.25pt);
			\filldraw [orange] (p7) circle (1.25pt);
			\filldraw [white] (p8) circle (1.25pt);

		\end{tikzpicture}
		}\hfill
	\subfloat[\hspace{-0.25cm}\parbox{1in}{$\Diamond (\Box g)$}]{\label{subfig:Af3}
                \begin{tikzpicture}[scale=1.3]%
			\tikzstyle{point}=[circle,draw=black,fill=white,inner sep=0pt,minimum width=4pt,minimum height=4pt]
			  \node (p2)[point] at (0,2) {};
			\node (p4)[point] at (1,1) {};
			\node (p5)[point] at (1,2) {};
			\node (p6)[point] at (2,0) {};
			\node (p7)[point] at (2,1) {};
			\node (p8)[point] at (2,2) {};

			\draw [orange,thick](p4) -- (p5);
			\draw [orange,thick](p4) -- (p7);
			\draw [orange,thick](p5) -- (p7);
			\draw [dashed](p2) -- (p4);
			\draw [dashed](p2) -- (p5);
			\draw [dashed](p4) -- (p6);
			\draw [dashed](p6) -- (p7);
			\draw [dashed](p7) -- (p8);
			\draw [dashed](p5) -- (p8);
			
			\begin{scope}[on background layer]
			\fill [fill=orange!50](p4.center) -- (p5.center) -- (p7.center);
			\end{scope}
			
			\filldraw [white] (p2) circle (1.25pt);
			\filldraw [orange] (p4) circle (1.25pt);
			\filldraw [orange] (p5) circle (1.25pt);
			\filldraw [white] (p6) circle (1.25pt);
			\filldraw [orange] (p7) circle (1.25pt);
			\filldraw [white] (p8) circle (1.25pt);

		\end{tikzpicture}
		}		
	\caption{\label{fig:example-model-A-and-formulas}Examples of basic topological operators on Polyhedra.~\ref{subfig:modelA}) Polyhedral model A. Circles denote 0-dimensional simplexes. The valuation of atomic propositions $r,g,b$ in model A is given by the colours \emph{red}, \emph{green}, and \emph{blue}, respectively. The points in model A that satisfy the following example formulas are shown in orange in:~\ref{subfig:Af1}) $\Box g$, i.e., points satisfying the \emph{topological interior} of the part in \emph{green} in model A; 
	\ref{subfig:Af2}) $\Diamond g$, i.e., points satisfying the \emph{topological closure} of the part in \emph{green} in model A; 
 	\ref{subfig:Af3})  $\Diamond (\Box g)$, i.e., points satisfying the closure of the interior of the part in \emph{green} in model A.}
	
\end{figure}

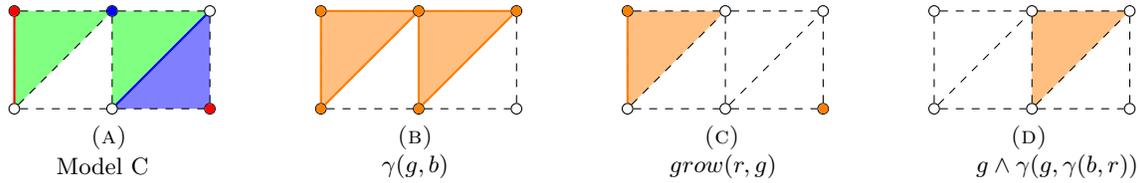
\begin{figure}[t]
	
	\subfloat[\hspace{-0.5cm}\parbox{1in}{Model C}]{\label{subfig:model}
		\begin{tikzpicture}[scale=1.3]%
			\tikzstyle{point}=[circle,draw=black,fill=white,inner sep=0pt,minimum width=4pt,minimum height=4pt]
			\node (p2)[point] at (1,0) {};
			\node (p3)[point] at (1,1) {};
			\node (p4)[point] at (2,0) {};
			\node (p5)[point] at (2,1) {};
			\node (p6)[point] at (3,0) {};
			\node (p7)[point] at (3,1) {};

			\draw [dashed] (p2) -- (p3);
			\draw [dashed] (p3) -- (p5);
			\draw [dashed] (p2) -- (p5);
			\draw [dashed] (p4) -- (p5);
			\draw [dashed] (p2) -- (p4);
			\draw [dashed] (p4) -- (p7);
			\draw [dashed] (p4) -- (p6);
			\draw [dashed] (p5) -- (p7);
			\draw [dashed] (p6) -- (p7);

			\draw [red,thick](p2) -- (p3);
			
			\draw [blue,thick](p4) -- (p7);		
			
			\begin{scope}[on background layer]
			\fill [fill=blue!50](p4.center) -- (p6.center) -- (p7.center);
			\fill [fill=green!50](p2.center) -- (p3.center) -- (p5.center);	
			\fill [fill=green!50](p4.center) -- (p5.center) -- (p7.center);
			\end{scope}
			
			\filldraw [blue] (p5) circle (1.25pt);

			\filldraw [red] (p6) circle (1.25pt);
			\filldraw [red] (p3) circle (1.25pt);

		\end{tikzpicture}	
	}\hfill
	\subfloat[\hspace{-0.25cm}\parbox{1in}{$\slreach(g,b)$}]{\label{subfig:f1}
		\begin{tikzpicture}[scale=1.3]%
			\tikzstyle{point}=[circle,draw=black,fill=white,inner sep=0pt,minimum width=4pt,minimum height=4pt]
			\node (p2)[point] at (1,0) {};
			\node (p3)[point] at (1,1) {};
			\node (p4)[point] at (2,0) {};
			\node (p5)[point] at (2,1) {};
			\node (p6)[point] at (3,0) {};
			\node (p7)[point] at (3,1) {};

			\draw [dashed] (p2) -- (p3);
			\draw [dashed] (p3) -- (p5);
			\draw [dashed] (p2) -- (p5);
			\draw [dashed] (p4) -- (p5);
			\draw [dashed] (p2) -- (p4);
			\draw [dashed] (p4) -- (p7);
			\draw [dashed] (p4) -- (p6);
			\draw [dashed] (p5) -- (p7);
			\draw [dashed] (p6) -- (p7);
						
			\draw [orange,thick](p2) -- (p3);
			\draw [orange,thick](p3) -- (p5);
			\draw [orange,thick](p2) -- (p5);
			\draw [orange,thick](p5) -- (p7);
			\draw [orange,thick](p4) -- (p7);
			\draw [orange,thick](p4) -- (p5);
					
			\begin{scope}[on background layer]
			\fill [fill=orange!50](p2.center) -- (p3.center) -- (p5.center);\fill [fill=orange!50](p4.center) -- (p5.center) -- (p7.center);
			\end{scope}
			
			\filldraw [orange] (p2) circle (1.25pt);
			\filldraw [orange] (p3) circle (1.25pt);
			\filldraw [orange] (p5) circle (1.25pt);
			\filldraw [orange] (p4) circle (1.25pt);
			\filldraw [orange] (p7) circle (1.25pt);
		\end{tikzpicture}	
	}\hfill
	\subfloat[\hspace{-0.5cm}\parbox{1in}{$grow(r,g)$}]{\label{subfig:f3}
		\begin{tikzpicture}[scale=1.3]%
			\tikzstyle{point}=[circle,draw=black,fill=white,inner sep=0pt,minimum width=4pt,minimum height=4pt]
			\node (p2)[point] at (1,0) {};
			\node (p3)[point] at (1,1) {};
			\node (p4)[point] at (2,0) {};
			\node (p5)[point] at (2,1) {};
			\node (p6)[point] at (3,0) {};
			\node (p7)[point] at (3,1) {};

			\draw [dashed] (p2) -- (p3);
			\draw [dashed] (p3) -- (p5);
			\draw [dashed] (p2) -- (p5);
			\draw [dashed] (p4) -- (p5);
			\draw [dashed] (p2) -- (p4);
			\draw [dashed] (p4) -- (p7);
			\draw [dashed] (p4) -- (p6);
			\draw [dashed] (p5) -- (p7);
			\draw [dashed] (p6) -- (p7);

			\draw [orange,thick](p2) -- (p3);

			\begin{scope}[on background layer]
			\fill [fill=orange!50](p2.center) -- (p3.center) -- (p5.center);	
			\end{scope}
			
			\filldraw [orange] (p3) circle (1.25pt);						
			\filldraw [orange] (p6) circle (1.25pt);						
		\end{tikzpicture}	
	}\hfill
	\subfloat[\hspace{-0.5cm}\parbox{1in}{$g \land \slreach(g,\slreach(b,r))$}]{\label{subfig:f4}
		\begin{tikzpicture}[scale=1.3]%
			\tikzstyle{point}=[circle,draw=black,fill=white,inner sep=0pt,minimum width=4pt,minimum height=4pt]
			\node (p2)[point] at (1,0) {};
			\node (p3)[point] at (1,1) {};
			\node (p4)[point] at (2,0) {};
			\node (p5)[point] at (2,1) {};
			\node (p6)[point] at (3,0) {};
			\node (p7)[point] at (3,1) {};

			\draw [dashed] (p2) -- (p3);
			\draw [dashed] (p3) -- (p5);
			\draw [dashed] (p2) -- (p5);
			\draw [dashed] (p4) -- (p5);
			\draw [dashed] (p2) -- (p4);
			\draw [dashed] (p4) -- (p7);
			\draw [dashed] (p4) -- (p6);
			\draw [dashed] (p5) -- (p7);
			\draw [dashed] (p6) -- (p7);

			\begin{scope}[on background layer]
			\fill [fill=orange!50](p4.center) -- (p5.center) -- (p7.center);
			\end{scope}

		\end{tikzpicture}	
	}
	\caption{\label{fig:example-model-and-formulas}Examples illustrating the reachability operator on Polyhedra.~\ref{subfig:model}) Polyhedral model C. Circles and colours have the same meaning as in Figure~\ref{fig:example-model-A-and-formulas}. 
	Dashed segments and white points and triangles do not satisfy any atomic proposition.~\ref{subfig:f1}) In orange: Points satisfying $\slreach(g,b)$ applied to the model in~\ref{subfig:model}.
	Note that these include points in the closure of the green area since the first point of the paths passing by green and reaching blue do not need to be green themselves.
	\ref{subfig:f3}) In orange: Points satisfying $grow(r,g)$. 
	Note that only the points corresponding to the red area and one green triangle in the model in~\ref{subfig:model} satisfy this formula. %
	\ref{subfig:f4}) In orange: Points satisfying $g \land \slreach(g,\slreach(b,r))$; by nesting reachability, quite complex spatial formulas may be defined.}
	
\end{figure}

\iflong
Note that we can construct polyhedral models starting from distinct descriptions of $V(q_1), \dots, V(q_l)$.
Consider for each $i$ a region $V(q_i) \subseteq P$ obtained as a union of cells of a simplicial partition $\relint{L}_i$.
By Corollary~\ref{cor:simplicialPartitionRefinement}, we can find a simplicial partition $\relint{K}$ refining each partition $\relint{L}_i$, and thus the triple $\tuple{P, \Ksc, V}$ is a well defined polyhedral model.
\else
Note that the same polyhedron $P$ can be associated with different simplicial complexes: our semantics is not sensitive to such presentational ambiguity in the description of $P$.\footnote{Such ambiguity can be thought of as being similar, in spirit, to the infinitely many different programs that may result in equivalent Kripke frames, in a specification language for classical model checking applications.}
\fi
This is because, although we need to specify $\Ksc$ to spell out the restriction on the range of $V$, $\Ksc$ itself does not play a role in the semantics, as shown in the following proposition.
\begin{lem}\label{lem:invariance-wrt-triangulations}
	Let $\model{X} = \tuple{P, \Ksc, V}$ and $\model{X}' = \tuple{P, 		\Ksc', V}$ be two models sharing the same $P$ and $V$.
	For each $x\in P$ and $\phi$ we have:
		$\model{X}, x \vDash \phi  \iff \model{X}', x \vDash \phi$.
\end{lem}
Therefore, for the sake of readability, we will sometimes indicate a polyhedral model with the notation $\model{X} = \tuple{P, V}$, abstracting from the particular choice of $\Ksc$.
Nevertheless, we require $V$ to range over unions of cells of \emph{some} polyhedral partition, thus restricting the semantics to spatial regions definable in terms of polyhedra.
We will call a simplicial complex $\Ksc$ as in Definition~\ref{def:polyhedralModel} \emph{coherent} with the model $\model{X} = \tuple{P,V}$.

We mentioned that employing polyhedra allows for a \emph{finitary treatment} of the semantics.
The following results are essential to formalize this intuition, which will be further investigated in Section~\ref{sec:polyhedral-model-checking}.

\begin{defi}\label{def:logical-equivalence}
    Let $\model{X} = \tuple{P, V}$ be a polyhedral model.
    \emph{Logical equivalence} $\lequiv$ is the binary relation on $P$ such that $x \lequiv y$ if and only if, for every formula $\phi$: $\model{X}, x \vDash \phi \iff \model{X}, y \vDash \phi$. 
\end{defi}

\begin{lem}\label{lemma:semanticsTriangulation}
	Let $\model{X}$ be a polyhedral model and $\Ksc$ a simplicial complex coherent with $\model{X}$.
	Then for each cell $\relint{\sigma} \in \relint{\Ksc}$ and $x,y \in \relint{\sigma}$ we have $x \lequiv y$.
\end{lem}

\noindent
In particular, for every formula $\phi$, $\sem{\phi}$ is a (finite) union of cells of $\relint{\Ksc}$.

\begin{prop}\label{equivalence-classes-characteristic-formulas}
	Given a polyhedral model $\model{X}$, the relation $\lequiv$ has only finitely many equivalence classes. Furthermore, each equivalence class $C$ has a \emph{characteristic formula} $\phi^C$ such that $\model{X},x \vDash \phi^C \iff x \in C$.
\end{prop}

\noindent The above facts are also useful to prove an interesting feature of polyhedral models, namely that the $\Box$ modality can be considered a derived operator, since it is expressible using $\slreach$. This considerably simplifies proofs.

\begin{thm}\label{thm:near-derived-from-slreach}
    For each formula $\phi$, we have $\model{X},x \vDash \Box \phi \iff \model{X},x \vDash \lnot \slreach(\lnot \phi,\top)$.
\end{thm}

\noindent
Another property of polyhedral models which turns out to be fundamental in this work is that we can restrict our attention to a special class of paths --- rather than arbitrary paths --- to study the reachability operator $\slreach$:
\emph{piecewise linear paths}.

\begin{defi}[PL-path]\label{def:piecewise-linear-path}
	We call a path $\pi: [0,1] \to P$ \emph{piecewise linear} (or simply PL-path) if there exist values $r_0 = 0, r_1,\dots, r_k = 1$ such that for every $i = 0, \dots, k-1$ and $t \in [0,1]$: $\pi\big(  t r_i + (1-t) r_{i+1}  \big) = t \pi(r_i) + (1-t) \pi(r_{i+1})$.
\end{defi}

\begin{figure} %
    \centering
\subfloat[\parbox{1in}{Path 1}]{\label{subfig:PLpath1}
\begin{tikzpicture}[scale=1.3]%
    \tikzstyle{point}=[circle,draw=black,fill=white,inner sep=0pt,minimum width=4pt,minimum height=4pt]
    \node (p0)[point] at (0,0) {};      
    \node (p1)[point] at (0,1) {};
    \node (p2)[point] at (1,0) {};
    \node (p3)[point] at (1,1) {};
    \node (p4)[point] at (2,0) {}; 
    \node (p5)[point] at (2,1) {};
    
    \draw [dashed] (p0) -- (p2);
    \draw [dashed] (p2) -- (p3);
    \draw [dashed] (p3) -- (p5);
    \draw [dashed] (p2) -- (p5);
    \draw [dashed] (p4) -- (p5);
    \draw [dashed] (p2) -- (p4);
    
    \draw [red,thick](p0) -- (p1);
    \draw [red,thick](p1) -- (p3);
    \draw [red,thick](p0) -- (p2);
    \draw [red,thick](p2) -- (p3);
    \draw [red,thick](p0) -- (p3);

    \begin{scope}[on background layer]
    \fill [fill=red!50](p0.center) -- (p1.center) -- (p3.center);       
    \fill [fill=red!50](p0.center) -- (p3.center) -- (p2.center);   
    \fill [fill=green!50](p2.center) -- (p3.center) -- (p5.center); 
    \end{scope}
    
    \filldraw [red] (p0) circle (1.25pt);

    \filldraw [red] (p1) circle (1.25pt);
    \filldraw [red] (p3) circle (1.25pt);

    \node (x) at (0,1) {};
    \node at (0.2,0.9) {x};
    \filldraw [black] (x) circle (1.25pt);
    \draw [black,thick](p1) -- (0.4,0) -- (0.7,0) -- (1.2,0.7) -- (1.7,0.3) -- (2,0.4) -- (p5);     
\end{tikzpicture}   
}\quad\quad%
\subfloat[\hspace{-0.25cm}\parbox{1in}{Path 2}]{\label{subfig:PLpath2}
\begin{tikzpicture}[scale=1.3]%
    \tikzstyle{point}=[circle,draw=black,fill=white,inner sep=0pt,minimum width=4pt,minimum height=4pt]
    \node (p0)[point] at (0,0) {};      
    \node (p1)[point] at (0,1) {};
    \node (p2)[point] at (1,0) {};
    \node (p3)[point] at (1,1) {};
    \node (p4)[point] at (2,0) {}; 
    \node (p5)[point] at (2,1) {};
            
    \draw [dashed] (p0) -- (p2);
    \draw [dashed] (p2) -- (p3);
    \draw [dashed] (p3) -- (p5);
    \draw [dashed] (p2) -- (p5);
    \draw [dashed] (p4) -- (p5);
    \draw [dashed] (p2) -- (p4);

    \draw [red,thick](p0) -- (p1);
    \draw [red,thick](p0) -- (p2);
    \draw [red,thick](p2) -- (p3);
    \draw [red,thick](p0) -- (p3);
    \draw [red,thick](p1) -- (p3);
        
    \begin{scope}[on background layer]
    \fill [fill=red!50](p0.center) -- (p1.center) -- (p3.center);       
    \fill [fill=red!50](p0.center) -- (p3.center) -- (p2.center);   
    \fill [fill=green!50](p2.center) -- (p3.center) -- (p5.center); 
    \end{scope}
    
    \filldraw [red] (p0) circle (1.25pt);

    \filldraw [red] (p1) circle (1.25pt);
    \filldraw [red] (p3) circle (1.25pt);

    \node (y) at (0.3,0.7) {};
    \node at (0.2,0.7) {y};
    \draw [black,thick](0.3,0.7) -- (p3) -- (1.2,0.7) -- (1.5,0.2) -- (p4) -- (1.7,0.7);

\end{tikzpicture}   
}
\caption{\label{fig:PLpaths}
Two examples of piecewise linear paths, one starting in point $x$ and the other in point $y$.}
\end{figure}
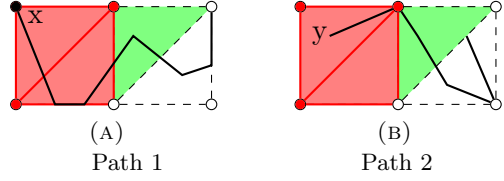

Figure~\ref{fig:PLpaths} shows two examples of piecewise linear paths.
We indicate that a path is piecewise linear with the notation $\pi : [0,1] \toPL P$.
Intuitively, a PL-path is obtained by connecting a finite number of segments and parametrizing them in a suitable way.
\iflong
We shall to  a bit sloppy on properties of such parametrization (e.g., we will often say that the concatenation of two PL-paths is again a PL-path, which requires defining the concatenation and parametrizing its domain), but we trust that all the fine details can be easily filled by the eager reader.
\fi
Even if PL-paths are much simpler than arbitrary paths, when it comes to connectivity in polyhedral models the two classes are interchangeable, as shown in the following lemma.

\begin{lem}\label{lemma:pathToPLPath}
	Let $\Ksc$ be a simplicial complex and $x,y \in \poly{\Ksc}$.
	Then there exists a path in $\poly{\Ksc}$ from $x$ to $y$ if and only if there is a PL-path in $\poly{\Ksc}$ from $x$ to $y$.
\end{lem}

\noindent
Using the previous result, we can give an alternative semantic characterization of the reachability operator, which is relevant for the proofs of decidability of model checking (Section~\ref{sec:polyhedral-model-checking}) and the characterisation of logical equivalence via bisimilarity (Section~\ref{sec:simplicial-bisimilarity}).

\begin{lem}\label{lem:piecewise-linear-semantics} We have:
		$\model{X}, x \vDash \slreach( \phi, \psi )$ if and only if
			there is a PL-path $\pi : [0,1] \toPL P$ such that $\pi(0) = x$ and $\pi((0,1)) \subseteq \sem{\phi}$ and $\pi(1) \in \sem{\psi}$.
\end{lem}

\noindent
We conclude this section by pointing out an interesting property of PL-paths in connection to simplicial complexes.
By definition, PL-paths are obtained by concatenating a finite number of line segments, but in general each segment might traverse different cells of the simplicial complex $\Ksc$.
However we can give an alternative decomposition of these paths so that each portion is fully contained in a unique cell; paths that admit such a decomposition are called \emph{simplicial paths}.

\begin{defi}\label{def:simplicial-path}
    A path $\pi: [0,1] \to P$ is \emph{simplicial} if and only if there is a finite sequence $s_0 =0 < \cdots < s_ k =1$ of values in $[0,1]$ and cells $\relint{\sigma}_1,\dots,\relint{\sigma}_k \in \relint{\Ksc}$ such that, for all $i = 1,\dots,k$, we have $\pi((s_{i-1},s_{i})) \subseteq \widetilde{\sigma_i}$.
\end{defi}

\noindent
Notice that the property of being \emph{simplicial} depends on the simplicial complex $\Ksc$, since a path might be simplicial with respect to a certain triangulation of $P$, but not with respect to another.
For PL-paths, this property holds independently from the simplicial complex considered.

\begin{lem}\label{lem:pl-are-simplicial}
    Any piecewise linear path is simplicial.
\end{lem}

%% file: model-checking.tex
\section{Geometric Model Checking}
\label{sec:polyhedral-model-checking}

Given a polyhedral model $\model{X}$, this section is devoted to identifying a corresponding Kripke-style, finite model $\KM(\model{X})$. Notably, $\KM(\model{X})$ is also a topological model in the sense of~\cite{vBB07} when equipped with the Alexandrov topology, and it is a quotient of $\model{X}$ that preserves and reflects the semantics of each formula.
The goal of this section is to extend the standard Kripkean semantics of modal logic to the language of \SLCS, by defining a suitable semantics for $\slreach$ and by showing that $\model{X}$ and $\KM(\model{X})$ are logically equivalent, in the sense that, for all $x \in \model{X}$ and for all formulas $\Phi$, letting $\relint{\sigma}$ be the only cell such that $x\in\relint{\sigma}$, we have:
$\model{X},x \models \Phi$ if and only if $\KM(\model{X}),\relint{\sigma} \models \Phi$.
To do so, we introduce a suitable notion of path in $\KM(\model{X})$ corresponding to a simplicial path in $\model{X}$.
Model checking on $\model{X}$ can then be carried out using $\KM(\model{X})$.

\begin{defi}\label{def:KrpS}
	Given a polyhedral model $\model{X} = \tuple{ P, \Ksc, V}$, we define the Kripke model $\KM(\model{X}) = \tuple{\relint{\Ksc}, \relint{\face}, \relint{V} }$, where
	\begin{itemize}
		\item $\relint{\Ksc}$ is the simplicial partition of $|\Ksc|$ generated by $\Ksc$, as Defined in Lemma~\ref{lemma:partition},
		\item $\relint{\face} \subseteq \relint{\Ksc} \times \relint{\Ksc}$ with 
		$\relint{\sigma}_1 \; \relint{\face} \; \relint{\sigma}_2$ iff $\sigma_1 \; \face \; \sigma_2$, and
		\item $\relint{\sigma} \in \relint{V}(p)$ iff $\relint{\sigma} \subseteq V(p)$
	\end{itemize}
	where $\face$ is the face relation of the simplicial complex $\Ksc$.
\end{defi}

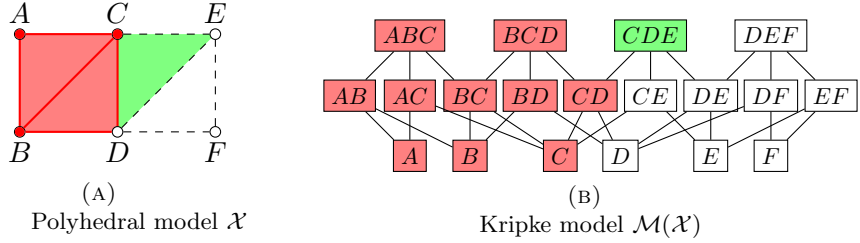
\begin{figure}
\centering
\subfloat[\hspace{-0.7cm}\parbox{1.5in}{Polyhedral model $\model{X}$}]{
	\label{subfig:polyModelX}
	\begin{tikzpicture}[scale=1.3,label distance=-2pt]%
	    \tikzstyle{point}=[circle,draw=black,fill=white,inner sep=0pt,minimum width=4pt,minimum height=4pt]
	    \node (p0)[point,label={270:$B$}] at (0,0) {};
	    	\filldraw [red] (p0) circle (1.25pt);
	    \node (p1)[point,label={ 90:$A$}] at (0,1) {};
	    	\filldraw [red] (p1) circle (1.25pt);
	    \node (p2)[point,label={270:$D$}] at (1,0) {};
	    \node (p3)[point,label={ 90:$C$}] at (1,1) {};
	    	\filldraw [red] (p3) circle (1.25pt);
	    \node (p4)[point,label={270:$F$}] at (2,0) {};
	    \node (p5)[point,label={ 90:$E$}] at (2,1) {};

	    \draw [red   ,thick](p0) -- (p1);
	    \draw [red   ,thick](p0) -- (p2);
	    \draw [red   ,thick](p0) -- (p3);
	    \draw [red   ,thick](p1) -- (p3);
	    \draw [red   ,thick](p2) -- (p3);
	    \draw [dashed      ](p2) -- (p4);
	    \draw [dashed      ](p2) -- (p5);
	    \draw [dashed      ](p3) -- (p5);
	    \draw [dashed      ](p4) -- (p5);
	        
	    \begin{scope}[on background layer]
	    \fill [fill=red!50  ](p0.center) -- (p1.center) -- (p3.center);
	    \fill [fill=red!50  ](p0.center) -- (p3.center) -- (p2.center);
	    \fill [fill=green!50](p2.center) -- (p3.center) -- (p5.center);
	    \end{scope}
	\end{tikzpicture}
}\hspace{2em}
\subfloat[$ $\\Kripke model $\KM(\model{X})$]{\label{subfig:KripkeModelM}
\begin{tikzpicture}[scale=0.8, every node/.style={transform shape}]
    \tikzstyle{kstate}=[rectangle,draw=black,fill=white]
    \node[kstate,fill=red!50  ] (P0) at (  1,0) {$B$};
    \node[kstate,fill=red!50  ] (P1) at (  0,0) {$A$};
    \node[kstate              ] (P2) at (3.5,0) {$D$};
    \node[kstate,fill=red!50  ] (P3) at (2.5,0) {$C$};
    \node[kstate              ] (P4) at (  6,0) {$F$};
    \node[kstate              ] (P5) at (  5,0) {$E$};

    \node[kstate,fill=red!50  ] (E0) at (-1,1) {$AB$};
    \node[kstate,fill=red!50  ] (E1) at ( 2,1) {$BD$};
    \node[kstate,fill=red!50  ] (E2) at ( 1,1) {$BC$};
    \node[kstate,fill=red!50  ] (E3) at ( 0,1) {$AC$};
    \node[kstate,fill=red!50  ] (E4) at ( 3,1) {$CD$};
	\node[kstate              ] (E5) at ( 6,1) {$DF$};
	\node[kstate              ] (E6) at ( 5,1) {$DE$};
	\node[kstate              ] (E7) at ( 4,1) {$CE$};
	\node[kstate              ] (E8) at ( 7,1) {$EF$};

    \node[kstate,fill=red!50  ] (T0) at ( 2,2) {$BCD$};
    \node[kstate,fill=red!50  ] (T1) at ( 0,2) {$ABC$};
    \node[kstate              ] (T2) at ( 6,2) {$DEF$};
    \node[kstate,fill=green!50] (T3) at ( 4,2) {$CDE$};

    \draw (P0) to (E0);
    \draw (P0) to (E1);
    \draw (P0) to (E2);

    \draw (P1) to (E0);
    \draw (P1) to (E3);

    \draw (P2) to (E1);
    \draw (P2) to (E4);
    \draw (P2) to (E5);
    \draw (P2) to (E6);

    \draw (P3) to (E2);
    \draw (P3) to (E3);
    \draw (P3) to (E4);
    \draw (P3) to (E7);

    \draw (P4) to (E5);
    \draw (P4) to (E8);

    \draw (P5) to (E6);
    \draw (P5) to (E7);
    \draw (P5) to (E8);

    \draw (E0) to (T1);
    
    \draw (E1) to (T0);
    
    \draw (E2) to (T0);
    \draw (E2) to (T1);
    
    \draw (E3) to (T1);
    
    \draw (E4) to (T0);
    \draw (E4) to (T3);

	\draw (E5) to (T2);

	\draw (E6) to (T2);
	\draw (E6) to (T3);

	\draw (E7) to (T3);

	\draw (E8) to (T2);

\end{tikzpicture}
}
\caption{\label{fig:kripke-model-running-example}
The polyhedral model $\model{X}$ of Figure~\ref{fig:bisimilarity-running-example} (\ref{subfig:polyModelX}) and its corresponding Kripke model $\KM(\model{X})$ (\ref{subfig:KripkeModelM}).
We indicate a cell by the set of the vertices of the corresponding simplex.
The accessibility relation $\relint{\face}$ is represented via its Hasse diagram (reflexive and transitive edges are omitted).
The atomic propositions $g$ and $r$ are indicated in green and red respectively.
}
\end{figure}

\noindent
Notice that, since $\face$ is reflexive, anti-symmetric and transitive, then so is $\relint{\face}$.
An example of a 2D polyhedral model together with its corresponding Kripke model $\KM(\model{X})$ is depicted in Figure~\ref{fig:kripke-model-running-example} (where reflexive and transitive relations are omitted).
Each cell is identified by the vertices of the corresponding simplexes indicated in the polyhedron on the left in Figure~\ref{fig:kripke-model-running-example}.

Such Kripke models can be depicted as a number of rows.
The nodes in the bottom row represent the zero-dimensional cells (i.e., the vertices), the middle row the one-dimensional ones (i.e., the line segments) and the top-row the two-dimensional cells (i.e., the triangles).
To represent 3D polyhedral models we would have one more row representing the three-dimensional cells (i.e., tetrahedrons), and so on for higher dimensions.
We emphasize that  $\KM(\model{X})$, although disregarding much of the information about a polyhedral model (e.g., the position and size of the simplexes) encodes all the information which is expressible using \SLCS, while being a finite --- thus computationally tractable --- representation of $\model{X}$.

By definition of the relation $\relint{\face}$, we have that for two different cells  $\relint{\sigma}_1$ and $\relint{\sigma}_2$,
$\relint{\sigma}_1 \, \relint{\face} \, \relint{\sigma}_2$ entails that the two cells are spatially adjacent --- more precisely, $\sigma_1$ is part of the \emph{boundary} of $\sigma_2$.
This becomes particularly relevant when studying PL-paths in $P$.
In fact, a PL-path can transit between two different cells $\relint{\sigma}_1$ and $\relint{\sigma}_2$ only if either $\relint{\sigma}_1 \relint{\face} \relint{\sigma}_2$ or $\relint{\sigma}_1 \relint{\facer} \relint{\sigma}_2$ holds.
Based on this intuition, the next definition introduces special paths on these Kripke models corresponding to PL-paths on the polyhedron.

\begin{defi}[$\pm$-path]\label{definition:pmpath}	
	Given a polyhedral model $\model{X}$, with $\KM(\model{X}) = \tuple{ \relint{\Ksc}, \relint{\face}, \relint{V} }$ as in Definition~\ref{def:KrpS}, let $\face^{\pm}$ be the relation $\relint{\face} \cup \relint{\facer}$.
	We say that $\pi: \{0,\dots,k\} \to \relint{\Ksc}$ is a \emph{$\pm$-path} (and we indicate it with $\pi: \{0,\dots,k\} \toPM \relint{\Ksc}$) if $k\geq 2$ and
		$\pi(0) \;\relint{\face}\; \pi(1) \;\face^{\pm}\; \pi(2) \;\face^{\pm}\; \dots \;\face^{\pm}\; \pi(k-1) \;\relint{\facer}\; \pi(k)$.
\end{defi}

\noindent

Intuitively, a $\pm$-path $\pi$ represents the cells traversed by a PL-path $\pi'$ in $P$; $\pi(0)$ is the cell containing  the first point $\pi'(0)$ of the PL-path $\pi'$; $\pi(1)$ is the cell containing the the next portion of PL-path $\pi'$, i.e., the image via the PL-path $\pi'$ of 
an interval of the form $(0,s)$ or $(0,s]$, for some $s$, and so on, up to cell $\pi(k)$ that contains the end-point $\pi'(1)$ of the PL-path $\pi'$. An example of a PL-path and its corresponding $\pm$-path is shown in Figure~\ref{fig:pm-path-example}. 
Notice that the first 
step $\pi(0)\pi(1)$ in the $\pm$-path $\pi$ reflects the move from the starting point in the PL-path to the subsequent segment on the PL-path in polyhedron $P$.  The starting point must obviously be `connected' to this subsequent path segment. This is only possible if the starting point is in the closure of the cell that contains the interval, in other words, $\pi(0) \,\relint{\face}\, \pi(1)$,
that is $\pi(0)$ must be on the boundary of $\Cl_P(\pi(1))$.
Moreover, if $\pi(0) = \pi(1)$, then, depending on the dimension of the cell containing the starting point, the subsequent portion of the path can be in the \emph{same} cell as the starting point. This is for example the case in path $y$ in Figure~\ref{fig:PLpaths}. Obviously, the path segment following the starting point cannot be part of a cell that, in the poset, is strictly below the cell containing the starting point, because then the point and the segment cannot be connected to each other. A similar situation holds for the end-point of the PL-path and the last segment of the path leading to it. In that case the end-point should be in a cell that is in the closure of the cell containing the last segment of the PL-path before reaching the end-point, in other words, $\pi(k-1)\,  \relint{\facer}\, \pi(k)$. Figure~\ref{fig:PLpaths} illustrates this situation in both path $x$ and path $y$ for different situations.

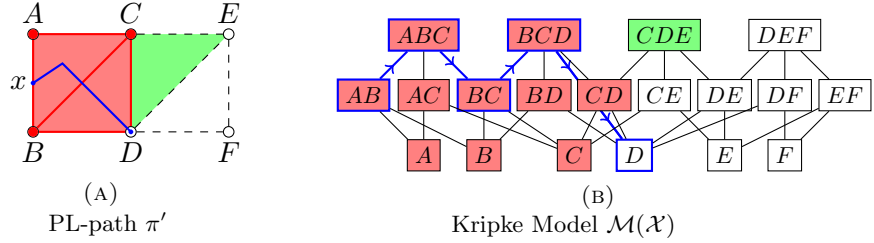
\begin{figure}
\centering
\subfloat[\hspace{-0.5cm}\parbox{1in}{PL-path $\pi'$}]{
	\label{subfig:path-in-polyModelX}
	\begin{tikzpicture}[scale=1.3,label distance=-2pt]
	    \tikzstyle{point}=[circle,draw=black,fill=white,inner sep=0pt,minimum width=4pt,minimum height=4pt]
	    \node (p0)[point,label={270:$B$}] at (0,0) {};
	    	\filldraw [red] (p0) circle (1.25pt);
	    \node (p1)[point,label={ 90:$A$}] at (0,1) {};
	    	\filldraw [red] (p1) circle (1.25pt);
	    \node (p2)[point,label={270:$D$}] at (1,0) {};
	    \node (p3)[point,label={ 90:$C$}] at (1,1) {};
	    	\filldraw [red] (p3) circle (1.25pt);
	    \node (p4)[point,label={270:$F$}] at (2,0) {};
	    \node (p5)[point,label={ 90:$E$}] at (2,1) {};

	    \draw [red   ,thick](p0) -- (p1);
	    \draw [red   ,thick](p0) -- (p2);
	    \draw [red   ,thick](p0) -- (p3);
	    \draw [red   ,thick](p1) -- (p3);
	    \draw [red   ,thick](p2) -- (p3);
	    \draw [dashed      ](p2) -- (p4);
	    \draw [dashed      ](p2) -- (p5);
	    \draw [dashed      ](p3) -- (p5);
	    \draw [dashed      ](p4) -- (p5);
	        
	    \begin{scope}[on background layer]
	    \fill [fill=red!50  ](p0.center) -- (p1.center) -- (p3.center);
	    \fill [fill=red!50  ](p0.center) -- (p3.center) -- (p2.center);
	    \fill [fill=green!50](p2.center) -- (p3.center) -- (p5.center);
	    \end{scope}

	    \node at (-.15,.5) {$x$};
	    \fill[blue] (0,.5) circle (.7pt);
	    \draw [blue,thick](0,.5) -- (.3,.7) -- (1,0);
	    \fill[blue] (1,0) circle (.7pt);

	\end{tikzpicture}
}\hspace{2em}
\subfloat[$  $\\\hspace{-0.25cm}\parbox{1.5in}{$  $Kripke Model $\KM(\model{X})$}]{\label{subfig:path-in-KripkeModelM}
\begin{tikzpicture}[scale=0.8, every node/.style={transform shape}]
    \tikzstyle{kstate}=[rectangle,draw=black,fill=white]
    \tikzset{->-/.style={decoration={
		markings,
		mark=at position #1 with {\arrow{>}}},postaction={decorate}}}
    \node[kstate,fill=red!50  ] (P0) at (  1,0) {$B$};
    \node[kstate,fill=red!50  ] (P1) at (  0,0) {$A$};
    \node[kstate              ,draw=blue,thick] (P2) at (3.5,0) {$D$};
    \node[kstate,fill=red!50  ] (P3) at (2.5,0) {$C$};
    \node[kstate              ] (P4) at (  6,0) {$F$};
    \node[kstate              ] (P5) at (  5,0) {$E$};

    \node[kstate,fill=red!50  ,draw=blue,thick] (E0) at (-1,1) {$AB$};
    \node[kstate,fill=red!50  ] (E1) at ( 2,1) {$BD$};
    \node[kstate,fill=red!50  ,draw=blue,thick] (E2) at ( 1,1) {$BC$};
    \node[kstate,fill=red!50  ] (E3) at ( 0,1) {$AC$};
    \node[kstate,fill=red!50  ] (E4) at ( 3,1) {$CD$};
	\node[kstate              ] (E5) at ( 6,1) {$DF$};
	\node[kstate              ] (E6) at ( 5,1) {$DE$};
	\node[kstate              ] (E7) at ( 4,1) {$CE$};
	\node[kstate              ] (E8) at ( 7,1) {$EF$};

    \node[kstate,fill=red!50  ,draw=blue,thick] (T0) at ( 2,2) {$BCD$};
    \node[kstate,fill=red!50  ,draw=blue,thick] (T1) at ( 0,2) {$ABC$};
    \node[kstate              ] (T2) at ( 6,2) {$DEF$};
    \node[kstate,fill=green!50] (T3) at ( 4,2) {$CDE$};

    \draw (P0) to (E0);
    \draw (P0) to (E1);
    \draw (P0) to (E2);

    \draw (P1) to (E0);
    \draw (P1) to (E3);

    \draw (P2) to (E1);
    \draw (P2) to (E4);
    \draw (P2) to (E5);
    \draw (P2) to (E6);

    \draw (P3) to (E2);
    \draw (P3) to (E3);
    \draw (P3) to (E4);
    \draw (P3) to (E7);

    \draw (P4) to (E5);
    \draw (P4) to (E8);

    \draw (P5) to (E6);
    \draw (P5) to (E7);
    \draw (P5) to (E8);

    \draw[blue,thick,->-=.5] (E0) to (T1);
    
    \draw (E1) to (T0);
    
    \draw[blue,thick,->-=.5] (E2) to (T0);
    \draw[blue,thick,->-=.5] (T1) to (E2);
    
    \draw (E3) to (T1);
    
    \draw (E4) to (T0);
    \draw (E4) to (T3);

	\draw (E5) to (T2);

	\draw (E6) to (T2);
	\draw (E6) to (T3);

	\draw (E7) to (T3);

	\draw (E8) to (T2);

	\begin{scope}[on background layer]
		\draw[blue,thick,->-=.2] (T0) to (P2);
		\draw[blue,thick,->-=.8] (T0) to (P2);
	\end{scope}

\end{tikzpicture}
}
\caption{\label{fig:pm-path-example}On the left, a PL-path $\pi'$ (in blue);
on the right, the corresponding $\pm$-path $\pi$ (again in blue).
Only $\pi'(0)$ belongs to the simplex $AB$ and the path enters immediately into the cell $\relint{ABC}$:
this is possible since $\relint{AB} \,\relint{\face}\, \relint{ABC}$.
Likewise, the path ends with a transition from the cell $\relint{BCD}$ to $\relint{D}$, which is possible since $\relint{BCD} \,\relint{\facer}\, \relint{D}$.
These are exactly the requirements on the first and last steps of a $\pm$-path.
}
\end{figure}

We are now ready to define the formal semantics of \SLCS on $\KM(\model{X})$.

\begin{defi}[\SLCS semantics on $\KM(\model{X})$]\label{def:kripke-semantics}
	Consider $\KM(\model{X}) = \tuple{ \relint{\Ksc}, \relint{\face}, \relint{V} }$.
	Given $\relint{\sigma} \in \relint{\Ksc}$, satisfaction $\KM(\model{X}), \relint{\sigma} \vDash \phi$ over formulas $\phi$ is given by the following inductive clauses, where we let $\semKM{\phi}$ denote the set $\{\relint{\sigma} \in \relint{\Ksc} \,|\, \KM(\model{X}), \relint{\sigma} \vDash \phi\}$:
	\begin{equation*}
	\begin{array}{lcl}
		\KM(\model{X}), \relint{\sigma} \vDash \top
			&\multicolumn{2}{l}{\text{always holds}}  \\
		\KM(\model{X}), \relint{\sigma} \vDash p
			&\iff
			&\relint{\sigma} \in \relint{V}(p) \;\,\text{for } p \in \AP  \\
		\KM(\model{X}), \relint{\sigma} \vDash \neg \phi
			&\iff
			&\KM(\model{X}), \relint{\sigma} \nvDash \phi  \\
		\KM(\model{X}), \relint{\sigma} \vDash \phi \land \psi
			&\iff
			&\KM(\model{X}), \relint{\sigma} \vDash \phi \;\text{and}\; \KM(\model{X}), \relint{\sigma} \vDash \psi  \\
		\KM(\model{X}), \relint{\sigma} \vDash \Box \phi
			&\iff
			&\forall \relint{\tau} \in \relint{\Ksc}.\;  \text{if}\;  \relint{\sigma} \,\relint{\face}\, \relint{\tau}  \;\text{then}\;  \KM(\model{X}), \relint{\tau} \vDash \phi \\
		\KM(\model{X}), \relint{\sigma} \vDash \slreach( \phi, \psi )
			&\iff
			&\text{there exists a $\pm$-path $\pi: \{0,\dots,k\} \toPM \relint{\Ksc}$ such that} \\
			&&\pi(0) = \relint{\sigma},\,
				\pi(k) \in \semKM{\psi}
				\text{and } \pi(\{1,\dots,k-1\}) \subseteq \semKM{\phi}
	\end{array}
	\end{equation*}
\end{defi}

\iflong
\noindent
For $\phi$ a formula, we indicate with $\Pi_\phi \subseteq S$ the set of points satisfying $\phi$, so that the clause for $\slreach$ can be rewrite as follows:
\begin{equation*}
\begin{array}{lcl}
	\KM, s \vDash \slreach( \phi, \psi )
		&\iff
		&\exists \pi: \{0,\dots,k\} \toPM S.\begin{cases}
					\pi(0) = s  \\
					\forall j \in \{1,\dots,k-1\}.\; \pi(j) \in \Pi_\phi  \\
					\pi(k) \in \Pi_\psi
				\end{cases}
\end{array}
\end{equation*}
\fi
\noindent
The clauses for the Boolean operators and for $\Box$ are the standard interpretation of modal formulas on Kripke models (and on topological spaces, via the Alexandrov topology, see~\cite{vBB07}).
In the semantic clause for $\slreach$ we use $\pm$-paths since, as previously pointed out, these paths naturally correspond to PL-paths and Lemma~\ref{lem:piecewise-linear-semantics} allows us to restrict our attention to this class of topological paths.
We can give an example of this correspondence based on Figure~\ref{fig:pm-path-example}:
the PL-path $\pi'$ (on the left) witnesses that $x \in \sem{\slreach(r,\neg (r\vee g))}$, and the corresponding $\pm$-path $\pi$ (on the right) witnesses that $\relint{AB} \in \semKM{\slreach(r,\neg (r\vee g))}$.

The following theorem shows that this correspondence holds for every formula of the logic.

\begin{thm}\label{thm:kripke-model-checking}
	Let $\model{X} = \tuple{P, \Ksc, V}$ be a polyhedral model and  $x$ a point of $P$.
	Let $\sigma \in \Ksc$ be the unique simplex such that $x \in \relint{\sigma}$.
	For every formula $\phi$ of \SLCS we have $\model{X}, x \vDash \phi
	\iff
	\KM(\model{X}), \relint{\sigma} \vDash \phi$.
\end{thm}

\subsection{Geometric Model-checking Algorithm}
We briefly present the main aspects of the geometric model checking algorithm for \SLCS over polyhedra. 
The model checking algorithm takes as inputs a Kripke model $\KM(\model{X})$ of polyhedron $\model{X}$ and an \SLCS formula $\phi$.
The output is the satisfaction set $\sat(\phi)= \{\relint{\sigma} \in \relint{\Ksc} \,|\, \KM(\model{X}), \relint{\sigma} \vDash \phi\} = \semKM{\phi}$ of nodes in the model $\KM(\model{X})$ that correspond to the set of cells in $\model{X}$ that satisfy formula $\phi$, i.e., the algorithm is a \emph{global} model checking algorithm.
The satisfaction set $\sat$ is defined recursively on the structure of \SLCS formulas in the usual way (see for example~\cite{BK08}).
The algorithms for the Boolean operators are straightforward, and thus omitted.
The algorithm for the $\Box$ operator takes a set $\sat(\phi) = \semKM{\phi}$ and computes the corresponding satisfaction set $\sat(\Box \phi) = \semKM{\Box \phi} = \{  \relint{\sigma} \in \sat(\phi) \,|\,  \OUT(\relint{\sigma}) \subseteq \sat(\phi)  \}$, where $\OUT(\relint{\sigma})$ is the out-neighbourhood of $\relint{\sigma}$, i.e., $\OUT(\relint{\sigma}) = \{\relint{\tau} \,|\, \relint{\sigma} \,\relint{\face}\, \relint{\tau}\}$.
We omit a detailed description of the algorithm, since it only involves basic set-theoretic operations.
We give a more detailed description of the algorithm for the spatial reachability operator $\GL$ that, given $\sat(\phi)=\semKM{\phi}$ and $\sat(\psi)=\semKM{\psi}$, computes the satisfaction set $\sat(\GL(\phi,\psi))= \semKM{\GL(\phi,\psi)}$.

The semantics of $\GL$, as of Definition~\ref{def:kripke-semantics}, is computed via a variant of the flooding procedure that was already successfully employed in our previous work on point-based spatial model checking algorithms~\cite{CLLM16,CGLLM15,Gri16}, retaining its asymptotic complexity (in the geometric case linear in the number of nodes in $\KM(\model{X})$). 

The pseudo-code with comments is reported in Algorithm~\ref{fig:pseudo-code}.
In the code, we make use of the additional operators on cells $\OUT$ (previously introduced) and $\IN$, the in-neighbourhood,  defined as $\IN(\relint{\sigma}) = \{\relint{\tau} \,|\, \relint{\sigma} \,\relint{\facer}\, \relint{\tau}\}$.
Both operators are lifted to sets in the obvious way.
Therein, for brevity, we call a $\pm$-path $\pi$ a ``good'' path if it is witnessing the formula 
$\GL(\phi,\psi)$, i.e., the path must map (0,1) to $\semKM{\phi}$ and $\{1\}$ to $\semKM{\psi}$. %
Figure~\ref{fig:flooding} shows a step-by-step example of  an execution of the flooding algorithm described in Algorithm~\ref{fig:pseudo-code} for the property $\gamma(r,g)$, where $r$ denoted red and $g$ denotes green, in the polyhedral model and corresponding Kripke model shown in Figure~\ref{subfig:models}. Recall that the definition of the reachability operator $\gamma$ states that the first point of the paths passing by red and reaching green does not need to be red itself, which explains why also point $B$ is included in the results. 

\begin{figure}
\begin{algo}{\label{fig:pseudo-code}Pseudo-code for model checking the spatial reachability operator.}\\
\begin{lstlisting}[
    escapeinside={(*}{*)},
    basicstyle=\ttfamily\footnotesize,
    breaklines=true,
    postbreak=\mbox{$\hookrightarrow$\space},
    numbers=left,
    xleftmargin=2em,
    framexleftmargin=1.5em
]
    Input:  (*$\semKM{\phi}, \semKM{\psi}$*) sets of nodes of (*$\KM(\model{X})$*) satisfying (*$\phi$*) resp. (*$\psi$*)
    Output: (*$\semKM{\GL(\phi,\psi)}$*) satisfaction set of nodes in (*$\KM(\model{X})$*)
    
    // frontier: the points queued for the next iteration. 
    // Initialized with the next to last points of a good path.
    let frontier = (*$\semKM{\phi} \cap\OUT(\semKM{\psi})$*)
    
    // flooded: all the points that are in the middle of a good path.
    let flooded = frontier
    
    while frontier(*$\ne\emptyset$*):
        let (*$\sigma$*) = frontier.pop()
        for every (*$\tau \in \IN(\sigma) \cup \OUT(\sigma)$*):
            if (*$\tau \notin$*)flooded and (*$\tau \in \semKM{\phi}$*):
                frontier.add((*$\tau$*))
                flooded.add((*$\tau$*))
    
    // result: the starting points of a good path.
    let result = (*$\IN($*)flooded(*$)$*)
    
    return result
\end{lstlisting}
\end{algo}
\end{figure}

\begin{figure}[p]
\centering
\subfloat[\label{subfig:models}A polyhedral model $\model{X}$ and its corresponding Kripke model $\KM(\model{X})$.]{
\begin{tikzpicture}[scale=1.3,label distance=-2pt]
    \tikzstyle{point}=[circle,draw=black,fill=white,inner sep=0pt,minimum width=4pt,minimum height=4pt]

    \node (p0)[point,label={270:$B$}] at (0,0) {};
    \node (p1)[point,label={ 90:$A$}] at (0,1) {};
    	\filldraw [red  ] (p1) circle (1.25pt);
    \node (p2)[point,label={270:$D$}] at (1,0) {};
    	\filldraw [green] (p2) circle (1.25pt);
    \node (p3)[point,label={ 90:$C$}] at (1,1) {};
    	\filldraw [green] (p3) circle (1.25pt);

    \draw [red   ,thick](p0) -- (p1);
    \draw [dashed      ](p0) -- (p2);
    \draw [red   ,thick](p0) -- (p3);
    \draw [red   ,thick](p1) -- (p3);
    \draw [dashed      ](p2) -- (p3);
        
    \begin{scope}[on background layer]
    	\fill [fill=red!50  ](p0.center) -- (p1.center) -- (p3.center);
    \end{scope}
\end{tikzpicture}
\hspace{2em}
\begin{tikzpicture}[scale=0.8, every node/.style={transform shape}]
    \tikzstyle{kstate}=[rectangle,draw=black,fill=white]
    \tikzset{->-/.style={decoration={
		markings,
		mark=at position #1 with {\arrow{>}}},postaction={decorate}}}

    \node[kstate              ] (P0) at ( .5,0) {$B$};
    \node[kstate,fill=red!50  ] (P1) at (-.5,0) {$A$};
    \node[kstate,fill=green!50] (P2) at (2.5,0) {$D$};
    \node[kstate,fill=green!50] (P3) at (1.5,0) {$C$};
    
    \node[kstate,fill=red!50  ] (E0) at (-1,1) {$AB$};
    \node[kstate              ] (E1) at ( 2,1) {$BD$};
    \node[kstate,fill=red!50  ] (E2) at ( 1,1) {$BC$};
    \node[kstate,fill=red!50  ] (E3) at ( 0,1) {$AC$};
    \node[kstate              ] (E4) at ( 3,1) {$CD$};
	
    \node[kstate              ] (T0) at ( 2,2) {$BCD$};
    \node[kstate,fill=red!50  ] (T1) at ( 0,2) {$ABC$};

    \draw (P0) to (E0);
    \draw (P0) to (E1);
    \draw (P0) to (E2);

    \draw (P1) to (E0);
    \draw (P1) to (E3);

    \draw (P2) to (E1);
    \draw (P2) to (E4);

    \draw (P3) to (E2);
    \draw (P3) to (E3);
    \draw (P3) to (E4);

    \draw (E0) to (T1);
    
    \draw (E1) to (T0);
    
    \draw (E2) to (T0);
    \draw (E2) to (T1);
    
    \draw (E3) to (T1);
    
    \draw (E4) to (T0);
\end{tikzpicture}
}\\[2em]
\subfloat[\label{subfig:firstComp}Initialization of \texttt{frontier} (lines 4-5).]{
\begin{tikzpicture}[scale=1.3,label distance=-2pt]
    \tikzstyle{point}=[circle,draw=black,fill=white,inner sep=0pt,minimum width=4pt,minimum height=4pt]

    \node (p0)[point,label={270:$B$}] at (0,0) {};
    \node (p1)[point,label={ 90:$A$}] at (0,1) {};
    	\filldraw [red  ] (p1) circle (1.25pt);
    \node (p2)[point,label={270:$D$}] at (1,0) {};
    	\filldraw [green] (p2) circle (1.25pt);
    \node (p3)[point,label={ 90:$C$}] at (1,1) {};
    	\filldraw [green] (p3) circle (1.25pt);

    \draw [red   ,thick](p0) -- (p1); %
    \draw [dashed      ](p0) -- (p2); %
    \draw [blue  ,thick](p0) -- (p3); %
    \draw [blue  ,thick](p1) -- (p3); %
    \draw [dashed      ](p2) -- (p3); %
        
    \begin{scope}[on background layer]
    	\fill [fill=blue!30 ](p0.center) -- (p1.center) -- (p3.center);
    \end{scope}
\end{tikzpicture}
\hspace{1em}
\begin{tikzpicture}[scale=0.8, every node/.style={transform shape}]
    \tikzstyle{kstate}=[rectangle,draw=black,fill=white]
    \tikzset{->-/.style={decoration={
		markings,
		mark=at position #1 with {\arrow{>}}},postaction={decorate}}}

    \node[kstate              ] (P0) at ( .5,0) {$B$};
    \node[kstate,fill=red!50  ] (P1) at (-.5,0) {$A$};
    \node[kstate,fill=green!50] (P2) at (2.5,0) {$D$};
    \node[kstate,fill=green!50] (P3) at (1.5,0) {$C$};
    
    \node[kstate,fill=red!50  ] (E0) at (-1,1) {$AB$};
    \node[kstate              ] (E1) at ( 2,1) {$BD$};
    \node[kstate,fill=blue,text=white] (E2) at ( 1,1) {$BC$};
    \node[kstate,fill=blue,text=white] (E3) at ( 0,1) {$AC$};
    \node[kstate              ] (E4) at ( 3,1) {$CD$};
	
    \node[kstate              ] (T0) at ( 2,2) {$BCD$};
    \node[kstate,fill=blue,text=white] (T1) at ( 0,2) {$ABC$};

    \draw (P0) to (E0);
    \draw (P0) to (E1);
    \draw (P0) to (E2);

    \draw (P1) to (E0);
    \draw (P1) to (E3);

    \draw (P2) to (E1);
    \draw (P2) to (E4);

    \draw (P3) to (E2);
    \draw (P3) to (E3);
    \draw (P3) to (E4);

    \draw (E0) to (T1);
    
    \draw (E1) to (T0);
    
    \draw (E2) to (T0);
    \draw (E2) to (T1);
    
    \draw (E3) to (T1);
    
    \draw (E4) to (T0);
\end{tikzpicture}
\hspace{1em}
$\begin{array}{ll}
	{\color{blue}\mathtt{frontier}}
		&= \{AC,BC,ABC\}  \phantom{A,AB}
\end{array}$
}\\[1em]
\subfloat[\label{subfig:secondComp}Initialization of \texttt{flooded} and first flooding iteration (lines 7-15).]{
\begin{tikzpicture}[scale=1.3,label distance=-2pt]
    \tikzstyle{point}=[circle,draw=black,fill=white,inner sep=0pt,minimum width=4pt,minimum height=4pt]

    \node (p0)[point,label={270:$B$}] at (0,0) {};
    \node (p1)[point,label={ 90:$A$}] at (0,1) {};
    	\filldraw [blue ] (p1) circle (1.25pt);
    \node (p2)[point,label={270:$D$}] at (1,0) {};
    	\filldraw [green] (p2) circle (1.25pt);
    \node (p3)[point,label={ 90:$C$}] at (1,1) {};
    	\filldraw [green] (p3) circle (1.25pt);

    \draw [blue  ,thick](p0) -- (p1); %
    \draw [dashed      ](p0) -- (p2); %
    \draw [orange,thick](p0) -- (p3); %
    \draw [orange,thick](p1) -- (p3); %
    \draw [dashed      ](p2) -- (p3); %
        
    \begin{scope}[on background layer]
    	\fill [fill=orange!30 ](p0.center) -- (p1.center) -- (p3.center);
    \end{scope}
\end{tikzpicture}
\hspace{1em}
\begin{tikzpicture}[scale=0.8, every node/.style={transform shape}]
    \tikzstyle{kstate}=[rectangle,draw=black,fill=white]
    \tikzset{->-/.style={decoration={
		markings,
		mark=at position #1 with {\arrow{>}}},postaction={decorate}}}

    \node[kstate              ] (P0) at ( .5,0) {$B$};
    \node[kstate,fill=blue,text=white] (P1) at (-.5,0) {$A$};
    \node[kstate,fill=green!50] (P2) at (2.5,0) {$D$};
    \node[kstate,fill=green!50] (P3) at (1.5,0) {$C$};
    
    \node[kstate,fill=blue,text=white] (E0) at (-1,1) {$AB$};
    \node[kstate              ] (E1) at ( 2,1) {$BD$};
    \node[kstate,fill=orange] (E2) at ( 1,1) {$BC$};
    \node[kstate,fill=orange] (E3) at ( 0,1) {$AC$};
    \node[kstate              ] (E4) at ( 3,1) {$CD$};
	
    \node[kstate              ] (T0) at ( 2,2) {$BCD$};
    \node[kstate,fill=orange] (T1) at ( 0,2) {$ABC$};

    \draw (P0) to (E0);
    \draw (P0) to (E1);
    \draw (P0) to (E2);

    \draw (P1) to (E0);
    \draw (P1) to (E3);

    \draw (P2) to (E1);
    \draw (P2) to (E4);

    \draw (P3) to (E2);
    \draw (P3) to (E3);
    \draw (P3) to (E4);

    \draw (E0) to (T1);
    
    \draw (E1) to (T0);
    
    \draw (E2) to (T0);
    \draw (E2) to (T1);
    
    \draw (E3) to (T1);
    
    \draw (E4) to (T0);
\end{tikzpicture}
\hspace{1em}
$\begin{array}{ll}
	{\color{blue}\mathtt{frontier}}
		&= \{A,AB\} \\
	{\color{orange}\mathtt{flooded}}
		&= \{AC,BC,ABC\}  \phantom{A,AB}
\end{array}$
}\\[1em]
\subfloat[\label{subfig:thirdComp}Second flooding iteration (lines 10-15).]{
\begin{tikzpicture}[scale=1.3,label distance=-2pt]
    \tikzstyle{point}=[circle,draw=black,fill=white,inner sep=0pt,minimum width=4pt,minimum height=4pt]

    \node (p0)[point,label={270:$B$}] at (0,0) {};
    \node (p1)[point,label={ 90:$A$}] at (0,1) {};
    	\filldraw [orange] (p1) circle (1.25pt);
    \node (p2)[point,label={270:$D$}] at (1,0) {};
    	\filldraw [green] (p2) circle (1.25pt);
    \node (p3)[point,label={ 90:$C$}] at (1,1) {};
    	\filldraw [green] (p3) circle (1.25pt);

    \draw [orange ,thick](p0) -- (p1); %
    \draw [dashed      ](p0) -- (p2); %
    \draw [orange,thick](p0) -- (p3); %
    \draw [orange,thick](p1) -- (p3); %
    \draw [dashed      ](p2) -- (p3); %
        
    \begin{scope}[on background layer]
    	\fill [fill=orange!30 ](p0.center) -- (p1.center) -- (p3.center);
    \end{scope}
\end{tikzpicture}
\hspace{1em}
\begin{tikzpicture}[scale=0.8, every node/.style={transform shape}]
    \tikzstyle{kstate}=[rectangle,draw=black,fill=white]
    \tikzset{->-/.style={decoration={
		markings,
		mark=at position #1 with {\arrow{>}}},postaction={decorate}}}

    \node[kstate              ] (P0) at ( .5,0) {$B$};
    \node[kstate,fill=orange] (P1) at (-.5,0) {$A$};
    \node[kstate,fill=green!50] (P2) at (2.5,0) {$D$};
    \node[kstate,fill=green!50] (P3) at (1.5,0) {$C$};
    
    \node[kstate,fill=orange] (E0) at (-1,1) {$AB$};
    \node[kstate              ] (E1) at ( 2,1) {$BD$};
    \node[kstate,fill=orange] (E2) at ( 1,1) {$BC$};
    \node[kstate,fill=orange] (E3) at ( 0,1) {$AC$};
    \node[kstate              ] (E4) at ( 3,1) {$CD$};
	
    \node[kstate              ] (T0) at ( 2,2) {$BCD$};
    \node[kstate,fill=orange] (T1) at ( 0,2) {$ABC$};

    \draw (P0) to (E0);
    \draw (P0) to (E1);
    \draw (P0) to (E2);

    \draw (P1) to (E0);
    \draw (P1) to (E3);

    \draw (P2) to (E1);
    \draw (P2) to (E4);

    \draw (P3) to (E2);
    \draw (P3) to (E3);
    \draw (P3) to (E4);

    \draw (E0) to (T1);
    
    \draw (E1) to (T0);
    
    \draw (E2) to (T0);
    \draw (E2) to (T1);
    
    \draw (E3) to (T1);
    
    \draw (E4) to (T0);
\end{tikzpicture}
\hspace{1em}
$\begin{array}{ll}
	{\color{blue}\mathtt{frontier}}
		&= \emptyset \\
	{\color{orange}\mathtt{flooded}}
		&= \{A,AB,AC,BC,ABC\}
\end{array}$
}\\[1em]
\subfloat[\label{subfig:fourthComp}Computation of \texttt{result} (lines 17-18).]{
\begin{tikzpicture}[scale=1.3,label distance=-2pt]
    \tikzstyle{point}=[circle,draw=black,fill=white,inner sep=0pt,minimum width=4pt,minimum height=4pt]

    \node (p0)[point,label={270:$B$}] at (0,0) {};
    	\filldraw [purple] (p0) circle (1.25pt);
    \node (p1)[point,label={ 90:$A$}] at (0,1) {};
    	\filldraw [purple ] (p1) circle (1.25pt);
    \node (p2)[point,label={270:$D$}] at (1,0) {};
    	\filldraw [green] (p2) circle (1.25pt);
    \node (p3)[point,label={ 90:$C$}] at (1,1) {};
    	\filldraw [purple] (p3) circle (1.25pt);

    \draw [purple  ,thick](p0) -- (p1); %
    \draw [dashed      ](p0) -- (p2); %
    \draw [purple  ,thick](p0) -- (p3); %
    \draw [purple  ,thick](p1) -- (p3); %
    \draw [dashed      ](p2) -- (p3); %
        
    \begin{scope}[on background layer]
    	\fill [fill=purple!30 ](p0.center) -- (p1.center) -- (p3.center);
    \end{scope}
\end{tikzpicture}
\hspace{1em}
\begin{tikzpicture}[scale=0.8, every node/.style={transform shape}]
    \tikzstyle{kstate}=[rectangle,draw=black,fill=white]
    \tikzset{->-/.style={decoration={
		markings,
		mark=at position #1 with {\arrow{>}}},postaction={decorate}}}

    \node[kstate,fill=purple,text=white] (P0) at ( .5,0) {$B$};
    \node[kstate,fill=purple,text=white] (P1) at (-.5,0) {$A$};
    \node[kstate,fill=green!50] (P2) at (2.5,0) {$D$};
    \node[kstate,fill=purple,text=white] (P3) at (1.5,0) {$C$};
    
    \node[kstate,fill=purple,text=white] (E0) at (-1,1) {$AB$};
    \node[kstate              ] (E1) at ( 2,1) {$BD$};
    \node[kstate,fill=purple,text=white] (E2) at ( 1,1) {$BC$};
    \node[kstate,fill=purple,text=white] (E3) at ( 0,1) {$AC$};
    \node[kstate              ] (E4) at ( 3,1) {$CD$};
	
    \node[kstate              ] (T0) at ( 2,2) {$BCD$};
    \node[kstate,fill=purple,text=white] (T1) at ( 0,2) {$ABC$};

    \draw (P0) to (E0);
    \draw (P0) to (E1);
    \draw (P0) to (E2);

    \draw (P1) to (E0);
    \draw (P1) to (E3);

    \draw (P2) to (E1);
    \draw (P2) to (E4);

    \draw (P3) to (E2);
    \draw (P3) to (E3);
    \draw (P3) to (E4);

    \draw (E0) to (T1);
    
    \draw (E1) to (T0);
    
    \draw (E2) to (T0);
    \draw (E2) to (T1);
    
    \draw (E3) to (T1);
    
    \draw (E4) to (T0);
\end{tikzpicture}
\hspace{1em}
$\begin{array}{ll}
	\mathtt{flooded}
		&= \{A,AB,AC,BC,ABC\} \\
	{\color{purple}\mathtt{result}}
		&= \{A,B,C,AB, \\
		& \hspace{5em} AC,BC,ABC\}
\end{array}$
}

\caption{\label{fig:flooding}An example of application of the algorithm described in Algorithm~\ref{fig:pseudo-code}.
Subfigure~\ref{subfig:models} depicts a polyhedral model $\model{X}$ (on the left) and its corresponding Kripke model $\KM(\model{X})$.
Subfigures from~\ref{subfig:firstComp} %
to~\ref{subfig:fourthComp}  depict the steps performed on $\KM(\model{X})$ to compute the satisfaction set $\sat(\gamma(r,g))$, with reference to the pseudo-code in Algorithm~\ref{fig:pseudo-code}.
For each subfigure the value of the relevant variables at the end of the computation step are indicated and illustrated in corresponding colours in both models (blue for \texttt{frontier}, orange for \texttt{flooded} and purple for \texttt{result}).
}
\end{figure}
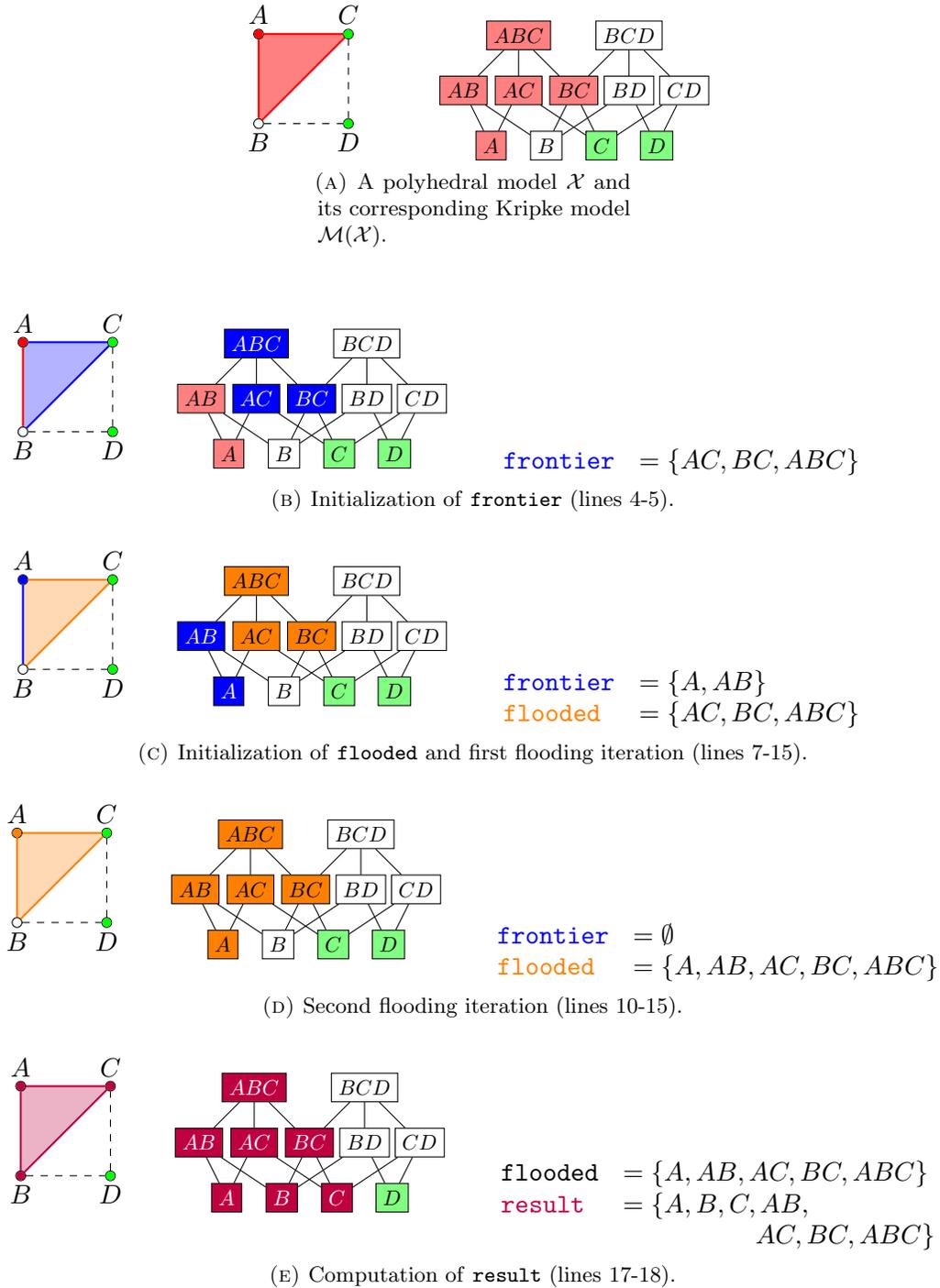

We include a concise proof of the correctness of the algorithm.

\begin{proof}[Correctness, sketch]
To be consistent with the comments in the pseudo-code, we keep calling a $\pm$-path $\pi:\{0,\dots,k\} \to \relint{\Ksc}$ witnessing the satisfaction of $\GL(\phi,\psi)$ a ``good'' path.

First, notice that we can divide a good path into three parts:
the initial point $\pi(0)$, the points on the central segment $\pi(\{1,\dots,k-1\})$ all satisfying $\phi$ and the final point $\pi(k)$ satisfying $\psi$.
To compute the set $\semKM{\GL(\phi,\psi)}$ we work ``backwards'' from an initial set of nodes that satisfy $\phi$ and that have an outgoing edge to a node satisfying $\psi$.
So, first, we compute the set $C := \semKM{\phi} \cap \OUT(\semKM{\psi})$ (the cells corresponding to nodes with index $\pi(k-1)$ on some good path $\pi$). Then we use a standard flooding procedure to build the set $D$ of the nodes of the graph that are connected to $C$ via a non-directed path, i.e., abstracting from the direction of the edges, passing only through nodes satisfying $\semKM{\phi}$ (these correspond to the cells of the form $\pi(j)$, with $j=1,\ldots,k-1$, for any good path $\pi$ in the polyhedron). Finally we compute the set $\semKM{\GL(\phi,\psi)} = \mbox{${\IN}(D)$}$ (the cells corresponding to the initial nodes of good paths, i.e., $\pi(0)$ for any good path $\pi$).
\end{proof}

\subsection{Computational Complexity}

\noindent
In what follows, we indicate by $n$ the number of cells of $\relint{\Ksc}$ and by $d$ the dimension of $\Ksc$, that is, the \emph{maximum} dimension of a simplex $\sigma$ in $\Ksc$.
Therefore, the number of nodes of the encoding $\KM(\model{X})$ is $n$.
Each simplex $\sigma$ having dimension $d_{\sigma} \leq d$ has $d_{\sigma}+1$ vertices.
Furthermore, in the graph encoding $\relint{\Ksc}$, the in-neighbourhood of each $\relint{\sigma} \in \relint{K}$ is precisely the set of the proper faces of $\sigma$, having cardinality $2^{d_{\sigma}+1}-1$. 
By this, the total number of edges of $\KM(\model{X})$ is at most $n\cdot(2^{d+1} - 1)$.
We let $N$ to be the size of this encoding, that is, the sum of the number of nodes and edges, which is at most equal to $n + n\cdot(2^{d+1} - 1) = n\cdot2^{d+1}$.

Note that the complexity of our encoding grows exponentially in $d$ if $d$ is not fixed.
The design space for algorithms that scale better with $d$, possibly exploiting specialised data structures (see, e.g., the recent work~\cite{Boissonnat2014}) will be explored in future research, depending on the considered use case.
In the case under consideration (i.e., 3D meshes) we have the fixed dimension $d=3$.
Therefore the ``exponential'' contribution of $d$ to the computational complexity becomes a constant. In other words, when the dimension $d$ is fixed, as in 3D meshes, the size of $\KM(\model{X})$ is of order $\mathcal{O}(n)$.

Note that the flooding procedure (Algorithm~\ref{fig:pseudo-code}) has linear
computational complexity in the number of nodes and edges of $\KM(\model{X})$, that is, $N$.
The computation of the Boolean operators and of the $\Box$-operator are also
linear in $N$. As in~\cite{CLLM16}, since each subformula is
checked independently from the others,  the asymptotic computational complexity
of the model checking algorithm is of order $\mathcal{O}(N\cdot h)$, where $h$
is the cardinality of the set of subformulas of the \SLCS
formula to be checked.

As remarked before, \textbf{once the dimension $d$ of the image is fixed} the exponential contribution of $d$ to the computational complexity is negligible and therefore the total complexity of the spatial model checking algorithm is of order $\mathcal{O}(n \cdot h)$.

Finally, we briefly address the complexity of the encoding.
In the current prototype implementation of the model-checking algorithm (see Section~\ref{sec:polylogica} for further details), the input polyhedron is described by a list of $n$ simplexes with maximum dimension $d$, each one being represented by a list of vertices.
To compute the Kripke frame of Definition~\ref{def:KrpS} from this description, an explicit enumeration is performed of the subsets of each simplex, incrementally building the arrays of out- and in-neighbourhoods.
This results in a time complexity in $\mathcal{O}(N)$, which becomes $\mathcal{O}(n)$ once the dimension $d$ is fixed.
Therefore, for $d$ fixed, the total complexity (encoding plus model checking) is in $\mathcal{O}(n \cdot h)$.

%% file: tool.tex
\section{\PolyLogicA:\@ a Model Checker for Polyhedra}
\label{sec:polylogica}

Based on the theory and the model checking algorithm presented in Section~\ref{sec:polyhedral-model-checking}, we developed the prototype model checker \PolyLogicA: a \emph{Polyhedral Logic-based Analysis tool}.
The model checker is implemented in the functional language FSharp.\footnote{See \url{https://www.fsharp.org}}
\PolyLogicA is  Free and Open Source Software, distributed under the Apache 2.0 license.\footnote{The tool is currently available in a branch of the main \VoxLogicA repository, see \url{https://github.com/vincenzoml/VoxLogicA}.}

Currently \PolyLogicA represents a polyhedral model $\model{X}$ through an explicit encoding of the Kripke model $\KM(\model{X})$ according to Definition~\ref{def:KrpS}.
The encoding is stored as a graph having the cells as nodes and with the covering relation of $\relint{\face}$ as the edge relation.
The current implementation stores the out-neighbourhood and the in-neighbourhood of each node $\sigma$ in two separate arrays, allowing access in constant time to these sets.

\subsection{Functional description}

A \PolyLogicA specification consists of a text file that can make use of four commands: \texttt{let}, for declaring functions and constants; \texttt{import},  for importing libraries of such declarations; \texttt{load}, to specify the file to be loaded as a model; \texttt{save}, to specify the logic formulas that need to be computed, and saved, possibly making use of previous \texttt{let} declarations.

Models are required to be based on a fixed simplicial complex. The given simplicial complex constitutes the subspace of $\R^d$ that is explored by the tool. Thus, one needs to explicitly include, in the input model, the ``environment'' 
in which the objects live.
As an example, consider the simplicial complex depicted in Figure~\ref{fig:example-model-A-and-formulas}. The semantics of the formula $\phi = g \lor r \lor b$ is indeed the set of points $S$ that belong to the whole coloured triangle (made up of four smaller ones). 

The file describing the input model uses a custom \texttt{json}-based\footnote{See \url{https://www.json.org/}.} format. The information contained in the file consists of: a list $p$ of $d$-dimensional vectors, denoting the coordinates of the 0-cells of the polyhedron;
a list of atomic proposition identifiers; a list of simplexes. Each simplex is specified by the list of the indexes of its vertices in $p$, and its specification also contains the list of atomic propositions holding at the cell corresponding to the simplex.%

Logic formulas are just a concrete syntax for \SLCS\ as presented in Section~\ref{sec:slcs}.
Currently, \PolyLogicA does not implement additional extra-logical operators (contrary to \VoxLogicA, which also implements imaging primitives).
\PolyLogicA is in spirit a \emph{global, explicit-state} model checker, that is, the set of simplexes satisfying a given formula is computed and returned at once.
The output of the model checker is a list, in \texttt{json} format,
containing an element for each formula $\phi$ that the specification requires to be checked. Each element of such list contains in turn a list representing the truth values of $\phi$ at each point (cell) of the input model. Finally, a simple 3D, web-based visualizer has been implemented along with the prototype (see the screenshot in Figure~\ref{fig:example-visualizer}), which will be refined in future work.

\subsection{Implementation Details}\label{sec:implementation} The current version of \PolyLogicA
has been implemented sharing part of the code base for point-based spatial model checking with \VoxLogicA.
\PolyLogicA inherits from its parent tool the multithreaded, memoizing computation engine, and the parser for the input language of the tool.
Basically, after expanding \texttt{let} bindings, each formula is converted into a \emph{directed acyclic graph} where nodes are \emph{tasks}, and arcs are \emph{dependencies}. Each task is a basic logical primitive, to be applied to specific arguments. Task $A$ depends upon task $B$ if and only if the result of $B$ is an argument of $A$. The implementation guarantees that, while being constructed, such data structure is kept \emph{minimal} in the sense that the same primitive on the same arguments will never be added twice to the graph, thus implementing \emph{memoization} at a syntactic level, without the need to hash possibly large sets of states at run-time. After having been constructed, the task graph is executed in parallel as much as possible, exploiting the available CPU cores.

\subsection{A Visualiser for \PolyLogicA}

In addition to the polyhedra model checker we developed \PolyVisualiser, a prototype polyhedra visualiser. The user interface of the visualiser is shown in the screenshots in Figure~\ref{fig:the-cube}. The visualiser is an interactive tool with which polyhedra can be explored from different perspectives by means of the mouse or an other pointing device. The tool allows to zoom in and out on the visualised model, to translate it rigidly and to turn it around a fixed point.
The main purpose of the visualiser is to inspect the result of spatial properties that were checked using PolyLogicA. From a pop-up menu, shown on the right of the screenshot, a property can be selected and from further menu items one can select the way in which the results of the evaluation of the property will be visualised. For example, in the screenshot in Figure~\ref{fig:the-cube} the cells that satisfy the selected property are shown in green and the cells that do \emph{not} satisfy the property are shown in red. Also the degree of transparency of the colours can be set, thus facilitating the inspection of the interior of 3D objects. Cells that are points or line segments or triangles are coloured in a direct way. To facilitate inspection of the solid components, being them covered by triangles, tetrahedra are reduced in size and coloured without transparency. Moreover, the relative size of tetrahedra can be manipulated real-time through the appropriate control panel.

In the following sections \PolyVisualiser is used to illustrate two examples of realistic size. In those examples, rather than using the colour green for cells in the polyhedron that satisfy the selected formula and the colour red for those cells that do not satisfy the formula, we will use different levels of transparency while retaining the original colour of the cells. Cells that satisfy the selected property are shown in their original colour but made less transparent (i.e., more opaque) than cells that do not satisfy the property. This way the contours of the complete model are preserved and may serve as a visual reference framework for the interpretation of the 3D spatial model checking results which pop-out as opaque objects as shown in Figure~\ref{fig:3DmazeMC}.

\begin{figure}
\centering
	\subfloat[]{\label{subfig:p}
		\includegraphics[valign=c,height=5em]{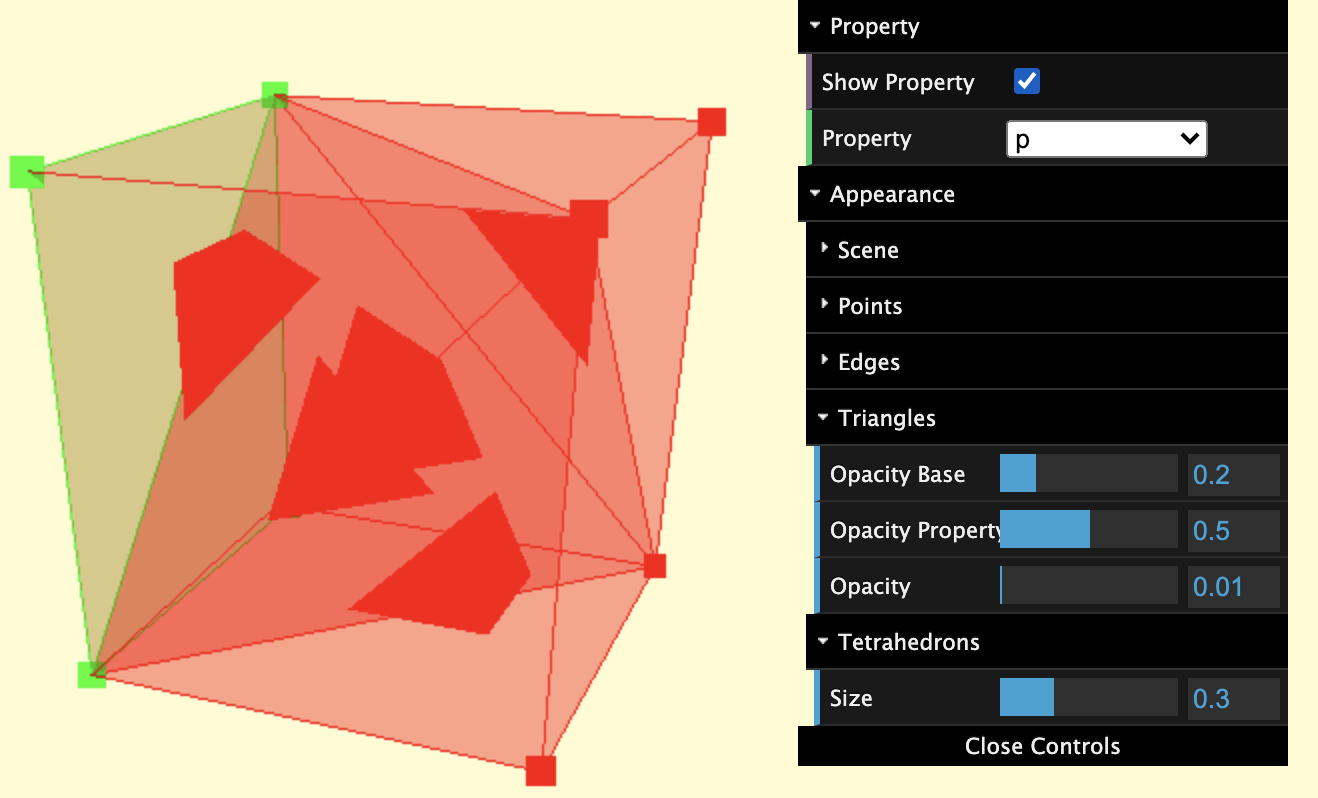}
	}
	\subfloat[]{\label{subfig:q}
		\includegraphics[valign=c,height=5em]{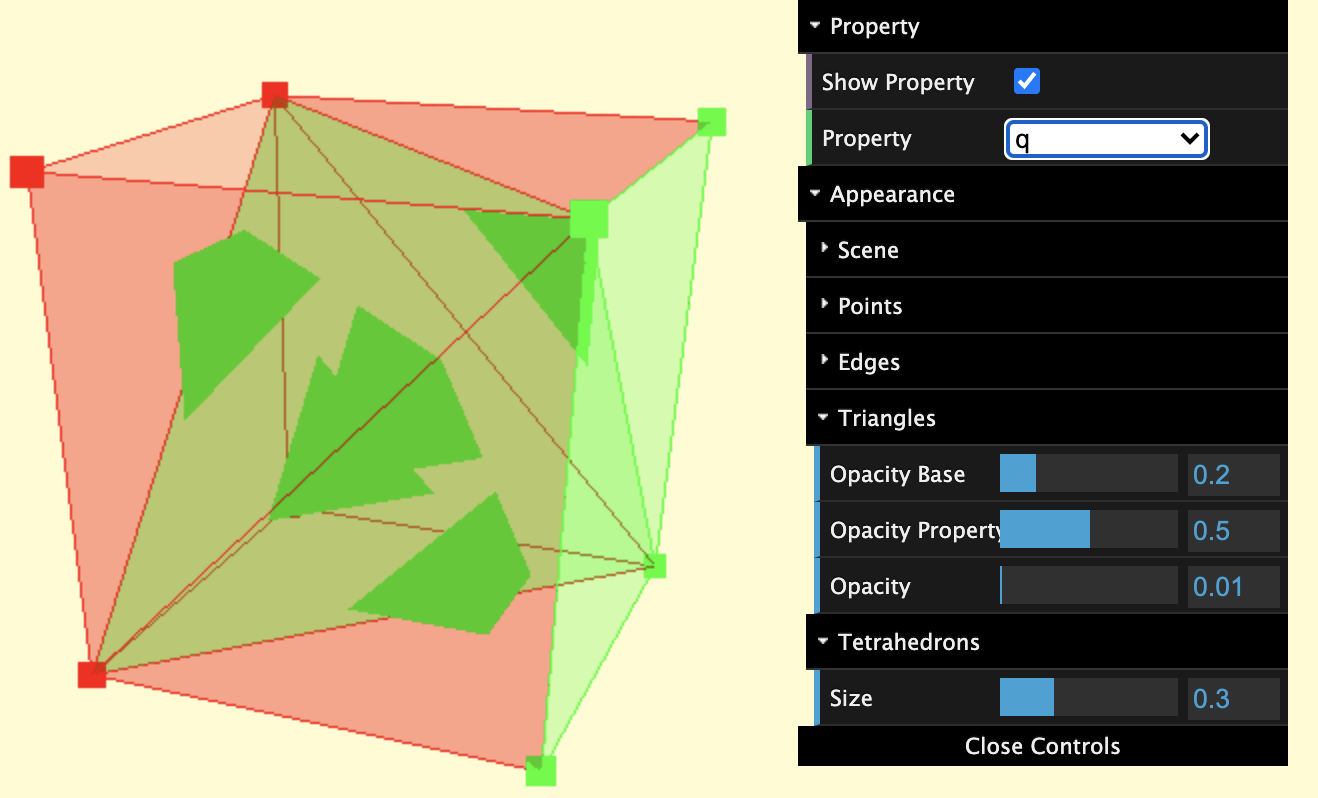}
	}\\
	\subfloat[]{\label{subfig:grw_pq}
		\includegraphics[valign=c,height=5em]{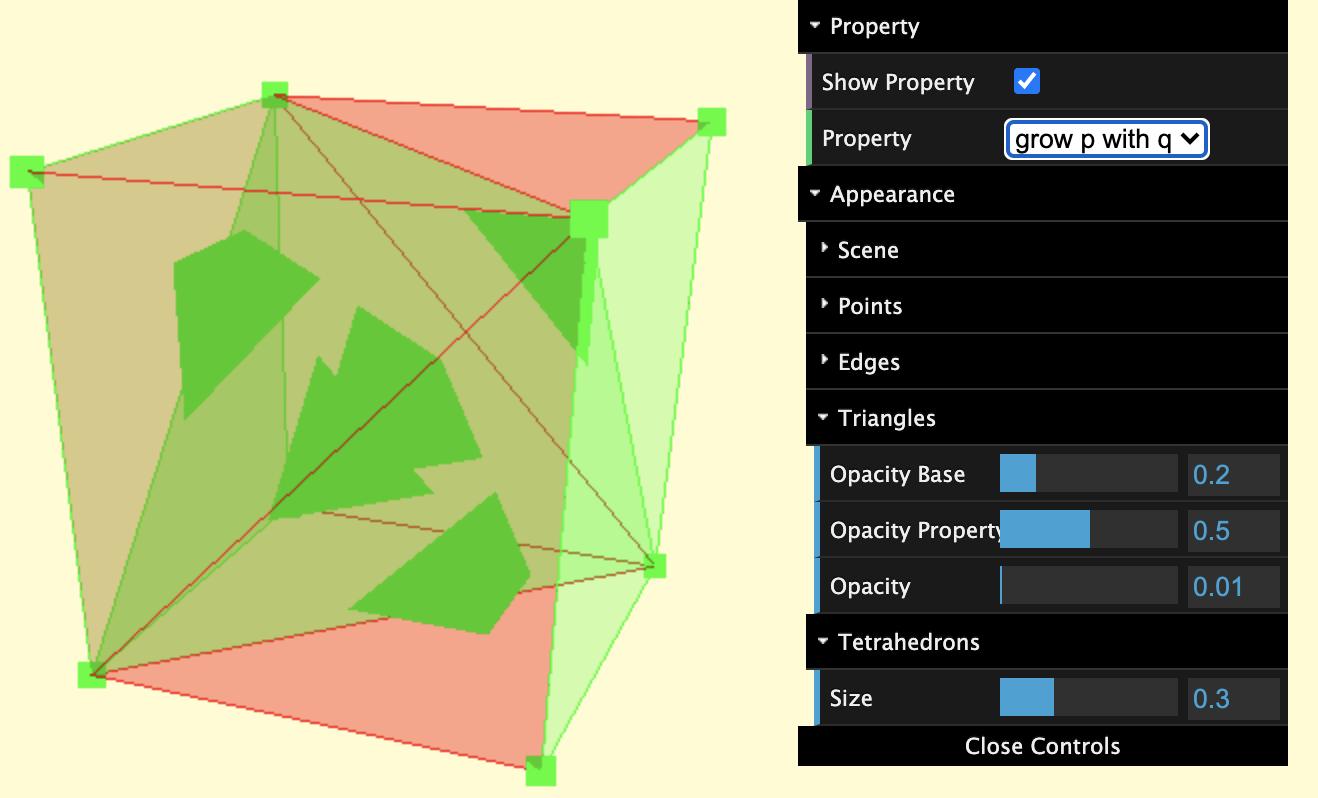}
	}
\caption{\label{fig:the-cube} Visualiser for \PolyLogicA: Model of a simple cube showing cells satisfying the atomic propositions ``p'' and ``q'' in green ((\ref{subfig:p}) and (\ref{subfig:q}), respectively) and model checking result (in green) of cells satisfying ``grow p with q'' (\ref{subfig:grw_pq}).}
\end{figure}

\subsection{Proof of Concept Experiments}
In order to obtain a first impression of the actual feasibility of our (first-in-its-kind) geometric model checking approach we present two examples: one involving the analysis of a 3D maze (generated by the authors) and one involving the analysis of an existing mesh model of the anatomy of the human body. Both examples are analysed with PolyLogicA and we show screenshots of the PolyVisualiser to illustrate the results. 
The results are also made available online\footnote{\url{http://ggrilletti2.scienceontheweb.net/polyVisualizer/polyVisualizer_static_maze.html}} for inspection in an interactive way.

\subsubsection{A 3D Maze}

The first example concerns a 3D maze.\footnote{The maze has been generated by the authors.} The maze, shown in Figure~\ref{subfig:green_rooms}, consists of ``rooms'' that are connected by ``corridors''. The rooms come in four colours: white, black, green and one room in red. The green rooms are all situated at the outer boundary of the maze and represent the surroundings of the maze that can be reached via an exit. The white, black and red rooms, and related corridors, are situated inside the cube and form the maze itself. Figure~\ref{subfig:black_or_white_rooms} shows all the white and black rooms. Figure~\ref{subfig:red_rooms} shows the red room situated at the centre of the maze. All corridors between rooms are dark grey. The valid paths through the maze should only pass by white rooms (and related corridors) to reach a green room without passing by black rooms or corridors that connect to black rooms.  

Typical spatial queries or properties that are of interest for such a maze are:
\begin{enumerate}
\item[Q1:] Which are the white rooms and connecting corridors from which one can reach a green room (i.e., an exit) without having to pass by a black room? Show those white rooms, the related corridors and the green exit.
\item[Q2:] Which are the white rooms and connecting corridors from which both a green room (exit) and a red room can be reached, without having to pass by black rooms (and related corridors)? Show those white rooms, the related corridors, the green exit and the red room.
\item[Q3:] Which are the white rooms (and related corridors) from which it is \emph{not} possible to reach a green room without having to pass by a black room? 
\end{enumerate} 

Figure~\ref{spec:3Dmaze} shows the specification of the above properties in a syntax close to the spatial logic introduced in Section~\ref{sec:slcs}. This syntax is shared with that of \VoxLogicA (see~\cite{CBLM21} for a tutorial) and is mostly self-explanatory. 
The boolean operators are \texttt{|} (or), \texttt{\&} (and), \texttt{!} (not). The spatial reachability operator $\gamma$ is denoted by the keyword  \texttt{through}. The surround operator \texttt{sur} is implicitly defined in terms of the operator \texttt{through} according to its definition provided in Section~\ref{sec:slcs}. In the first line of Figure~\ref{spec:3Dmaze}, the model file is loaded. In lines 4 to 7 atomic propositions are defined for the simplicial cells of the various coloured rooms. In line 8, the atomic proposition of cells belonging to corridors is defined.
These definitions make direct reference to the information stored in the elements of the model file, in \texttt{json}  format, representing the cells of the polyhedron model. In particular, \texttt{ap} denotes an atomic proposition associated with such elements.

In line 11, a simple application of these atomic propositions is defined. The property \texttt{blackOrWhite} is satisfied by all cells of black rooms and those of white rooms. The result of this property can be saved  (as shown in line 44) and used by \PolyVisualiser to visualise the outcome. The saved file essentially consists of a list of booleans indicating for each cell whether it satisfies the property. The result is shown in a screenshot of the visualiser in Figure~\ref{subfig:black_or_white_rooms} where the cells that satisfy the property are opaque and all the others almost transparent. Similarly, in Figure~\ref{subfig:red_rooms} the cells that satisfy the atomic proposition \texttt{red} are shown.

Lines 14 to 20 define properties to distinguish different kinds of corridors. In particular, corridors that connect white rooms (\texttt{corridorWW}), those between white and green rooms (\texttt{corridorWG}), those between white and red rooms (\texttt{corridorWR}) and those between white and black rooms (\texttt{corridorWB}). These will turn useful in the specification of the three properties Q1, Q2 and Q3. A point satisfies \texttt{corridorWW} if it (lays in a cell that) belongs to a corridor (i.e., a cell  satisfying \texttt{corridor}) and from which only white rooms can be reached via the corridor itself (i.e., only cells satisfying \texttt{white} but neither \texttt{green}, nor \texttt{black}, nor \texttt{red} can be reached from cells of the corridor).

In lines 24 and 27 property Q1 is specified as \texttt{whiteToGreen}. It is expressed in terms of the \texttt{through}-operator ($\gamma$). It says that we want all (cells of) white rooms, corridors between white rooms and corridors between white and green rooms by which one can reach a green room, i.e., without passing by black rooms or corridors connected to black rooms. If we also want to include the green room that is being reached in the set of results we have to add it, as is done in line 27 in \texttt{connWG}, leading to all cells satisfying \texttt{whiteToGreen} and the (cells of) green rooms from which cells satisfying \texttt{whiteToGreen} can be reached. Cells satisfying \texttt{connWG} are shown in their original opaque colour in Figure~\ref{subfig:connectionWhiteGreen}. Cells that do not satisfy the property are shown in a transparent manner. Note that, in general, the \texttt{through(x,y)}-operator would also admit satisfaction by cells that are adjacent to those satisfying \texttt{x} in that case, as shown in the example of Figure~\ref{subfig:f1}. In the model of the maze rooms and corridors do not have such adjacent cells. However, in general one could consider the use of \texttt{(x | y) \& through(x,y)} to make sure that only cells satisfying at least \texttt{x} or \texttt{y} are part of the result. The choice for the particular definition of \texttt{through} (and $\gamma$) has been motivated by the fact that it poses minimal restrictions. More restricted variants can be easily defined based on \texttt{through}. This would not be the case if one opts for a less basic definition of \texttt{through}.
\begin{figure}
$
\begin{array}{l c l}
\mbox{corridorWW} &\equiv & \gamma(C,W) \land \neg\gamma(C, G \lor B \lor R)\\
\mbox{corridorWG} &\equiv & \gamma(C,W) \land \gamma(C, G)\\
\mbox{corridorWR} &\equiv & \gamma(C,W) \land \gamma(C, R)\\
\mbox{corridorWB} &\equiv & \gamma(C,W) \land \gamma(C, B)\\
\mbox{whiteToGreen} &\equiv & \gamma((W \lor \mbox{corridorWW} \lor \mbox{corridorWG}),G)\\
\mbox{Q1} &\equiv &  \mbox{whiteToGreen} \lor \gamma(G,\mbox{whiteToGreen})\\
\mbox{Q2} &\equiv & \gamma((\mbox{Q1} \lor \mbox{corridorWR}), R) \lor \gamma((R \lor \mbox{corridorWR}),\mbox{Q1})\\
\mbox{Q3} &\equiv & (W \lor \mbox{corridorWW}) \land \neg \mbox{whiteToGreen}.
\end{array}
$
\caption{\SLCS\ formulas expressing properties Q1, Q2 and Q3; atomic proposition letters $G,W,B,R,C$ are assumed given and their meaning is the obvious one ($C$ for ``corridor'', $G$ for green and similarly for the other colours).\label{fig:SLCSQ1Q2Q3}}
\end{figure}

In line 30 property Q2 is specified as \texttt{connRWG}. In this case we are looking for white rooms and related corridors from which both a green room and a red one can be reached, and we also want to include in the result such red and green rooms. The specification makes use of \texttt{connWG} and allows passing by corridors connecting white rooms to red rooms \texttt{corridorWR}. The result is shown in Figure~\ref{subfig:whiteConnectsRedGreen}.

In line 34 property Q3 is specified as \texttt{whiteNoGreen}. In this case we are looking for white rooms and related corridors from which it is \emph{not} possible to reach a green room without passing by a black room. This can be specified as white rooms and corridors between white rooms that are not satisfying \texttt{whiteToGreen} (i.e., property Q1 above). The same property could also be specified in an alternative way making use of the surround operator \texttt{sur} defined in lines 37-38. In that case these cells are characterised as those belonging to white rooms and corridors between white rooms that are completely surrounded by corridors leading only to black rooms in specification \texttt{whiteSblack} in line 41. The result is shown in Figure~\ref{subfig:no_exit_rooms}.

Finally, in line 44 an example is shown of saving a particular result, in this case that of property \texttt{blackOrWhite}. Saving the other results can be performed in a similar manner and is not shown in the specification.

For completeness,  in Figure~\ref{fig:SLCSQ1Q2Q3}, 
the \SLCS\ formulas expressing properties Q1, Q2 and Q3 are shown.

\begin{figure}
\begin{lstlisting}[
	basicstyle=\ttfamily\footnotesize,
	escapeinside={(*}{*)},
        breaklines=true,
        postbreak=\mbox{$\hookrightarrow$\space},
        xleftmargin=2em,
        framexleftmargin=1.5em,
	numbers=left]
load model = "mazeModel.json"

// Atomic propositions for rooms and corridors
let green      = ap("G")
let white      = ap("W")
let black      = ap("B")
let red        = ap("R")
let corridor   = ap("corridor")

// Black or White rooms
let blackOrWhite = black | white

// Corridors: white-to-white, white-to-green, white-to-red, white-to-black:
let corridorWW = through( corridor, white ) &
                !through( corridor, green | black | red )
let corridorWG = through( corridor, white ) &  
                 through( corridor, green )
let corridorWR = through( corridor, white ) & 
                 through( corridor, red )
let corridorWB = through( corridor, white ) & 
                 through( corridor, black )

// Q1: White rooms from which a green room can be reached not passing by black rooms
let whiteToGreen = through((white | corridorWW | corridorWG), green)

// as Q1 but including the green room that is reached
let connWG = whiteToGreen | through(green,whiteToGreen) 

// Q2: White rooms from which both a red and a green room can be reached not passing by black rooms
let connRWG = through((connWG | corridorWR), red) | 
                 through((red | corridorWR), connWG)

// Q3: White rooms with no path to green rooms and their connecting corridors
let whiteNoGreen = (white | corridorWW) & !whiteToGreen

// Surround operator in terms of reach in the polyhedra setting
let reach(x,y) = x | through(y,x)
let sur(x,y)	 = x & !reach(!(x | y),!y) 

// Q3(alternative): White rooms and their connecting corridors surrounded only by corridors to black rooms
let whiteSblack = sur((white | corridorWW), corridorWB) 

// Save the result for property blackOrWhite
save "blackOrWhite" blackOrWhite

\end{lstlisting}
\caption{\label{spec:3Dmaze} PolyLogicA specification of some spatial properties for the 3D maze model.}
\end{figure}

\begin{figure}
\centering
	\subfloat[]{\label{subfig:green_rooms}
		\includegraphics[valign=c,height=6em]{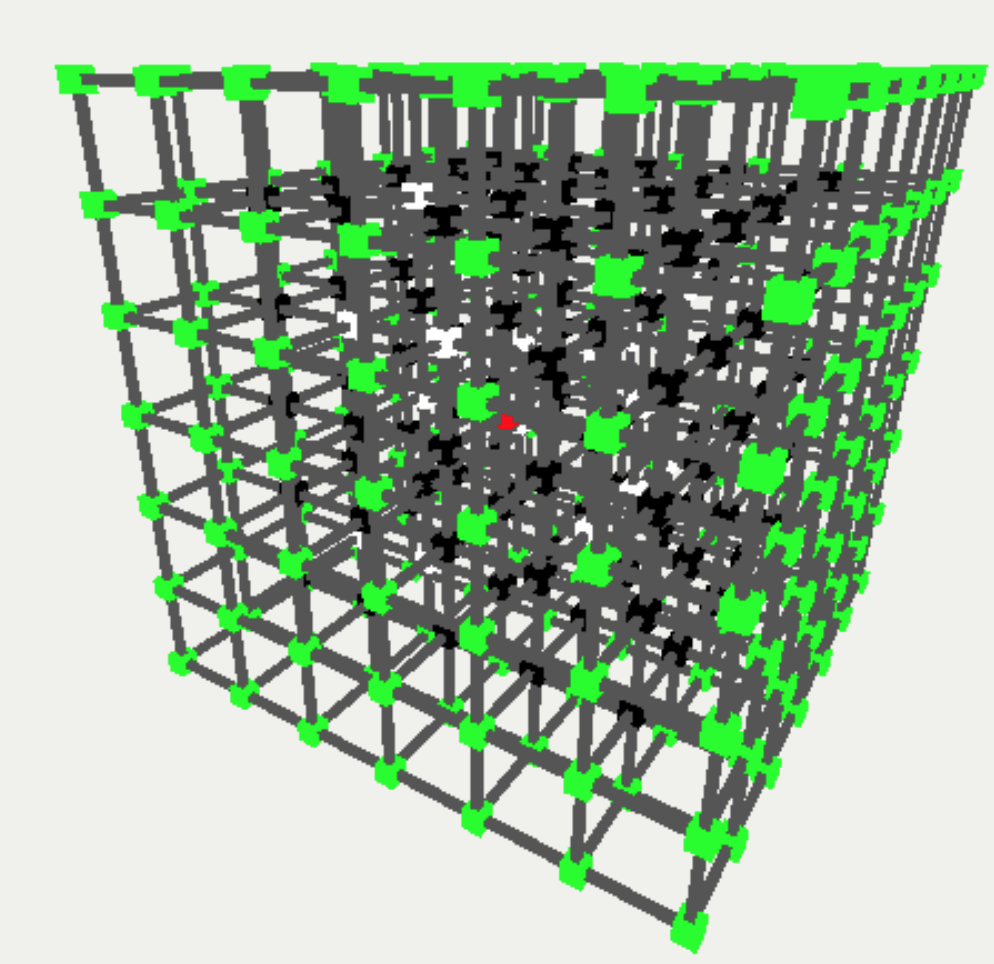}
	}
	\subfloat[]{\label{subfig:black_or_white_rooms}
		\includegraphics[valign=c,height=6em]{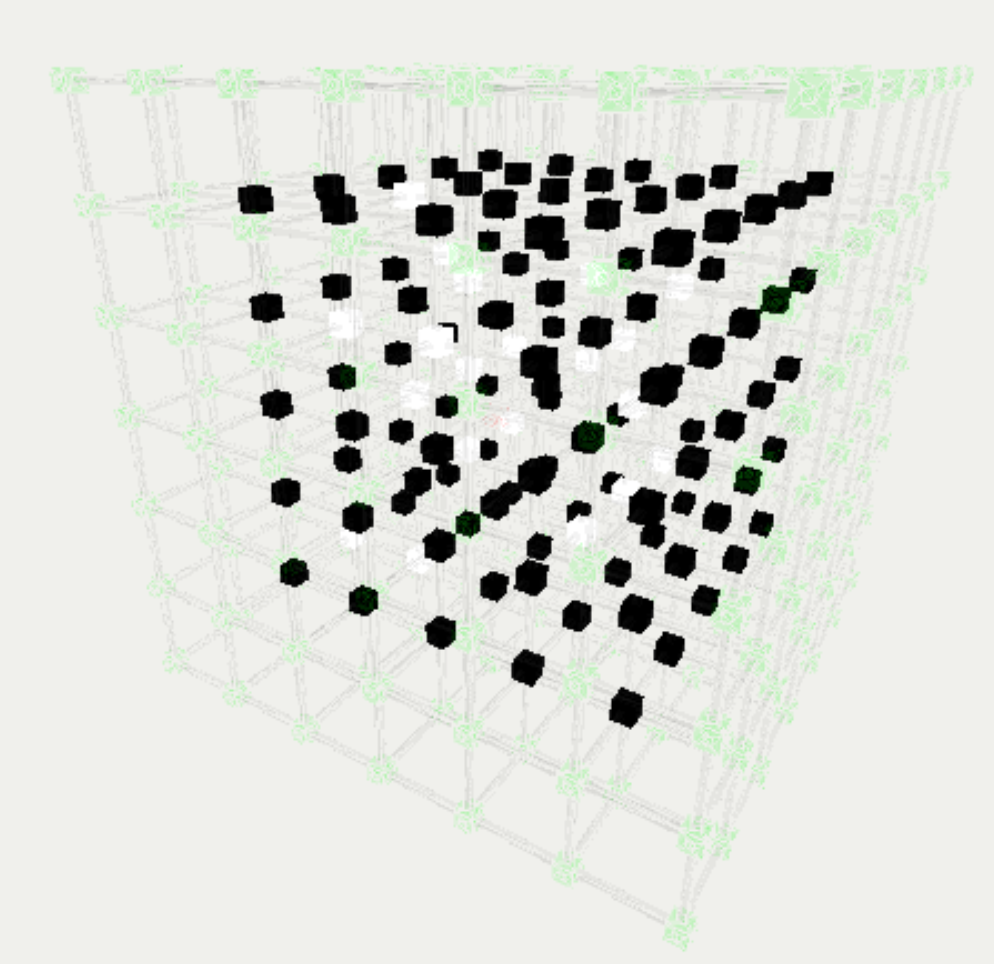}
	}
	\subfloat[]{\label{subfig:red_rooms}
		\includegraphics[valign=c,height=6em]{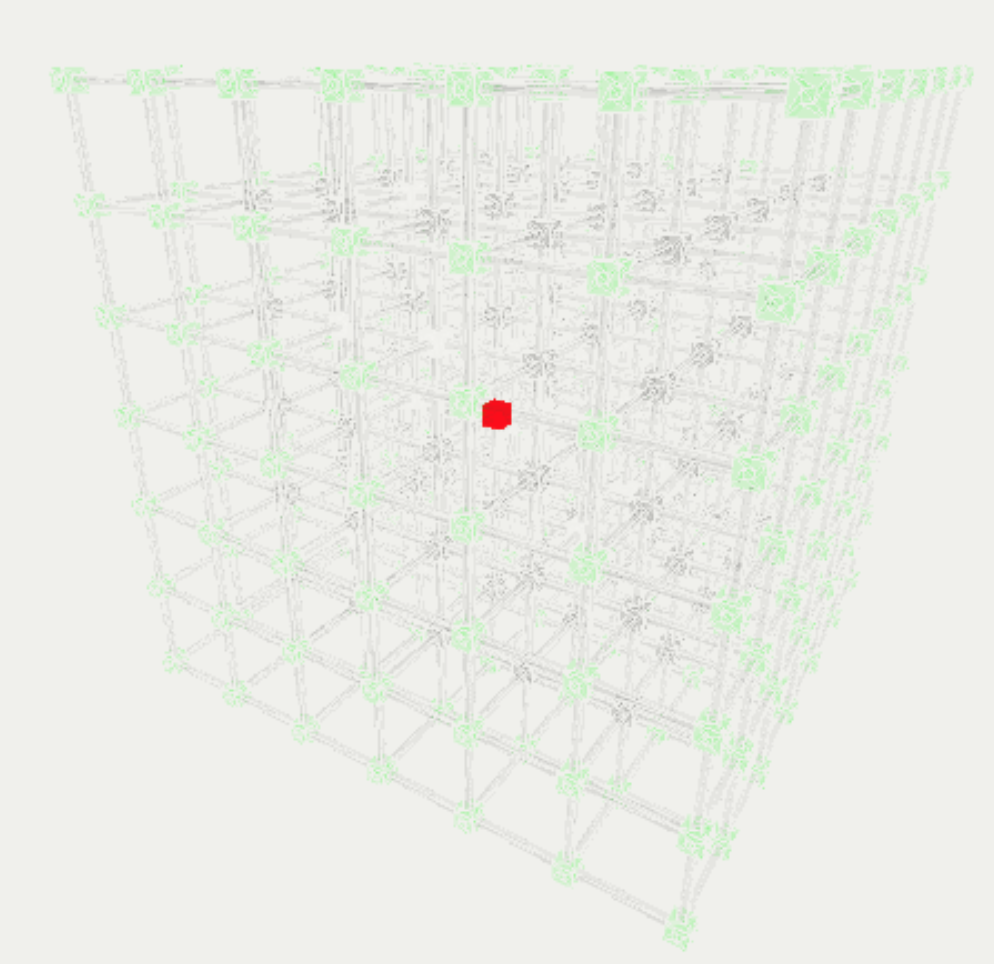}
	}
\caption{\label{fig:3Dmaze} 3D maze~(\ref{subfig:green_rooms}), black and white~(\ref{subfig:black_or_white_rooms}) and red rooms~(\ref{subfig:red_rooms}) in the 3D maze.}
\end{figure}

\begin{figure}
\centering
	\subfloat[Q1]{\label{subfig:connectionWhiteGreen}
		\includegraphics[valign=c,height=6em]{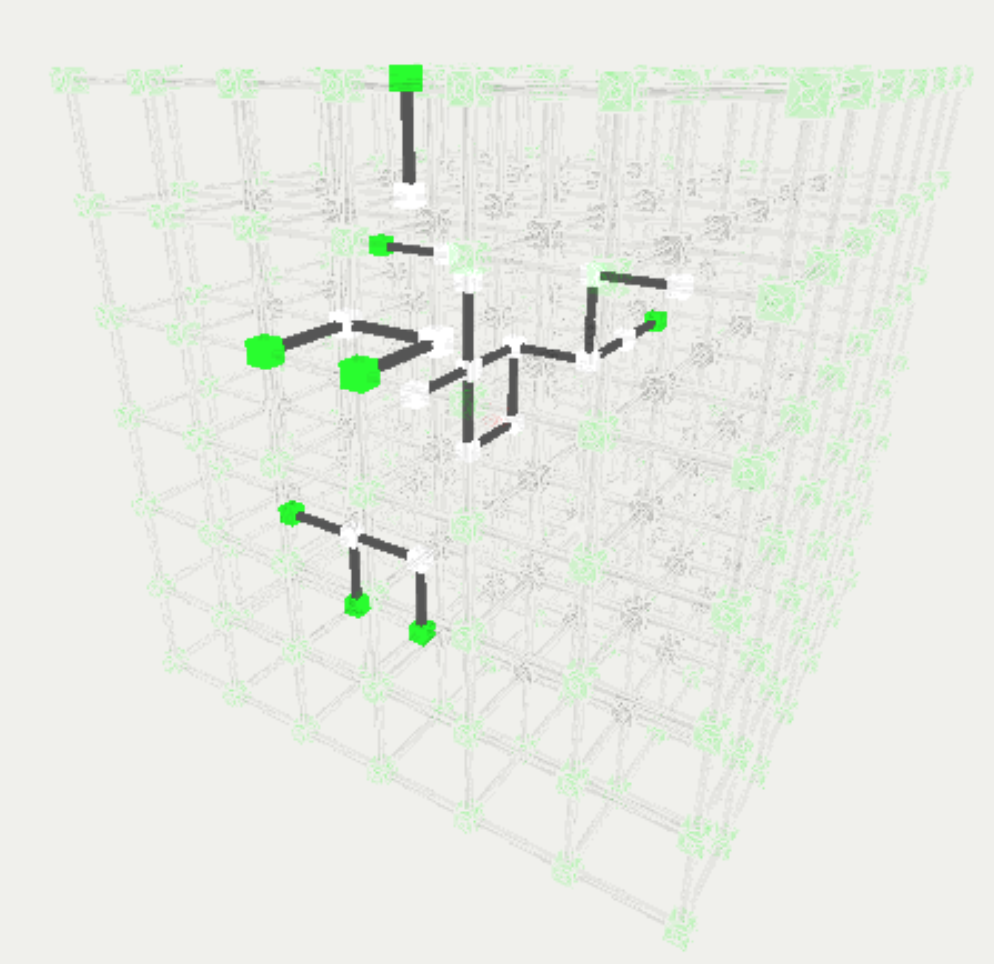}
	}
	\subfloat[Q2]{\label{subfig:whiteConnectsRedGreen}
		\includegraphics[valign=c,height=6em]{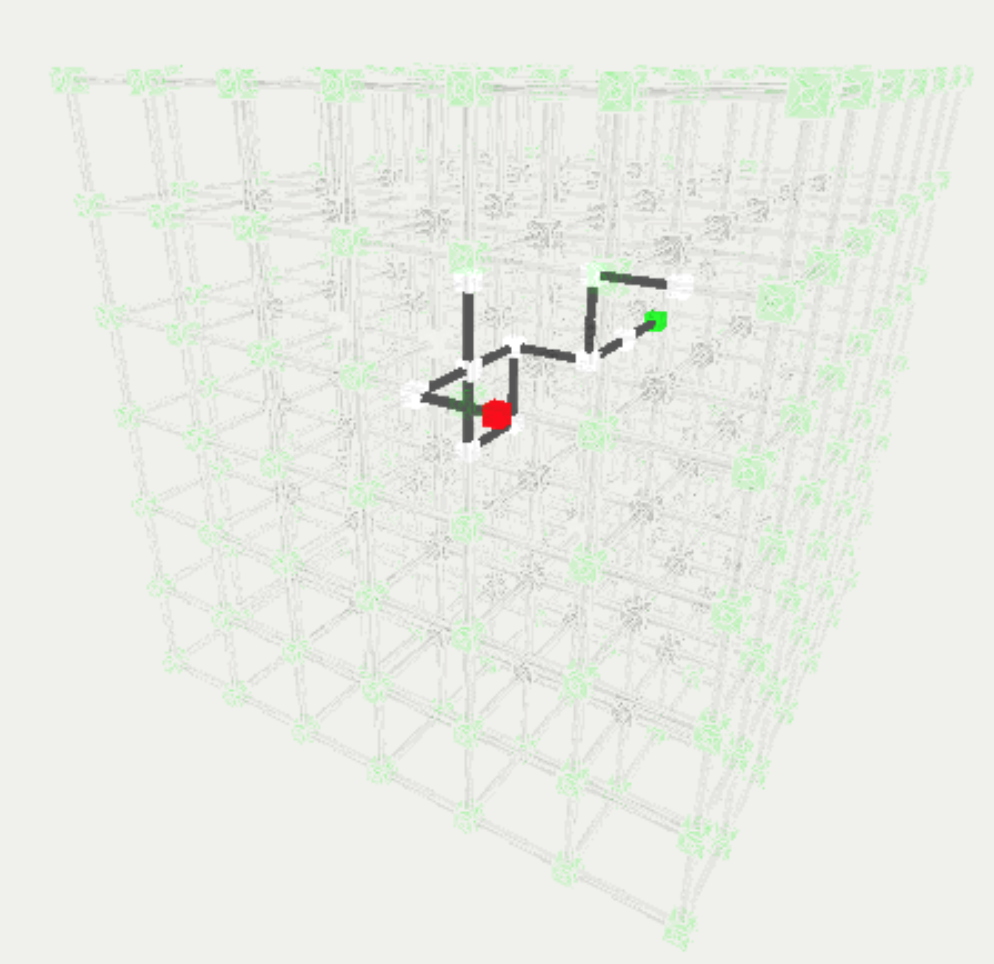}
	}
	\subfloat[Q3]{\label{subfig:no_exit_rooms}
		\includegraphics[valign=c,height=6em]{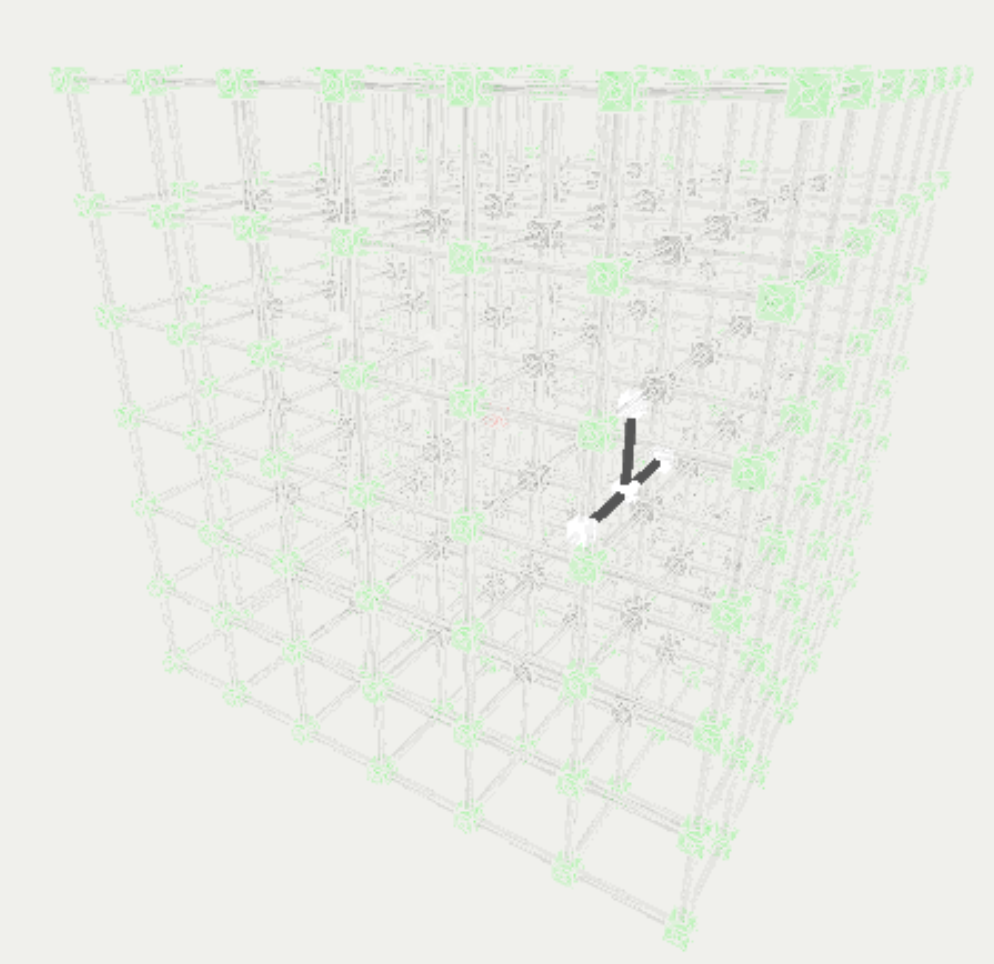} 
	}
\caption{\label{fig:3DmazeMC} Spatial model checking results of the properties in Figure~\ref{spec:3Dmaze} for the 3D maze of Figure~\ref{fig:3Dmaze}. Q1: White rooms and their connecting corridors from which a green room can be reached not passing by black rooms, including the green room that is reached (\texttt{connWG}); Q2: White rooms and their connecting corridors from which both a red and a green room can be reached not passing by black rooms (\texttt{connRWG}); Q3: White rooms and their connecting corridors with no path to green rooms  (\texttt{whiteNoGreen} or equivalently \texttt{whiteSblack}).}
\end{figure}

\begin{table}
	\small
	\begin{tabular}{ l r r r }
		\toprule
			& \textbf{Number of tasks} 
		    & \textbf{Check} (ms.) & \textbf{Total time} (ms.)\\
		\midrule
		\textbf{Property Q1 (\texttt{connWG})} & 24 & 200 & 5,450 \\
		\textbf{Property Q2 (\texttt{connRWG})} & 31 &   250 & 5,500 \\
		\textbf{Property Q3 (\texttt{whiteNoGreen})} & 25 &   100 & 5,350 \\
		\textbf{All properties in one run} & 34 & 300 & 5,550 \\
		\bottomrule
	\end{tabular}

	\caption{\label{tab:eval_maze} \PolyLogicA  performance results (rounded to the nearest multiple of 50ms) of the evaluation of the specification in Figure~\ref{spec:3Dmaze} on three properties of the 3D mesh of the maze shown in Figure~\ref{fig:3Dmaze}. The size of the model is 147,245 %
	cells. The time spent in parsing the input file (circa 12 megabytes) is about 4 seconds, 
	whereas building the Kripke structure takes about 1 second. 
	If all properties are evaluated in the same specification, the execution time is just slightly above the maximum one, since the machine has 8 physical cores, and \PolyLogicA evaluates independent properties in parallel needing less than 8 cores for the maze example.}
\end{table}

Table~\ref{tab:eval_maze} shows information on the model checking time for the properties Q1-Q3 individually and when evaluated all together. Also the total time is shown that includes the parsing of the model, the generation of the Kripke model and the actual model checking time. It is easy to see that 
most of the time goes in preparing the model (4 seconds for parsing and 1 second for building the Kripke model)  rather than the actual model checking which only takes a few hundred milliseconds. This is very encouraging, because it means that the actual model checking procedure is not the bottleneck and we are confident that more efficient parsing and Kripke model generating procedures can be found than the non-optimised ones we used for this first proof-of-concept set-up. Work on such optimisations is planned as part of future work. Note also that it is possible to verify several properties at once in which case the Kripke structure needs to be generated only one time. In future work an option could be to enable saving and loading such Kripke structures once they have been generated in order to reduce the total time needed for analysis. The evaluation was performed on a desktop machine equipped with an Intel core i9 9900k cpu and 32Gb of RAM.

Even though the 3D maze example seems simple, the fact that examples like this can be generated in various forms makes it a promising example for a future benchmark. The example lends itself for endless variations in size and shape of the maze, introducing ``holes" or forms of ``rings". Setting up a suitable benchmark for geometric model checking of polyhedra is planned as part of future work.

\subsubsection{Digital Anatomy}
 
In the second example we apply PolyLogicA on an \emph{existing} 3D mesh  from the medical domain with the aim to assess the feasibility of the geometric model checking approach on existing meshes of realistic size. The mesh visualised in Figure~\ref{fig:vis}\ref{sub@subfig:3dimage} consists of about 1.5 milion simplicial cells. A custom converter has been implemented to obtain a model file in \texttt{json} format from the input \texttt{obj} mesh;\footnote{\emph{Wavefront \texttt{obj}} is a widely used file format for 3D meshes.
See \url{https://en.wikipedia.org/wiki/Wavefront\_.obj\_file}}
 the atomic propositions used are strings of the form \texttt{"Ci"} for \texttt{C} in $\{r,g,b\}$ (standing for \emph{red}, \emph{green}, \emph{blue}) and \texttt{i} an integer between $0$ and $3$. Each such atomic proposition denotes the intensity level at each point of the red, green or blue component of the colour vector associated to the simplicial cell, with intensities quantised to four discrete target levels indicated by the integers. %

Based on such atomic propositions, it is possible to encode spatial properties involving approximations of the colour gradients\footnote{This type of mesh has colours associated to vertices, and no textures. In future versions of the tool we plan to encode colour intervals in the logic, in a similar way as has been  done in the parent tool \VoxLogicA. However, this is ongoing work, as the presence or colours, textures, materials, and so on, may depend upon the chosen file format of the mesh and will require further investigation.} that are visible in the rendered mesh. Using the specification of spatial properties shown in Figure~\ref{fig:specification}, we have checked the properties \texttt{heart} (Figure~\ref{fig:vis}\ref{sub@subfig:hearth}), \texttt{spleen} (Figure~\ref{fig:vis}\ref{sub@subfig:spleen}), \texttt{vein} (Figure~\ref{fig:vis}\ref{sub@subfig:someVeins}), 
and \texttt{selectedVein} (Figure~\ref{fig:vis}\ref{sub@subfig:selectedVein}). The syntax of the specification is the same as the one explained in the 3D maze example. 
Some definitions (e.g., \texttt{spleenWithError}) need to handle errors in the segmentations due to the usage of quantised colours leading to somewhat loose approximations; this is no problem for our purposes, as in this example we merely want to illustrate a first feasibility test and performance of the tool on existing meshes.  After loading the model file (line 1), in the specification first some areas are identified based on their approximative colour (lines 3--8). In particular, a blueish colour is defined as a combination of the blue and red intensity levels. The core of the spleen can be found in a similar way, as well as an over-approximation of the spleen by a relaxation of the constraint on the level of red. Heart, spleen or kidneys are identified by a particular combination of the levels of red and green.

This simple first approximation is then used to identify the individual organs exploiting, for example, the knowledge that veins are blueish and leading to the core of the spleen (line 10, definition \texttt{vein}). The spleen itself is first approximated by \texttt{spleenOver}, an over-approximation of the spleen that touches the veins, but may have some overlap with the heart. Recall here that the \texttt{through}-operator may be satisfied by some points (cells) that do not satisfy \texttt{spleenOver} (see the example in Figure~\ref{subfig:f1}).

The heart is specified as that part of \texttt{heartOrSpleenOrKidneys} that is not part of the over-approximated spleen that can reach a vein \texttt{spleenWithError}, but can reach a vein by itself. Subsequently, the spleen can be characterised more precisely by excluding the points belonging to the heart (line 12). Finally, the reachability operator \texttt{through} is used (line 14) to identify the veins that reach the spleen but not the heart.

We have tested the specification using \PolyLogicA\ on a desktop machine equipped with an Intel core i9 9900k cpu and 32Gb of RAM.\footnote{The memory limit is actually never hit, even with lower amounts of available RAM on other machines.} Three models have been used: the original one, and two versions that have been obtained by simplifying the original mesh, using a built-in algorithm in MeshLab~\cite{meshlab}. The number of tasks executed (the nodes of the directed acyclic graph described in Section~\ref{sec:implementation}, that is, the cardinality of the set of subformulas) is $33$. Table~\ref{tab:some-numbers} reports the model size and the execution time, broken down into parsing of the \texttt{json} model file, computing the Kripke structure, and actual model checking. Note that parsing a terse textual format for such large objects is time consuming, and ought to be replaced in future work by functions to load and save more specific mesh-based file formats. However, in the meantime, we find these results particularly encouraging, as the model checking times are quite small, and would permit to check many more formulas in a single run, thus compensating for the time needed to generate the Kripke structure.

\setlength{\tabcolsep}{10pt}
\renewcommand{\arraystretch}{1.5}

\begin{table}
	\small
	\begin{tabular}{ l r r r r r }
		\toprule
		    & \textbf{Size} (cells) & \textbf{Parse} (secs.) & \textbf{Kripke} (secs.) & \textbf{Check} (secs.) & \textbf{Total} (secs.) \\
		\midrule
		\textbf{Small} & 91,771 & 2.5 & 0.5 & 0.3 & 3.3 \\
		
		\textbf{Medium}  & 591,269 & 15 & 7 &  1.6 & 23.6 \\
		
		\textbf{Full-size} & 1,587,430 & 44 & 32 & 5 & 81\\
		\bottomrule
	\end{tabular}

	\caption{\label{tab:some-numbers} \PolyLogicA  performance results of the evaluation of the specification in Figure~\ref{fig:specification} on three variants of the 3D Medical mesh shown in Figure~\ref{subfig:3dimage}.}
\end{table}

\begin{figure}
\begin{lstlisting}[
	basicstyle=\ttfamily\footnotesize,
	xleftmargin=2em,
        framexleftmargin=1.5em,
	numbers=left]
load model = "RAJA2.json"

let blueish = (ap("b3")|ap("b2")) & (ap("r0")| ap("r1"))
let cyan = ap("r2") & (ap("g0"))

let spleenCore = ap("b3") & (ap("r1")) & ap("g1")
let spleenOver = ap("b3") & (ap("r2")|ap("r1")) & ap("g1")
let heartOrSpleenOrKidneys = ap("r2") & (ap("g1"))

let vein = through(blueish,spleenCore) 
let spleenWithError = (spleenOver | through(spleenOver,vein))
let heart = through(heartOrSpleenOrKidneys & (not(spleenWithError)),vein)
let spleen = spleenWithError & (!heart)
let selectedVein = vein & (!through(vein,heart))
\end{lstlisting}
\caption{\label{fig:specification} Specification of the spatial properties to produce the results shown in Figure~\ref{fig:vis}.}
\end{figure}

Future work will include the implementation of a fast loader for 3D meshes, in order to eliminate the parsing of (very large) model files in \texttt{json} format, and the optimisation of the translation from simplicial complexes to Kripke models, which currently exploits purely-functional data structures for ease of prototyping. We note in passing that the intermediate Kripke model may be cached, for speeding up the execution of multiple analyses on the same model. 

Future work will also include further experiments with more complex spatial formulas and different kinds of meshes with the aim to generate a basic benchmark suitable for future comparisons of performance.

\begin{figure}[t]
	\centering
	\subfloat[\parbox{2in}{3D medical image}]{\label{subfig:3dimage}
		\includegraphics[valign=c,height=5.5em]{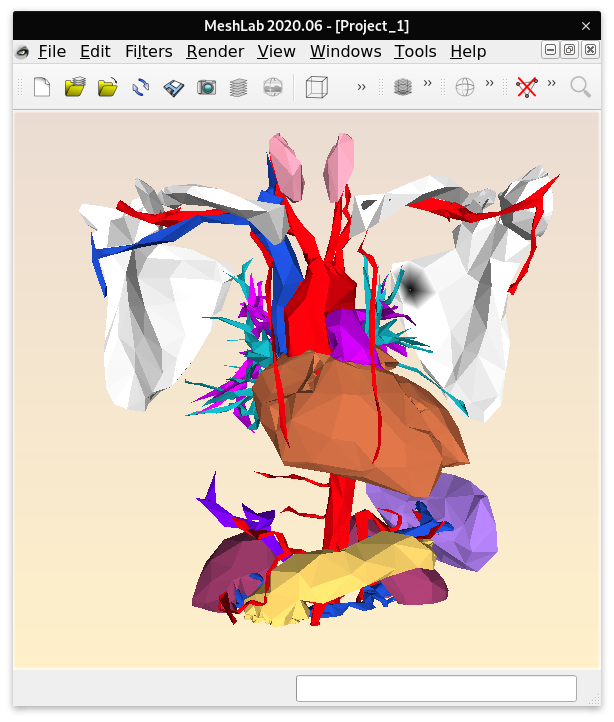}
	}
	\subfloat[\parbox{1in}{Heart (red)}]{\label{subfig:hearth}
		\includegraphics[valign=c,height=5em,trim={0 0 0cm 0},clip]{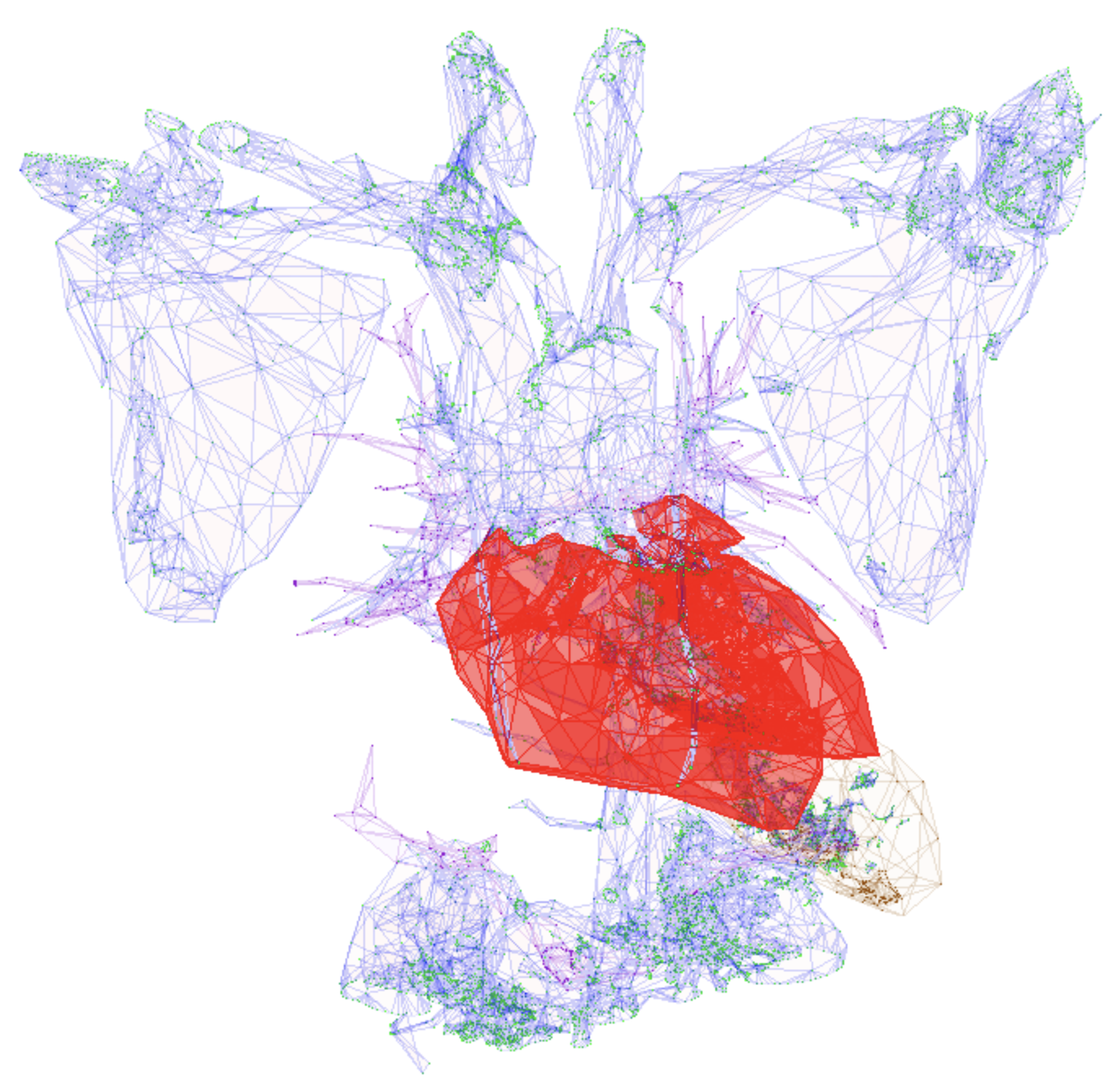}
	}

	\subfloat[\parbox{1in}{Spleen (brown)}]{\label{subfig:spleen}
		\includegraphics[valign=c,height=5em,trim={0 0 0cm 0},clip]{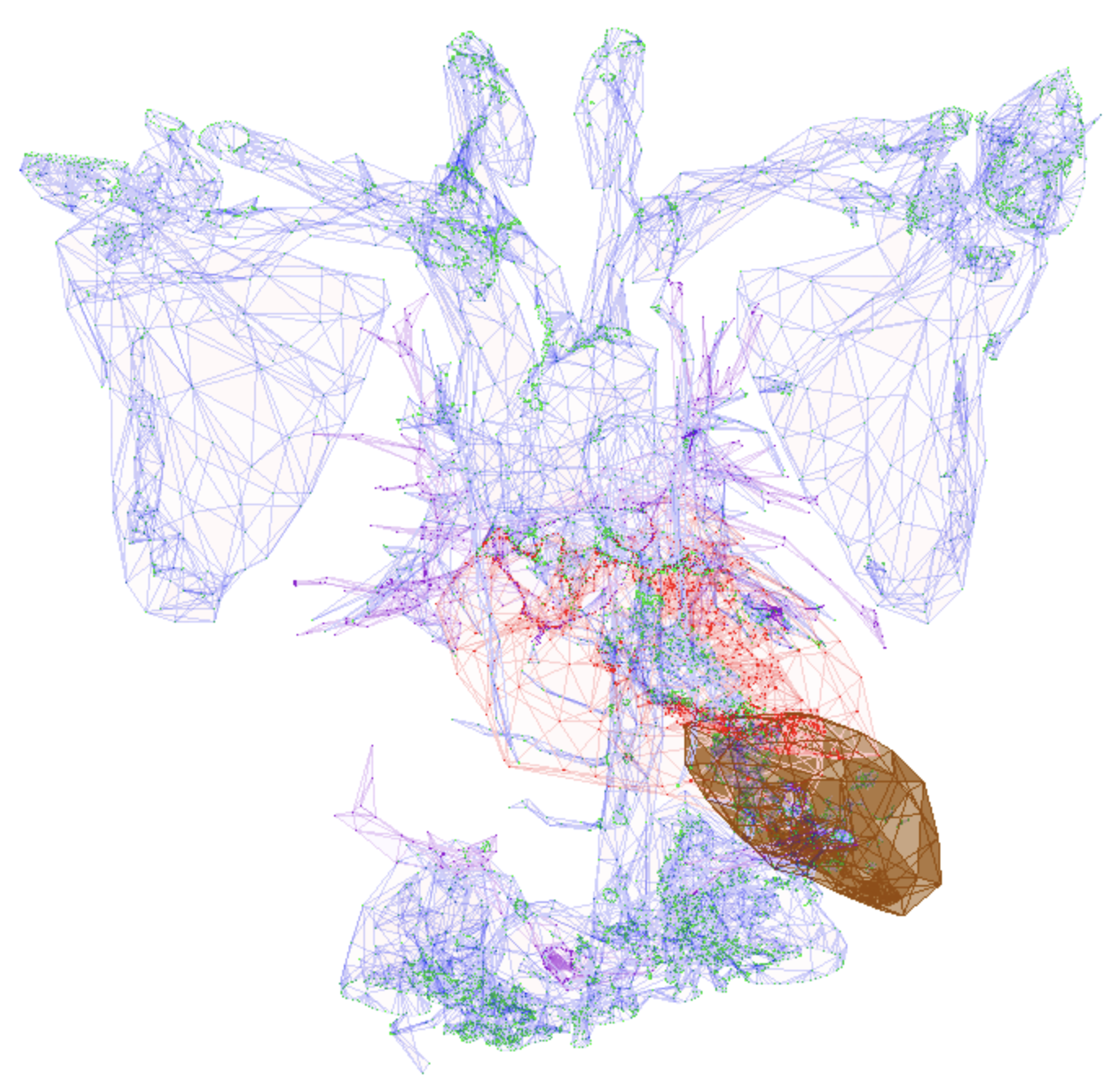}}\qquad
	\subfloat[\parbox{1in}{Veins (violet)}]{\label{subfig:someVeins}
		\includegraphics[valign=c,height=5em,trim={0 0 0cm 0},clip]{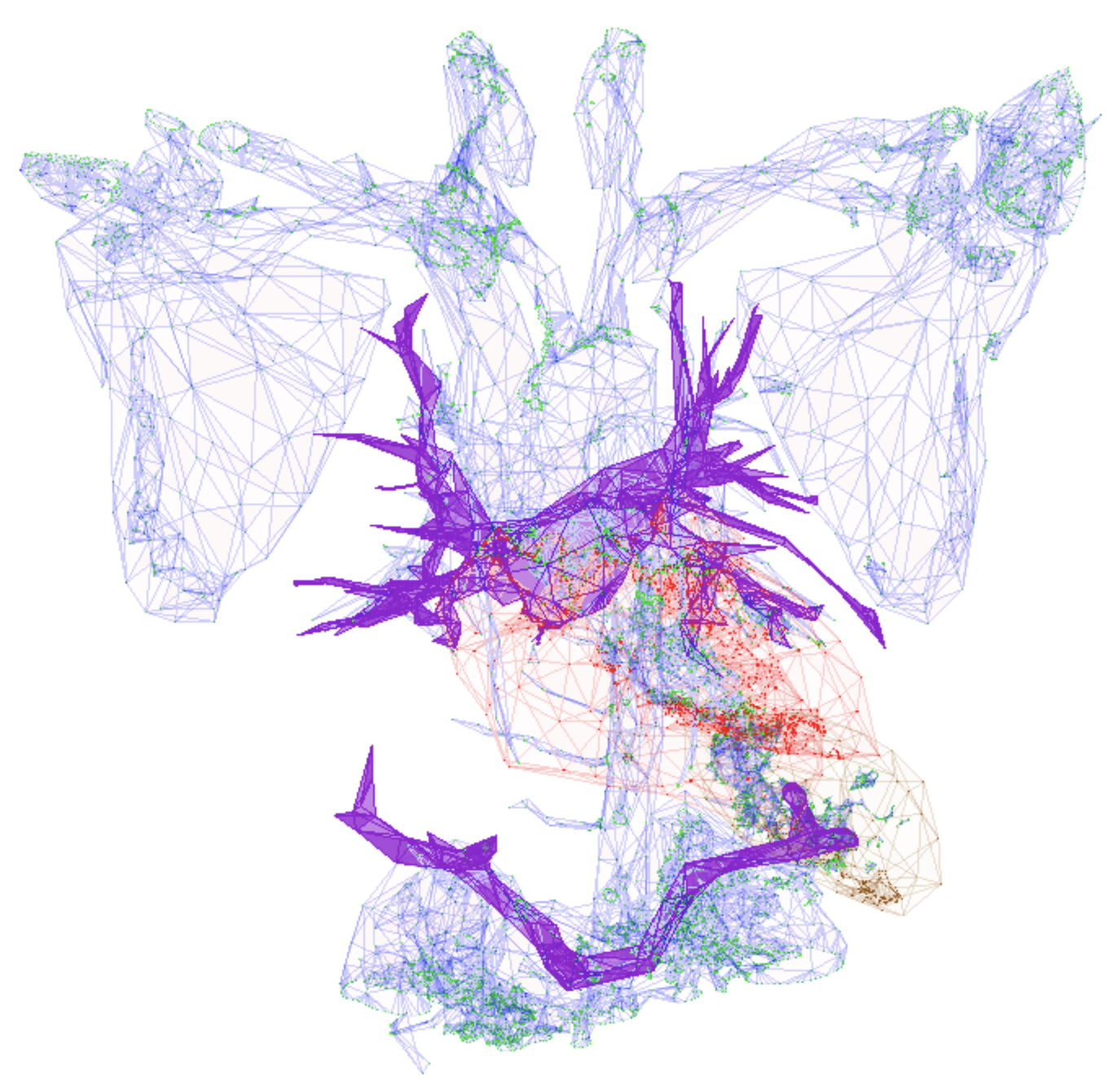}}\qquad
	\subfloat[\parbox{1in}{Selected vein\\ (violet)}]{\label{subfig:selectedVein}
		\includegraphics[valign=c,height=5em,trim={0 0 0cm 0},clip]{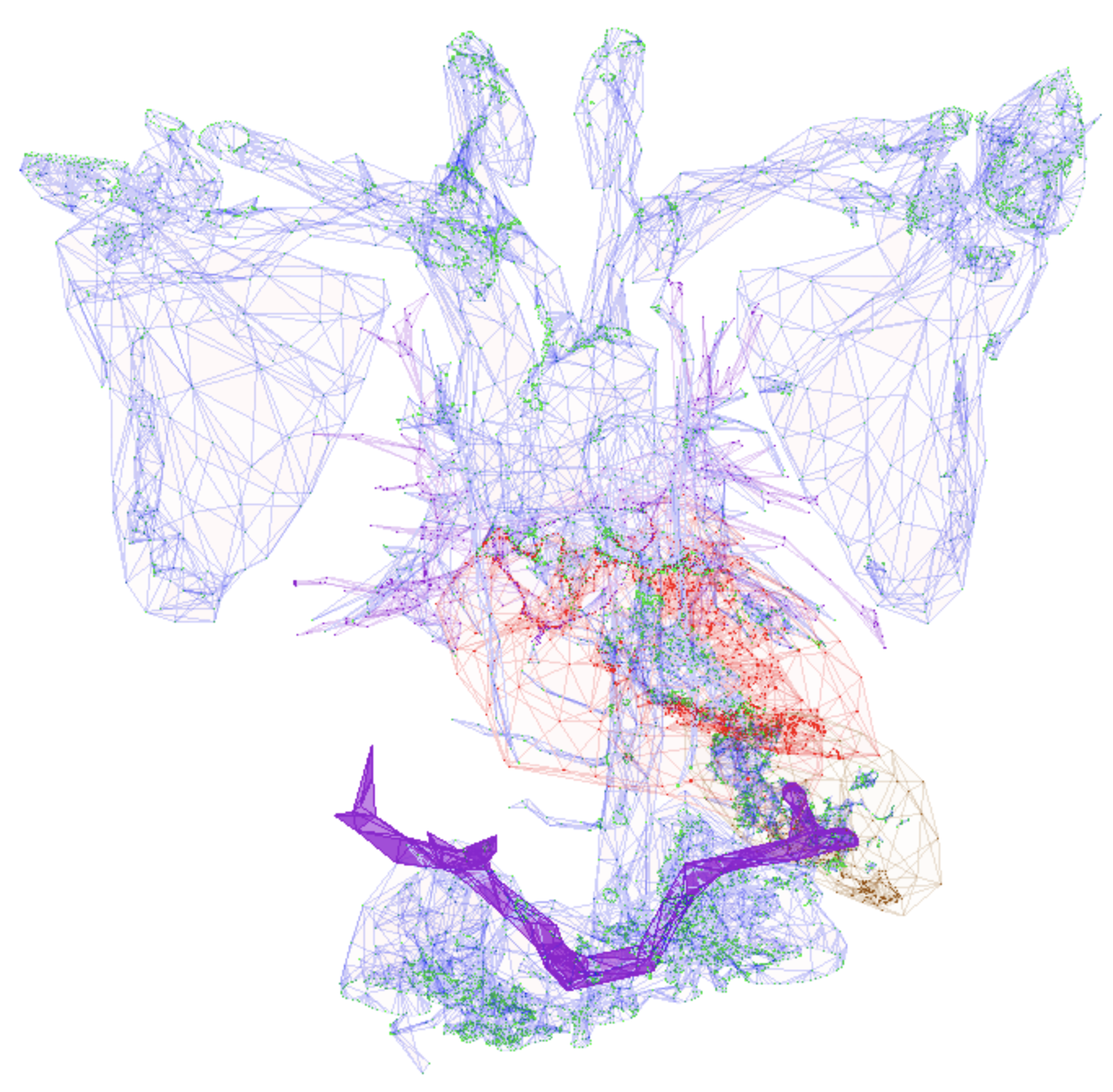}}
	\label{fig:example-visualizer}
	\caption{\label{fig:vis} 
	(\ref{subfig:3dimage}) A 3D medical illustration, courtesy of \url{www.sketchfab.com} (copyright:  COEUR et vaissaaux by Chair\_Digital\_Anatomy -- The Unesco Chair of digital anatomy (Paris University) -- is licensed under Creative Commons Attribution, see \url{https://creativecommons.org/licenses/by/4.0/legalcode}), visualized using MeshLab~\cite{meshlab}. PolyLogicA is used to segment the heart (\ref{subfig:hearth}) shown in red, spleen (\ref{subfig:spleen}) shown in brown, and some veins (\ref{subfig:someVeins}) shown in violet, and then to segment a specific vein (the one that reaches the spleen) using a reachability predicate (\ref{subfig:selectedVein}), also shown in violet. In each image, the non segmented parts are shown as mostly transparent cells. %
	}

\end{figure}

%% file: simplicial-bisimilarity.tex
\section{Simplicial Bisimilarity}\label{sec:simplicial-bisimilarity}
\label{sec:bisim}

A further, more theoretically inspired, direction of research aiming at increasing the performance of geometric spatial model checking is to exploit suitable spatial bisimilarity to reduce the models.
The main idea is to identify areas that satisfy the same spatial properties, before applying model checking.
Similar approaches have been exploited in traditional model checking techniques, based on bisimilarity for modal logics (see e.g.~\cite{VanBenthem1984}). 

As a first step into this direction, in this section we characterise logical  \SLCS equivalence via \emph{simplicial bisimilarity}. Recall  the results summarised in~\cite{AiePB07}, defining bisimilarity for topological spatial logics, so that any two points are bisimilar if and only if they are logically equivalent.

To account for the addition of the reachability operator in the logical language, the definition of bisimilarity makes use of the point-wise lifting of a relation to a path, defined in a formal way below.
In the following, fix a model $\model{X} := \tuple{P,\Ksc,V}$.

\begin{defi}\label{def:extend-relation-to-paths}
    Given a relation $R \subseteq P \times P$, let the \emph{extension of $R$ to paths} be the binary relation between paths $\hat R$, such that $\pi_1 \hat{R} \pi_2$ if and only for all $t \in [0,1]$ we have $\pi_1(t) R \pi_2(t)$.
\end{defi}

\begin{defi}[Simplicial bisimilarity]\label{def:simplicial-bisimulation}
    A binary relation $\sim\, \subseteq P\times P$ is a \emph{simplicial bisimulation} if and only if for all $x,y$ with $x \sim y$:
    \begin{enumerate}
        \item\label{cond:bisim-atoms} for all $p \in \AP$, $x \in V(p) \iff y \in V(p)$;
        \item\label{cond:bisim-forth} for each \emph{simplicial} path $\pi_x$, with $\pi_x(0) = x$, there is a simplicial path $\pi_y$ with $\pi_y(0) = y$, and $\pi_x \hat \sim \pi_y$;
        \item\label{cond:bisim-back} for each \emph{simplicial} path $\pi_y$, with $\pi_y(0) = y$, there is a simplicial path $\pi_x$ with $\pi_x(0) = x$, and $\pi_x \hat \sim \pi_y$.
    \end{enumerate} 

    \noindent The largest simplicial bisimulation, if it exists, is called \emph{simplicial bisimilarity}.
\end{defi}

\noindent
The notion of simplicial bisimulation resembles that of stuttering equivalence  for process calculi~\cite{BCG88,DNV95,Gr+17} but it should be noted that it is defined for \emph{continuous} structures.
\noindent In the following, we state the three main facts that conclude this section. Detailed proofs can be found in Appendix~\ref{sec:appendix}.

\begin{thm}\label{thm:logical-equivalence-is-bisimilarity}
    Logical equivalence is a simplicial bisimulation.
\end{thm}

\begin{thm}\label{thm:bisimilar-are-logically-equivalent}
    Each simplicial bisimulation is included in logical equivalence. 
\end{thm}
 
\begin{cor}[of Theorem~\ref{thm:logical-equivalence-is-bisimilarity} and~\ref{thm:bisimilar-are-logically-equivalent}]
    In a polyhedral model, the largest simplicial bisimulation always exists, and it coincides with logical equivalence.
\end{cor}
\begin{exa}\label{exa:bisimilarity-1d}
    Consider the polyhedral, $1$-dimensional model with cells the points $x = -1$, $y = 0$, $z = 1$, and the open segments $s = (-1,0)$ and $t = (0,1)$. Consider the set of atomic propositions $\{a,b\}$. Let $V(a) = \{y\}$ and $V(b) = z$. According to topo-bisimilarity~\cite{vBB07}, which characterises the modal fragment of our language, all the points in $s \cup t$ are equivalent, as there is no modal formula telling $s$ and $t$ apart. However, if $\slreach$ is added to the picture, let $\phi = \slreach(\lnot a,b)$. The points of $s$ do not satisfy $\phi$, but the points of $t$ do. No point of $s$ is bisimilar to a point of $t$.
\end{exa}    
\begin{exa}\label{exa:bisimilarity-2d}
In Figure~\ref{fig:bisimilarity-running-example}, we propose a simple illustration of the concept of simplicial bisimilarity. Note how the two presented paths pass through a different number of equivalent cells (in a way akin to classical ``stuttering'' forms of bisimilarities for process calculi~\cite{BK08}). 
\end{exa}

\begin{figure}[t]
    \centering
\subfloat[\parbox{0.5in}{Model}]{\label{subfig:bisim-model}
\begin{tikzpicture}[scale=1.3]%
    \tikzstyle{point}=[circle,draw=black,fill=white,inner sep=0pt,minimum width=4pt,minimum height=4pt]
    \node (p0)[point] at (0,0) {};      
    \node (p1)[point] at (0,1) {};
    \node (p2)[point] at (1,0) {};
    \node (p3)[point] at (1,1) {};
    \node (p4)[point] at (2,0) {}; 
    \node (p5)[point] at (2,1) {};
    
    \draw [dashed] (p0) -- (p2);
    \draw [dashed] (p2) -- (p3);
    \draw [dashed] (p3) -- (p5);
    \draw [dashed] (p2) -- (p5);
    \draw [dashed] (p4) -- (p5);
    \draw [dashed] (p2) -- (p4);

    \draw [red,thick](p0) -- (p1);
    \draw [red,thick](p1) -- (p3);
    \draw [red,thick](p0) -- (p2);
    \draw [red,thick](p2) -- (p3);
    \draw [red,thick](p0) -- (p3);

    \begin{scope}[on background layer]
    \fill [fill=red!50](p0.center) -- (p1.center) -- (p3.center);       
    \fill [fill=red!50](p0.center) -- (p3.center) -- (p2.center);       
    \fill [fill=green!50](p2.center) -- (p3.center) -- (p5.center);     
    \end{scope}
    
    \filldraw [red] (p0) circle (1.25pt);

    \filldraw [red] (p1) circle (1.25pt);
    \filldraw [red] (p3) circle (1.25pt);
     
\end{tikzpicture}   
}\hfill
\subfloat[\parbox{1.0in}{Bisimulation}]{\label{subfig:bisim}
\begin{tikzpicture}[scale=1.3]%
    \tikzstyle{point}=[circle,draw=black,fill=white,inner sep=0pt,minimum width=4pt,minimum height=4pt]
    \node (p0)[point] at (0,0) {};      
    \node (p1)[point] at (0,1) {};
    \node (p2)[point] at (1,0) {};
    \node (p3)[point] at (1,1) {};
    \node (p4)[point] at (2,0) {}; 
    \node (p5)[point] at (2,1) {};
    
    \draw [dashed] (p0) -- (p2);
    \draw [dashed] (p2) -- (p3);
    \draw [dashed] (p3) -- (p5);
    \draw [dashed] (p2) -- (p5);
    \draw [dashed] (p4) -- (p5);
    \draw [dashed] (p2) -- (p4);
    
    \draw [cyan,thick](p0) -- (p1);
    \draw [cyan,thick](p1) -- (p3);
    \draw [cyan,thick](p0) -- (p2);
    \draw [yellow,thick](p2) -- (p3);
    \draw [cyan,thick](p0) -- (p3);
    \draw [blue,thick](p3) -- (p5);
    \draw [red,thick](p2) -- (p5);

    \begin{scope}[on background layer]
    \fill [fill=cyan](p0.center) -- (p1.center) -- (p3.center);     
    \fill [fill=cyan](p0.center) -- (p3.center) -- (p2.center); 
    \fill [fill=violet!50](p2.center) -- (p3.center) -- (p5.center);    
    \end{scope}
    
    \filldraw [cyan] (p0) circle (1.25pt);

    \filldraw [cyan] (p1) circle (1.25pt);
    \filldraw [yellow] (p3) circle (1.25pt);
    \filldraw [orange] (p2) circle (1.25pt);
    \filldraw [red] (p5) circle (1.25pt);

\end{tikzpicture}   
}\hfill
\subfloat[\parbox{1in}{Path 1}]{\label{subfig:bisim-path1}
\begin{tikzpicture}[scale=1.3]%
    \tikzstyle{point}=[circle,draw=black,fill=white,inner sep=0pt,minimum width=4pt,minimum height=4pt]
    \node (p0)[point] at (0,0) {};      
    \node (p1)[point] at (0,1) {};
    \node (p2)[point] at (1,0) {};
    \node (p3)[point] at (1,1) {};
    \node (p4)[point] at (2,0) {}; 
    \node (p5)[point] at (2,1) {};
    
    \draw [dashed] (p0) -- (p2);
    \draw [dashed] (p2) -- (p3);
    \draw [dashed] (p3) -- (p5);
    \draw [dashed] (p2) -- (p5);
    \draw [dashed] (p4) -- (p5);
    \draw [dashed] (p2) -- (p4);
    
    \draw [red,thick](p0) -- (p1);
    \draw [red,thick](p1) -- (p3);
    \draw [red,thick](p0) -- (p2);
    \draw [red,thick](p2) -- (p3);
    \draw [red,thick](p0) -- (p3);

    \begin{scope}[on background layer]
    \fill [fill=red!50](p0.center) -- (p1.center) -- (p3.center);       
    \fill [fill=red!50](p0.center) -- (p3.center) -- (p2.center);   
    \fill [fill=green!50](p2.center) -- (p3.center) -- (p5.center); 
    \end{scope}
    
    \filldraw [red] (p0) circle (1.25pt);

    \filldraw [red] (p1) circle (1.25pt);
    \filldraw [red] (p3) circle (1.25pt);

    \node (x) at (0,1) {};
    \node at (0.2,0.9) {x};
    \filldraw [black] (x) circle (1.25pt);
    \draw [black,thick](p1) -- (0.4,0) -- (0.7,0) -- (1.2,0.7) -- (1.7,0.3) -- (2,0.4) -- (p5);     
\end{tikzpicture}   
}\hfill
\subfloat[\parbox{1in}{Path 2}]{\label{subfig:bisim-path2}
\begin{tikzpicture}[scale=1.3]%
    \tikzstyle{point}=[circle,draw=black,fill=white,inner sep=0pt,minimum width=4pt,minimum height=4pt]
    \node (p0)[point] at (0,0) {};      
    \node (p1)[point] at (0,1) {};
    \node (p2)[point] at (1,0) {};
    \node (p3)[point] at (1,1) {};
    \node (p4)[point] at (2,0) {}; 
    \node (p5)[point] at (2,1) {};
            
    \draw [dashed] (p0) -- (p2);
    \draw [dashed] (p2) -- (p3);
    \draw [dashed] (p3) -- (p5);
    \draw [dashed] (p2) -- (p5);
    \draw [dashed] (p4) -- (p5);
    \draw [dashed] (p2) -- (p4);

    \draw [red,thick](p0) -- (p1);
    \draw [red,thick](p0) -- (p2);
    \draw [red,thick](p2) -- (p3);
    \draw [red,thick](p0) -- (p3);
    \draw [red,thick](p1) -- (p3);
        
    \begin{scope}[on background layer]
    \fill [fill=red!50](p0.center) -- (p1.center) -- (p3.center);       
    \fill [fill=red!50](p0.center) -- (p3.center) -- (p2.center);   
    \fill [fill=green!50](p2.center) -- (p3.center) -- (p5.center); 
    \end{scope}
    
    \filldraw [red] (p0) circle (1.25pt);

    \filldraw [red] (p1) circle (1.25pt);
    \filldraw [red] (p3) circle (1.25pt);

    \node (y) at (0.3,0.7) {};
    \node at (0.2,0.7) {y};
    \draw [black,thick](0.3,0.7) -- (p3) -- (1.2,0.7) -- (1.5,0.2) -- (p4) -- (1.7,0.7);

\end{tikzpicture}   
}
\caption{\label{fig:bisimilarity-running-example}
An example of bisimilarity.~\ref{subfig:bisim-model}) Model with atomic propositions in green and red.~\ref{subfig:bisim}) bisimilarity, encoded via colours (points of the same colour are bisimilar).~\ref{subfig:bisim-path1}) A point $x$ and a simplicial path starting from $x$.~\ref{subfig:bisim-path2}) Another point $y$, bisimilar to $x$, and a simplicial path starting from $y$. The two paths are, in turn, bisimilar; note that these paths are also piecewise linear. The two paths cross a different set of cells, and have a different number of segments.}
\end{figure}
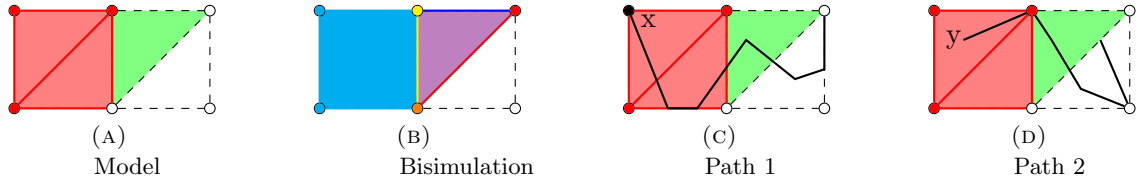

%% file: conclusions.tex
\section{Conclusions and Future Work}
\label{sec:conclusions}

We proposed an interpretation of the spatial logic \SLCS\ in continuous space using models based on polyhedra. \SLCS\ is a spatial logic stemming from topological modal logics enhanced with reachability. We developed a polyhedral semantics for \SLCS\ introducing a finite state relational Kripke model for any given polyhedral model. Such Kripke models represent the face-relation between the simplicial complexes of a polyhedron. We showed that such Kripke models provide a full logical invariant for the polyhedral model with respect to \SLCS{.} Based on these theoretical results we presented a novel spatial model checking algorithm for the verification of  \SLCS\ properties of polyhedral spaces, and in particular for the reachability operator of \SLCS{.} First applications of the prototype polyhedral model checker \PolyLogicA\ on examples of realistic size showed very encouraging results for the feasibility of this novel type of spatial model checking requiring only about 5 seconds for the actual model checking task of a full-size model consisting of circa 1.5 million of simplicial cells.

Future work will span several theoretical and applied directions. Regarding bisimilarity, by looking at the \emph{simplicial} paths that are used in Definition~\ref{def:simplicial-bisimulation}, one may wonder what logical operators (alternative to $\slreach$), and in what classes of models, can be characterised by lifting the restrictions or completely changing the kind of paths that are used therein. Furthermore, bisimilarity hints at \emph{minimization} in order to reduce the complexity of the analysis. The preliminary results presented in~\cite{ciancia2020spatial}, including the tool \emph{MiniLogicA} could be useful in this research direction. Note that the quotient mapping each simplex in $P$ to a point in $\relint{\Ksc}$ in Definition~\ref{def:KrpS} is an \emph{open map}, 
thus it preserves and reflects logical equivalence of the modal fragment of our language; additionally, as shown in Section~\ref{sec:bisim}, it preserves and reflects logical equivalence of the full language, thus simplicial bisimilarity. Not all open maps do so (just consider, e.g., the quotient with respect to classical modal logical equivalence). In future work, we plan to formalise the conditions on an arbitrary open map that make it preserve and reflect simplicial bisimilarity. 
The relationship between spatial logics and temporal logics, and related bisimilarities~\cite{KurtoninaR97} is also of interest, and in particular, comparing path-based spatial notions such as simplicial bisimilarity, to the so-called \emph{stuttering} equivalences, and their associated minimisation algorithms (see e.g.~\cite{BCG88,Gr+17}).

Spatio-temporal model checking in the style of~\cite{Gri16,CGLLM15} is a planned future development, the simplest case being the one where the underlying polyhedron does not change over time, and only the valuation of atomic propositions depends upon the temporal state of a system. More complex forms of dynamic spatial structures where the underlying polyhedron evolves over time are also of interest.

A promising application of \PolyLogicA is fully automated, declarative analysis of 3D \emph{meshes}. 
Clearly, we foresee 3D medical imaging to be a promising landscape for future applications. Furthermore, note that 3D meshes play a central role in several other fields, including architecture and computer-aided design (CAD), geographic information systems (GIS), see e.g., ~\cite{Breunig2020}, or the entertainment industry (consider 3D games or 3D animation movies), in education and in scientific visualisation.

Implementation-wise, \emph{GPU computing} could provide a computational boost to \PolyLogicA. See~\cite{BussiCG21} for a GPU implementation of the parent tool \VoxLogicA.
Finally, a user interface could be useful to explore large datasets, and to better visualise the interpretation of logic formulas, possibly exploiting results in~\cite{BroMO19} for validation.

\paragraph{Acknowledgements}

The authors are grateful to the anonymous reviewers for their constructive feedback and valuable suggestions.

%% file: appendix.tex
\section{Additional Lemmas and Proofs}
\label{sec:appendix}

\begin{prop}\label{prop:rho-and-gamma}
	Consider the operator $\rho$ of~\cite{BCLM19}. We have $\model{ X}, x \vDash \rho\, \psi [\phi] \iff \model{ X}, x \models \psi \lor \GL(\phi,\psi)$, and $\model{X}, x \vDash \GL(\phi,\psi) \iff \model{X}, x \vDash \rho \, (\phi \land \rho \, \psi[\phi])[\phi]$.
\end{prop}

\begin{proof}
	We note in passing that the following proof (and the definition of $\rho$ and $\GL$) generalise to arbitrary topological models.
	We first recall the formal definition of $\rho$, which uses right-open paths (that is, total continuous functions having $\R_{\geq 0}$ as a domain).
	Note that in~\cite{BCLM19}, $\rho$ was defined only on discrete spaces. However, following the direction of~\cite{CLLM16}, the definition applies in a natural way to continuous spaces as follows: 
	$\model{X},x \vDash \rho\, \psi\,[\phi]$  whenever there is a right-open path $\pi : \R_{\geq 0} \to P$ and an index $r$ such that $\pi(0) = x$, $\pi(r) \in \sem{\psi}$ and $\pi((0,r)) \subseteq \sem{\phi}$.
	We proceed by proving the four implications in the statement separately.

First, we show that $\m, x\models \rho\,\psi[\phi]$ implies $\m, x\models \psi \lor \slreach(\phi,\psi)$.
If $\m, x\models \rho\,\psi[\phi]$ then there is right-open path $\pi$ and $\ell \in \R_{\geq 0}$ such that
$\pi(0)=x, \pi((0,\ell)) \subseteq \sem{\phi}$ and $\pi(\ell) \in \sem{\psi}$.
We consider two distinct cases:\\
{\bf Case 1}: $\ell=0$. In this case $x \in \sem{\psi}$ and so $\m, x\models \psi \lor \slreach(\phi,\psi)$.\\
{\bf Case 2}: $\ell\neq 0$. In this case, let $\pi'$ with $\pi'(r)=\pi(r\ell)$ for all $r \in [0,1]$.
We have $\pi'(0)=\pi(0)=x$, $\pi'((0,1))=\pi((0,\ell))\subseteq \sem{\phi}$ and 
$\pi'(1)=\pi(\ell)\in \sem{\psi}$. This means that $\m, x\models \slreach(\phi,\psi)$, and so $\m, x\models \psi \lor \slreach(\phi,\psi)$.

We now show that $\m, x\models \psi \lor \slreach(\phi,\psi)$ implies $\m, x\models \rho\,\psi[\phi]$.
Suppose that $\m, x\models \psi \lor \slreach(\phi,\psi)$. We consider two distinct cases:\\
{\bf Case 1}: $\m, x\models \psi$. In this case it trivially holds $\m, x\models \rho\,\psi[\phi]$ (for any right-open path starting at $x$, just consider $\ell=0$).\\
{\bf Case 2}: $\m, x\models \slreach(\phi,\psi)$. In this case there is
a path $\pi$ such that $\pi(0)=x$, $\pi((0,1)) \subseteq \sem{\phi}$ and $\pi(1)\in \sem{\psi}$.
We obtain $\m, x\models \rho\,\psi[\phi]$ by taking $\ell=1$ in the definition of the semantic clause of $\rho$.

We now show that $\m,x\models \slreach(\phi,\psi)$ implies $\m,x\models \rho \,(\phi \land \rho\,\psi[\phi])[\phi]$.
Suppose  $\m,x\models \slreach(\phi,\psi)$, that is, there is a path
$\pi$ such that $\pi(0)=x$, $\pi((0,1)) \subseteq \sem{\phi}$ and $\pi(1)\in \sem{\psi}$.
Since $\pi$ is a total continuous function, there is an $\epsilon \in (0,1)$ such that $\pi(\epsilon)\in \sem{\phi}$. We define two additional paths $\pi_1(t) = \pi(t\epsilon)$ and $\pi_2(t)=\pi((1-\epsilon)t + \epsilon)$.
Notice that $\pi_2(0)=\pi(\epsilon)\in \sem{\phi}, \pi_2((0,1))=\pi((\epsilon,1)) \subseteq \sem{\phi}$
and $\pi_2(1)=\pi(1) \in \sem{\psi}$:
this shows that $\m,\pi(\epsilon) \models \phi \land \rho\,\psi[\phi]$.
Moreover, we have $\pi_1(0)=\pi(0)=x, \pi_1((0,1))=\pi((0,\epsilon)) \subseteq \sem{\phi}$ and
$\pi_1(1) = \pi_2(0) \in \sem{\phi \land \rho\,\psi[\phi]}$: this shows that $\m,x\models \rho \,(\phi \land \rho\,\psi[\phi])[\phi]$.

Finally, we show that $\m,x\models \rho\,(\phi \land \rho\,\psi[\phi])[\phi]$ implies $\m,x\models \slreach(\phi,\psi)$.
Suppose that $\m,x\models \rho\,(\phi \land \rho\,\psi[\phi])[\phi]$, that is, there is a right-open path $\pi_1$ and a value $\ell_1 \in \R_{\geq 0}$ such that $\pi_1(0) = x, \pi_1((0,\ell_1)) \subseteq \sem{\phi}$ and $\pi_1(\ell_1) \in \sem{\phi \land \rho\,\psi[\phi]} = \sem{\phi} \cap \sem{\rho\,\psi[\phi]}$. 
By the last condition, there is a right-open path $\pi_2$ and a value $\ell_2 \in \R_{\geq 0}$ such that $\pi_2(0) = \pi_1(\ell_1)$, $\pi_2((0,\ell_2)) \subseteq \sem{\phi}$ and $\pi_2(\ell_2)\in \sem{\psi}$.
We consider two distinct cases:\\
{\bf Case 1}: $l_1 = l_2 = 0$.
In this case $\pi_1(l_1) = \pi_2(l_2) = x$, and so $x \in \sem{\phi} \cap \sem{\psi}$.
In this case $\model{X},x \vDash \gamma(\phi,\psi)$ holds trivially.\\
{\bf Case 2}: $l_1 >0$ or $l_2 > 0$.
In this case define the path $\pi'$ by imposing $\pi'(t) = \pi_1(2 t l_1)$ and $\pi'(\frac{1}{2}+t) = \pi_2(2 t l_2)$ for $t\in [0,\frac{1}{2}]$ (notice that the path is well-defined since $\pi'(\frac{1}{2}) = \pi_1(l_1) = \pi_2(0)$).
Clearly $\pi'(0) = x$.
Moreover, $\pi'((0,1)) = \pi'((0,\frac{1}{2}]) \cup \pi'((\frac{1}{2},1)) = \pi_1((0,l_1)) \cup \pi_2((0,l_2)) \subseteq \sem{\phi}$ (under the assumption that $l_1 >0$ or $l_2>0$).
And finally $\pi'(1) = \pi_2(l_2) \in \sem{\psi}$.
This shows that $\model{X},x \vDash \GL(\phi,\psi)$, as desired.\end{proof}

\begin{prop}\label{prop:pathToBorder}
	Let $\sigma$ be a non-empty simplex and consider $x \in \relint{\sigma}$ and $y \in \sigma$.
	Then there is a \emph{linear} path (a segment in $\R^d$) $\pi: [0,1] \to \sigma$ with $\pi(0) = x$, $\pi(1) = y$ and $\pi([0,1)) \subseteq \relint{\sigma}$; more precisely, $\pi(t)=ty+(1-t)x$.
\end{prop}

\begin{proof}
	Let $V = \{ v_0,\dots,v_n \}$ be the set of vertices of $\sigma$.
	By definition of $\sigma$, $x$ and $y$ are in the convex hull of $V$, that is, there exist $\lambda_i$ and $\delta_i$ such that $x = \sum_{i=0}^n \lambda_i v_i$ and $y = \sum_{i=0}^n \delta_i v_i$.
	Since $x \in \relint{\sigma}$, every $\lambda_i$ is strictly greater than $0$.
	For $\pi$ as defined in the statement we have $\pi(t) = \sum_{i=0}^n \left(  t\delta_i + (1-t)\lambda_i  \right) v_i$.
	This function is clearly continuous (thus a path), and we have $\pi(0) = x$ and $\pi(1) = y$.
	Moreover, for every $t \in [0,1)$ and every $i \leq n$ we have $t\delta_i + (1-t)\lambda_i \geq (1-t)\lambda_i >0$, and so $\pi( [0,1) ) \subseteq \relint{\sigma}$.
\end{proof}

\begin{lem}\label{lemma:faceInTermsOfPartition}
	Let $\Ksc$ be a simplicial complex and $\sigma, \tau \in \Ksc$.
	Then $\sigma \face \tau$ iff $\relint{\sigma} \,\relint{\face}\, \relint{\tau}$ iff $\relint{\sigma} \subseteq \Cl(\relint{\tau})$ iff $\relint{\sigma} \cap \Cl(\relint{\tau}) \ne \emptyset$.
\end{lem}
\begin{proof}
	The first equivalence is simply the definition of $\relint{\face}$ (see Definition~\ref{definition:pmpath}).
	As for the other equivalences, firstly we show that $\sigma \face \tau$ implies $\relint{\sigma}\subseteq \Cl(\relint{\tau})$.
	In fact, if $\sigma \face \tau$ we have $\sigma \subseteq \tau$, and so $\relint{\sigma} \subseteq \sigma \subseteq \tau = \Cl(\relint{\tau})$.
	Secondly, note that since cells are nonempty sets we have that $\relint{\sigma}\subseteq \Cl(\relint{\tau})$ implies $\relint{\sigma} \cap \Cl(\relint{\tau}) \ne \emptyset$.
	Finally, we show that $\relint{\sigma} \cap \Cl(\relint{\tau}) \ne \emptyset$ implies $\sigma \face \tau$, concluding the proof.
	Notice that $\Cl(\relint{\tau}) = \tau = \bigcup\{ \relint{\tau}' \,|\, \tau' \face \tau \}$.
	So by Lemma~\ref{lemma:partition} we have that either $\relint{\sigma}\in \{ \relint{\tau}' \,|\, \tau' \face \tau \}$ or $\relint{\sigma} \cap \Cl(\relint{\tau}) = \emptyset$.
	Since the latter is not the case by assumption, we conclude $\sigma \face \tau$.
\end{proof}

\begin{lem}\label{lem:paths}
    Let $\pi_1$ and $\pi_2$ be two paths, with $\pi_1(1) = \pi_2(0)$; let $x,y \in [0,1]$, with $x < y$. Define in the obvious way the \emph{concatenation} $\pi_1;\pi_2$, with $\pi_1;\pi_2(0) = \pi_1(0)$ and $\pi_1;\pi_2(1) = \pi_2(1)$ and the sub-path $\subpath{\pi_1}{[x,y]}$, with $\subpath{\pi_1}{[x,y]}(0) = \pi_1(x)$ and $\subpath{\pi_1}{[x,y]}(1) = \pi_1(y)$. We have that: whenever $\pi_1$ and $\pi_2$ are piecewise-linear (simplicial), also $\pi_1;\pi_2$ is piecewise linear (simplicial); whenever $\pi_1$ is piecewise linear (simplicial), also $\subpath{\pi_1}{[x,y]}$ is piecewise linear (simplicial). Furthermore, for any two (additional) paths $\pi'_1$, $\pi'_2$ with $\pi'_1(1) = \pi'_2(0)$, and relation $R$, if $\pi_1 \hat R \pi'_1$ and $\pi_2 \hat R \pi'_2$, then $\pi_1;\pi_2 \hat R \pi'_1;\pi'_2$.
\end{lem}

\begin{proof}
    We omit the proof, which is straightforward.
\end{proof}

\noindent
{\bf Lemma~\ref{lem:invariance-wrt-triangulations}.}
\emph{	Let $\model{X} = \tuple{P, \Ksc, V}$ and $\model{X}' = \tuple{P, 		\Ksc', V}$ be two models sharing the same $P$ and $V$.
	For each $x\in P$ and $\phi$ we have:
		$\model{X}, x \vDash \phi  \iff \model{X}', x \vDash \phi$.
}
\begin{proof}
	Just note that $\Ksc$ does not appear in Definition~\ref{def:slcs-semantics}.		
\end{proof}

\noindent
{\bf Lemma~\ref{lemma:semanticsTriangulation}.}
{\em
	Let $\model{X}$ be a polyhedral model and $\Ksc$ a simplicial complex coherent with $\model{X}$.
	Then for each cell $\relint{\sigma} \in \relint{\Ksc}$ and $x,y \in \relint{\sigma}$ we have $x \lequiv y$.
}
\begin{proof}
	We prove the result by induction on the structure of $\phi$.
	Since most of the cases follow easily from the semantic clauses of the logic, we show only the cases for $\phi = \Box\psi$ and $\phi = \slreach(\psi,\chi)$.

	\smallskip
	\noindent
	\textbf{Case $\phi = \Box\psi$:}
	By inductive hypothesis, $\sem{\psi}$ is a union of cells of $\relint{\Ksc}$, and by Lemma~\ref{lemma:partition} also $P \setminus \sem{\psi}$ is a union of cells.
	Since there are only finitely many cells in $\relint{\Ksc}$, we have $\Cl(P \setminus \sem{\psi}) = \bigcup\{  \Cl(\relint{\sigma})  \,|\,  \relint{\sigma} \subseteq P\setminus\sem{\psi}  \}$.
	And since $\Cl(\relint{\sigma}) = \sigma = \bigcup\{ \relint{\tau} \,|\, \tau\face\sigma\}$, it follows that $\Cl(P \setminus \sem{\psi}) = \bigcup\{ \relint{\tau} \,|\, \exists \sigma.\, \tau\face\sigma \,\text{and}\, \relint{\sigma} \subseteq P\setminus\sem{\psi} \}$, that is, $\Cl(P \setminus \sem{\psi})$ is a union of cells.
	To conclude, by Lemma~\ref{lemma:partition} we have that $\sem{\phi} = \interior(\sem{\psi}) = P\setminus \Cl(P \setminus \sem{\psi})$ is again a union of cells.

	\smallskip
	\noindent
	\textbf{Case $\phi = \slreach(\psi,\chi)$:}
	Suppose that $\model{X}, x \vDash \slreach(\phi,\psi)$:
	we aim to show that $\model{X}, y \vDash \slreach(\phi,\psi)$.
	The fact that  $\model{X}, x \vDash \slreach(\phi,\psi)$ means that there exists a path $\pi$ such that $\pi(0) = x$, $\pi( (0,1) ) \subseteq \sem{\phi}$ and $\pi(1) \in \sem{\psi}$.
	By inductive hypothesis $\sem{\phi}$ is a union of cells, and since $\sem{\phi} \cap \pi((0,1))$ is not empty also the set $\Pi = \{ \relint{\tau} | \exists r\in(0,1).  \pi(r)\in \relint{\tau} \subseteq \sem{\phi} \}$ is not empty.
	
	Since $\pi((0,1)) \subseteq \bigcup \Pi$, we have $x \in \Cl( \pi((0,1)) ) \subseteq \bigcup\{ \tau | \exists r\in(0,1).  \pi(r)\in \relint{\tau} \subseteq \sem{\phi} \}$.
	Thus there exists a value $r\in (0,1)$ and a cell $\relint{\tau}$ such that $x \in \tau$ and $\pi(r) \in \relint{\tau}$.
	Since $x \in \relint{\sigma} \cap \Cl(\relint{\tau})$, by Lemma~\ref{lemma:faceInTermsOfPartition} we have $\relint{\sigma} \subseteq \Cl(\relint{\tau}) = \tau$ and consequently $y \in \tau$.
	By Proposition~\ref{prop:pathToBorder} (modulo inverting and reparametrising the path) there exists a path $\pi': [0,r] \to \tau$ such that $\pi'(0) = y$, $\pi'(r) = \pi(r)$ and $\pi'((0,r)) \subseteq \relint{\tau} \subseteq \sem{\phi}$.
	If we extend $\pi'$ by imposing $\pi'(t) = \pi(t)$ for $t \in (r,1]$, we obtain a path such that $\pi'(0) = y$, $\pi'((0,1)) \subseteq \sem{\phi}$ and $\pi'(1) \in \sem{\psi}$.
	In particular $\pi'$ witnesses that $\model{X}, y \vDash \slreach(\phi,\psi)$, as desired.
\end{proof}

\noindent
{\bf Proposition~\ref{equivalence-classes-characteristic-formulas}.}
{\em
Given a polyhedral model $\model{X}$, the relation $\lequiv$ has only finitely many equivalence classes. Furthermore, each equivalence class $C$ has a \emph{characteristic formula} $\phi^C$ such that $\model{X},x \vDash \phi^C \iff x \in C$.
}

\begin{proof}
	Fix $\Ksc$ coherent with $\model{X}$.
	By Lemma~\ref{lemma:semanticsTriangulation}, each equivalence class is a union of distinct cells of $\relint{\Ksc}$.
	But since there are only finitely many cells, there are finitely many equivalence classes too. 
    For two distinct equivalence classes $C$ and $C'$, let $\phi^{C,C'}$ be a formula satisfied by the elements of $C$ but not by the elements in $C'$---such a formula exists by definition of $\lequiv$.
    Then the formula $\phi^C = \bigwedge_{C' \ne C} \phi^{C,C'}$ is satisfied only by elements in $C$, as desired.
\end{proof}

\noindent
{\bf Theorem~\ref{thm:near-derived-from-slreach}.}
{\em
    For each formula $\phi$, we have $\model{X},x \vDash \Box \phi \iff \model{X},x \vDash \lnot \slreach(\lnot \phi,\top)$.
}
\begin{proof}
    Recall the definition of $\Diamond \phi := \lnot \Box \lnot \phi$ and that $\model{X},x \models \Diamond \phi \iff x \in \Cl(\sem{\phi})$ (see e.g.~\cite{vBB07}), where $\Cl$ is the topological closure operator.
	Thus, we prove the equivalent statement $\model{X}, x \vDash \Diamond\phi \iff \model{X}, x \vDash \slreach(\phi,\top)$.
    
    If $\model{X},x \models \slreach(\phi,\top)$, by definition, there is a path $\pi$ with $\pi(0) = x$ and $\pi((0,1)) \subseteq \sem{\phi}$. 
    In particular we have $x = \pi(0) \in \Cl(\,  \pi((0,1))  \,)$ by continuity of paths, and $\Cl(\,  \pi((0,1))  \,) \subseteq \Cl( \sem{\phi} )$ by monotonicity of $\Cl$.
    Therefore $x \in \Cl(\sem{\phi})$, which amounts to $\model{X},x \vDash \Diamond \phi$.
    
    If $\model{X}, x \models \Diamond \phi$, we have $x \in \Cl(\sem{\phi})$.
    Fix $\Ksc$ a simplicial complex coherent with $\model{X}$.
    By Lemma~\ref{lemma:semanticsTriangulation} $\sem{\phi}$ is a union of cells in $\relint{\Ksc}$, and so $\Cl(\sem{\phi}) = \bigcup \{  \Cl(\relint{\sigma})  \,|\,  \relint{\sigma} \subseteq \sem{\phi}  \}$.
    In particular, $x \in \Cl(\relint{\sigma}) = \sigma$ for one of these cells.
    Fix an arbitrary element $y\in \relint{\sigma}$---recall that $\relint{\sigma}$ is nonempty whenever $\sigma$ is nonempty.
    By Proposition~\ref{prop:pathToBorder} (note that the names $x$ and $y$ are inverted in the statement of the proposition) there is a path $\pi$ with $\pi(0) = y$, $\pi(1) = x$, and $\pi([0,1)) \subseteq \widetilde{\sigma} \subseteq \sem{\phi}$.
    The ``converse'' path $\pi'(i) := \pi(1-i)$ witnesses that $\model{X},x \models \slreach(\phi,\top)$.
\end{proof}

\noindent
{\bf Lemma~\ref{lemma:pathToPLPath}.}
\emph{
	Let $\Ksc$ be a simplicial complex and $x,y \in \poly{\Ksc}$.
	Then there exists a path in $\poly{\Ksc}$ from $x$ to $y$ if and only if there is a PL-path in $\poly{\Ksc}$ from $x$ to $y$.
}
\begin{proof}
	The right-to-left direction is trivial, so we focus on the left-to-right direction.
	Fix a path in $\poly{\Ksc}$ starting at $x$ and ending at $y$, and consider the set $\Pi \subseteq \relint{\Ksc}$ of cells traversed by the path.
	Clearly $\bigcup\Pi$ is a connected subset of $\poly{\Ksc}$.
	To obtain the desired result it suffices to show that $\bigcup\Pi$ is also \emph{PL-connected}, that is, for all $x,y \in \bigcup\Pi$ there exists a PL-path from $x$ to $y$.
	We are going to prove this by induction on the cardinality of $\Pi$.

	\smallskip
	\noindent
	\textbf{Base case:} If $\card{\Pi} = 1$, then any two points $x,y \in \bigcup \Pi$ belong to the same cell $\relint{\sigma}$, and since cells are convex there is a linear path (thus piecewise linear) connecting $x$ and $y$.
	
	\smallskip
	\noindent
	\textbf{Inductive step:} Suppose that $\card{\Pi} = n+1$ and that the result holds for sets with lower cardinality.
	Consider two points $x,z \in \bigcup \Pi$ and call $\relint{\sigma}$ the cell containing $x$.
	If $z \in \relint{\sigma}$ too we can reason as in the base case, so we can assume otherwise.
	Define $Z$ to be the connected component of $\bigcup\left( \Pi \setminus \{ \relint{\sigma} \} \right)$ containing $z$.
	$Z$ is itself the union of a set of cells $\Pi'$---since cells are connected.
	As $\card{\Pi'} < \card{\Pi}$, by inductive hypothesis $Z$ is PL-connected.
	Notice that $Y := \relint{\sigma} \cup Z$ is connected, for otherwise $Z$ would be disconnected from $\bigcup\Pi \setminus Z$, against the initial assumption that $\bigcup\Pi$ is connected.

	$\Cl(\relint{\sigma}) = \sigma$ and $\Cl(Z)$ are closed sets whose union covers $Y$.
	If the intersection $Y \cap \sigma \cap \Cl(Z)$ were empty, then $Y \cap \sigma$ and $Y\cap \Cl(Z)$ would disconnect $Y$.
	So there must be a point $y \in Y \cap \sigma \cap \Cl(Z)$.
	By Proposition~\ref{prop:pathToBorder} there exists a linear (thus PL) path $\pi_1$ connecting $x \in \relint{\sigma}$ and $y \in \sigma$, and with $\pi_1([0,1)) \subseteq \relint{\sigma} \subseteq Y $.
	Moreover, since $y \in \Cl(Z) = \Cl(\bigcup\Pi')$, there exists a cell $\relint{\tau} \in \Pi'$ such that $y \in \Cl(\relint{\tau}) = \tau$.
	So, again by Proposition~\ref{prop:pathToBorder} (with $w$ playing the role of $x$ in the statement of the proposition), there exists a linear path $\pi_2$ from $y \in \tau$ to an arbitrary point $w \in \relint{\tau}$, fully contained in $Y$.
	Finally, since $Z$ is PL-connected, there exists a PL-path $\pi_3$ from $w$ to $z$ in $Z \subseteq Y$.
	By concatenating $\pi_1$, $\pi_2$ and $\pi_3$ we obtain a PL-path connecting $x$ and $z$ contained in $Y \subseteq \bigcup \Pi$, as desired. \end{proof}

\noindent
{\bf Lemma~\ref{lem:piecewise-linear-semantics}.}
{\em
We have: $\model{X}, x \vDash \slreach( \phi, \psi )$ if and only if
			there is a PL-path $\pi : [0,1] \toPL P$ such that $\pi(0) = x$ and $\pi((0,1)) \subseteq \sem{\phi}$ and $\pi(1) \in \sem{\psi}$.
}
\begin{proof}
	The right-to-left implication follows trivially from the semantics clauses of the reachability operators, so we focus on the left-to-right implications.
	Fix a simplicial complex $\Ksc$ coherent with $\model{X}$.
	Suppose that there exists a path $\pi: [0,1] \to P$ such that $\pi(0) = x$, $\pi((0,1)) \subseteq \sem{\phi}$ and $\pi(1) \in \sem{\psi}$.
	Consider the set of cells $\Pi := \{  \relint{\sigma} \in \relint{\Ksc}  \,|\,  \pi((0,1)) \cap \relint{\sigma} \ne \emptyset  \}$.
	Clearly $\bigcup \Pi$ is connected, and so by Lemma~\ref{lemma:pathToPLPath} it is also PL-connected (i.e., for every two points in $\bigcup\Pi$ there is a PL-path in $\bigcup\Pi$ connecting them).
	Moreover $x,\pi(1) \in \Cl(\bigcup\Pi)$.
	
	Since $x\in \Cl(\bigcup\Pi)$, there exists a cell $\relint{\sigma} \in \Pi$ such that $x \in \Cl(\relint{\sigma}) = \sigma$.
	So by Proposition~\ref{prop:pathToBorder} there exists a linear path $\pi_1$ from $x$ to a point (arbitrarily chosen) $y\in \relint{\sigma} \subseteq \bigcup\Pi$ such that $\pi_1((0,1)) \subseteq \relint{\sigma}$ (note that the names $x$ and $y$ are inverted in the statement of the proposition).
	By a similar argument, there exists a linear path $\pi_3$ from a point $z \in \bigcup\Pi$ to $\pi(1) \in \sem{\psi}$, such that $\pi_3((0,1)) \subseteq \bigcup\Pi$.
	Finally, since $\bigcup\Pi$ is PL-connected, there exists a PL-path $\pi_2$ from $y$ to $z$ completely contained in $\bigcup\Pi$.
	By concatenating $\pi_1$, $\pi_2$ and $\pi_3$ we obtain a PL-path $\pi'$ such that $\pi'(0) = x$, $\pi'((0,1)) \subseteq \bigcup\Pi \subseteq \sem{\phi}$ and $\pi'(1) \in \sem{\psi}$, as desired.
\end{proof}

\noindent
{\bf Lemma~\ref{lem:pl-are-simplicial}.}
{\em
    Any piecewise linear path is simplicial.
}
\begin{proof}
    As cells are convex sets, the intersection between a cell and a segment is a segment.
    So any segment crosses each cell at most once.
    Observing that there are finitely many segments in a PL-path, one obtains the proof.
\end{proof}

\noindent
{\bf Theorem~\ref{thm:kripke-model-checking}.}
{\em
Let $\model{X} = \tuple{P, \Ksc, V}$ be a polyhedral model and  $x$ a point of $P$.
	Let $\sigma \in \Ksc$ be the unique simplex such that $x \in \relint{\sigma}$.
	For every formula $\phi$ of \SLCS we have $\model{X}, x \vDash \phi
	\iff
	\KM(\model{X}), \relint{\sigma} \vDash \phi$.
}

\begin{proof}
	Notice that by Lemma~\ref{lemma:semanticsTriangulation} the left side of the bi-implication is equivalent to $\relint{\sigma} \subseteq \sem{\phi}$.
	We proceed by induction on the structure of the formula;
	the only non trivial cases are when the formula is of the form $\Box\psi$ and when the formula is of the form $\slreach(\phi,\psi)$.

	\smallskip
	\noindent
	\textbf{Case $\Box\psi$:}
	We have that $\model{X}, x \vDash \Box\psi$ iff $x \notin \Cl(\sem{\neg\psi})$.
	By Lemma~\ref{lemma:semanticsTriangulation}, $\sem{\neg\psi}$ is a finite union of cells, so, by Lemma~\ref{lemma:faceInTermsOfPartition}, the condition is equivalent to $x \notin \bigcup\{ \tau \,|\, \relint{\tau}\subseteq \sem{\neg\psi} \}$;
	 and again, by Lemma~\ref{lemma:faceInTermsOfPartition}, it is equivalent to $\relint{\sigma} \,\relint{\not\face}\, \relint{\tau}$ for every $\relint{\tau} \subseteq \sem{\neg\psi}$.
	Again by Lemma~\ref{lemma:semanticsTriangulation}, for every $\relint{\tau} \in \relint{\Ksc}$ we have that if $\relint{\sigma} \,\relint{\face}\, \relint{\tau}$ then $\relint{\tau} \subseteq \sem{\psi}$, which by the inductive hypothesis amounts to $\KM(\model{X}), \relint{\sigma} \vDash \Box \psi$.
	Note that we have equivalences in both directions of the proof, so the bi-implication is obtained ``for free'' for this part.

	\smallskip
	\noindent
	\textbf{Case $\slreach(\phi,\psi)$:}
	Firstly, suppose that $\KM(\model{X}), \relint{\sigma} \vDash \slreach(\phi,\psi)$.
	Then there exists a $\pm$-path $\pi: \{0,\dots,k\} \toPM \relint{\Ksc}$ such that $\pi(0) = \relint{\sigma}$, $\pi(\{1,\ldots,k-1\}) \subseteq \semKM{\phi}$ and $\pi(k) \in \semKM{\psi}$.

	By inductive hypothesis and Lemma~\ref{lemma:semanticsTriangulation}, $\pi(j) \subseteq \sem{\phi}$ for $j \in \{1,\dots,k-1\}$ and $\pi(k) \subseteq \sem{\psi}$.

	For any two cells $\relint{\sigma} \,\relint{\face}\, \relint{\tau}$, we have $b_{\sigma} \in \relint{\sigma} \subseteq \tau$ (recall that $b_{\sigma}$ indicates the \emph{barycentre} of $\sigma$).
	So by Proposition~\ref{prop:pathToBorder} (where $b_\sigma$ and $b_\tau$ play the role of $y$ and $x$ respectively in the statement of the proposition) there exists a linear path connecting $b_{\sigma} \in \tau$ to $b_\tau \in \relint{\tau}$ and mapping the interval $(0,1)$ to $\relint{\tau}$.
	Using this fact, for every $i = 1,\dots k$ we can find a linear path $\pi'_i$ connecting $b_{\pi(i-1)}$ to $b_{\pi(i)}$ with $\pi'_i((0,1)) \subseteq \pi(i-1)$ or $\pi'_i((0,1)) \subseteq \pi(i)$---depending on whether $\pi(i-1) \,\relint{\face}\, \pi(i)$ or $\pi(i) \,\relint{\face}\, \pi(i-1)$.
	Concatenating these paths, we obtain a path $\pi' := \pi'_1;\dots;\pi'_k$ from $b_{\pi(0)}$ to $b_{\pi(k)} \in \sem{\psi}$ such that $\pi'((0,1)) \subseteq \pi(1) \cup \dots \cup \pi(k-1) \subseteq \sem{\phi}$.
	This shows that $\model{X}, b_{\pi(0)} \vDash \slreach(\phi,\psi)$, and so by Lemma~\ref{lemma:semanticsTriangulation} since $x,{b_{\pi(0)}} \in \relint{\sigma}$ we have $\model{X}, x \vDash \slreach(\phi,\psi)$.

	Conversely, assume that $\model{X}, x \vDash \slreach(\phi,\psi)$, which by Lemma~\ref{lem:piecewise-linear-semantics} amounts to the existence of a PL-path $\pi: [0,1] \toPL P$ such that $\pi(0) = x$, $\pi(1) \in \sem{\psi}$ and $\pi((0,1)) \in \sem{\phi}$.
	$\pi$ is a simplicial path by Lemma~\ref{lem:pl-are-simplicial}, so there exist points $s_0 = 0 < s_1 < \dots < s_l = 1$ such that each $\pi((s_i,s_{i+1}))$ is fully contained in a single cell.
	In particular, we can find cells $\relint{\tau}_1,\dots,\relint{\tau}_l,\relint{\tau}'_0, \dots, \relint{\tau}'_{l}$ such that $\pi((s_{j-1},s_{j})) \subseteq \relint{\tau}_{j}$ for $j = 1, \dots, l$, and such that $\pi(s_{j}) \in \relint{\tau}'_{j}$ for $j = 0, \dots, l$.
	Notice that $\relint{\tau}'_0 = \relint{\sigma}$, $\relint{\tau}'_{l} \subseteq \sem{\psi}$ (since $\pi(1) \in \relint{\tau}'_{l}$) and $\relint{\tau}_i, \relint{\tau}'_j \subseteq \sem{\phi}$ for $i = 1, \dots, l$ and $j = 1,\dots, l-1$ (since they contain points in $\sem{\phi}$).
	
	Observe that $\pi(s_j) \in \relint{\tau}'_{j} \cap \Cl(\relint{\tau}_{j+1})$ for $j = 0, \dots, l-1$, and so by Lemma~\ref{lemma:faceInTermsOfPartition} we have $\relint{\tau}'_{j} \,\relint{\face} \relint{\tau}_{j+1}$.
	With a similar argument, we also have $\relint{\tau}'_{j} \,\relint{\face}\, \relint{\tau}_{j}$ for $j = 1, \dots, l$.
	Rewriting the previous conditions we have $\relint{\tau}'_0 \,\relint{\face}\, \relint{\tau}_1 \,\relint{\facer}\, \relint{\tau}'_1 \,\relint{\face}\, \relint{\tau}_2 \,\relint{\facer}\, \dots \,\relint{\facer}\, \relint{\tau}'_l$, which by definition means that the sequence $\tuple{ \relint{\tau}'_0,\relint{\tau}_1, \relint{\tau}'_1, \relint{\tau}_2,\relint{\tau}'_2, \dots, \relint{\tau}_l,\relint{\tau}'_{l} }$ is a $\pm$-path of $\KM(\model{X})$.
	By previous considerations together with the inductive hypothesis applied to $\phi$ and $\psi$, we also have that $\relint{\tau}'_0 = \relint{\sigma}$, that $\relint{\tau}'_{l} \in \semKM{\psi}$ and that $\relint{\tau}_i, \relint{\tau}'_j \in \semKM{\phi}$ for $i = 1, \dots, l$ and $j = 1,\dots, l-1$.
	Thus we have $\KM(\model{X}), \relint{\sigma} \vDash \slreach(\phi,\psi)$, as desired.
\end{proof}

\noindent
{\bf Theorem~\ref{thm:logical-equivalence-is-bisimilarity}.}
{\em
Logical equivalence is a simplicial bisimulation.
}
\begin{proof}
    In the proof, for $S$ a set of logically equivalent points, we call ``characteristic formula'' of $S$ the characteristic formula of the equivalence class that includes $S$; similarly, we also speak of the ``characteristic formula'' of a point $x$.

    Consider two points $x$ and $y$, with $x \lequiv y$. Let us look at the conditions of Definition~\ref{def:simplicial-bisimulation}. First observe that, since $\phi$ can be an atomic proposition symbol, Condition~\ref{cond:bisim-atoms} holds. 
    We only prove Condition~\ref{cond:bisim-forth}, as the proof of Condition~\ref{cond:bisim-back} follows the same pattern.

	Equivalently, we shall prove, by induction on $k$, the following statement: for each $k \geq 1$, for each pair $x,y$ with $x\lequiv y$, for each path $\pi_x$ with $\pi_x(0) = x$, points $s_0, \ldots, s_k$, and cells $\widetilde{\sigma_1},\ldots,\widetilde{\sigma_k}$ making $\pi_x$ a simplicial path according to Definition~\ref{def:simplicial-path}, there is a simplicial path $\pi_y$ with $\pi_y(0) = y$ and $\pi_x \widehat\lequiv \pi_y$. 

    To ease readability, below, given the data above, we let $\phi_1,\dots,\phi_k$ be the characteristic formulas of the sets $\pi_x((s_0,s_1)), \dots, \pi_x((s_{k-1},s_k))$ respectively.
    Moreover, let $\phi'_0,\dots,\phi'_k$ be the characteristic formulas of the points $\pi_x(s_0),\dots,\pi_x(s_k)$ respectively.
    
    \noindent Next, the proof proceeds by induction on $k$. 

    \noindent For $k = 1$,
    observe that, by Definition~\ref{def:slcs-semantics}, we have $\model{X},x \vDash \slreach(\phi_1,\phi'_1)$. 
    By $x \lequiv y$, we have $\model{X},y \vDash \slreach(\phi_1,\phi'_1)$. 
    By Lemma~\ref{lem:piecewise-linear-semantics}, there is a \emph{piecewise linear} path $\pi_y$ with $\pi_y(0) = y$, $\pi_y((0,1)) \subseteq \sem{ \phi_1 }$, and $\pi_y(1) \in \sem{\phi'_1}$.
    By Lemma~\ref{lem:pl-are-simplicial}, $\pi_y$ is simplicial.
    Note that, since all the $\phi_i$ and $\phi'_i$ are characteristic formulas of equivalence classes, all the points in $\pi_x((0,1)) \cup \pi_y((0,1))$ are logically equivalent, and $\pi_x(1) \lequiv \pi_y(1)$. Therefore, $\pi_x \hat \lequiv \pi_y$.
    
    \noindent For $k > 1$, consider the sub-paths $\pi_x^a = \subpath{\pi_x}{[0,s_1]}$ and $\pi_x^b = \subpath{\pi_x}{[s_1,1]}$.
    By Lemma~\ref{lem:paths}, both sub-paths are simplicial.
    By the previous case, there is a simplicial path $\pi_y^a$ with $\pi_x^a \hat \lequiv \pi_y^a$.
    In particular, we have $\pi_x^a(1) \lequiv \pi_y^a(1)$.
    Noting that $\pi_x^a(1) = \pi_x(s_1) = \pi_x^b(0)$, we now apply the inductive hypothesis to the points $\pi_x^b(0)$, $\pi_y^a(1)$ and the simplicial path $\pi_x^b$, obtaining the simplicial path $\pi_y^b$ with $\pi_y^b(0) = \pi_y^a(1)$ and $\pi_x^b \hat \lequiv \pi_y^b$.
    Let $\pi_y = \pi_y^a;\pi_y^b$. By Lemma~\ref{lem:paths}, $\pi_y$ is simplicial, and we have  $\pi_x = \pi_x^a; \pi_x^b \hat \lequiv \pi_y^a;\pi_y^b = \pi_y$, proving the thesis. 
\end{proof}

\noindent
{\bf Theorem~\ref{thm:bisimilar-are-logically-equivalent}.}
{\em
Each simplicial bisimulation is included in logical equivalence.
}

\begin{proof}
    Given a simplicial bisimulation $\sim$, we need to show that for every formula $\phi$ and all points $x$ and $y$, if $x \sim y$ then $\model{X}, x \models \phi \iff \model{X}, y \models \phi$.
    We use induction on the structure of $\phi$.
    The cases for atomic propositions and Boolean operations are trivial, thus omitted.
    The case for the $\Box$ operator is also omitted, as $\Box$ is derived from $\slreach$ by Theorem~\ref{thm:near-derived-from-slreach}.
    
    \noindent
    Suppose $x \sim y$ and $\model{X},x \vDash \phi = \slreach(\phi_1,\phi_2)$.
    By Lemma~\ref{lem:piecewise-linear-semantics}, there is a piecewise linear path $\pi_x$ with $\pi_x(0) = x$, $\pi_x((0,1)) \subseteq \sem{\phi_1}$, and $\pi_x(1)\in \sem{\phi_2}$. 
    By Lemma~\ref{lem:pl-are-simplicial}, $\pi_x$ is simplicial. 
    By Condition~\ref{cond:bisim-forth} of Definition~\ref{def:simplicial-bisimulation}, and the fact that $x \sim y$, there is a simplicial path $\pi_y$ with $\pi_y(0) = y$, such that $\pi_x \hat \sim \pi_y$. 
    For each $r \in (0,1)$, we have $\pi_x(r) \sim \pi_y(r)$ and $\model{X},\pi_x(r) \models \phi_1$.
    By induction hypothesis, $\model{X},\pi_y(r) \vDash \phi_1$, therefore $\pi_y((0,1)) \subseteq \sem{ \phi_1 }$.
    We also have $\model{X},\pi_x(1) \vDash \phi_2$ and $\pi_x(1) \sim \pi_y(1)$. Thus, by induction hypothesis, $\pi_y(1) \in \sem{\phi_2}$.
    Summing up, we have $\model{X}, y \vDash \slreach(\phi_1,\phi_2)$, concluding one direction of our proof.
    Next, we should prove that conversely, if $\model{X},y \vDash \slreach(\phi_1,\phi_2)$ then $\model{X}, x \vDash \slreach(\phi_1,\phi_2)$. The argument is similar to the other case, using Condition~\ref{cond:bisim-back} of Definition~\ref{def:simplicial-bisimulation} instead of Condition~\ref{cond:bisim-forth}.
\end{proof}